\documentclass[aps,pra,groupedaddress,tightenlines,notitlepage,nofootinbib,onecolumn,longbibliography, superscriptaddress,11pt]{revtex4-2}
\usepackage{amsmath,amssymb,graphicx,amsmath,amssymb,amsthm,physics,bm,bbm,color,xcolor,mathtools,anyfontsize}%
\usepackage[sans]{dsfont}

\usepackage[colorlinks=true, citecolor=magenta, linkcolor=blue,            urlcolor=blue, linktocpage=true,hyperfootnotes=false]{hyperref}

\newcommand\numeq[1]%
  {\stackrel{\scriptscriptstyle(\mkern-1.5mu#1\mkern-1.5mu)}{=}}

\newcommand\numleq[1]%
  {\stackrel{\scriptscriptstyle(\mkern-1.5mu#1\mkern-1.5mu)}{\le}}

\usepackage{newpxtext,newpxmath,array}
\def\tr{\operatorname{tr}}
\def\id{\operatorname{id}}

\newcommand{\overbar}[1]{\mkern 2mu\overline{\mkern-2mu#1\mkern-2mu}\mkern 2mu}

\DeclarePairedDelimiterX{\channels}[2]{[}{]}{%
  #1\;\delimsize\|\;#2%
}
\DeclarePairedDelimiterX{\states}[2]{(}{)}{%
  #1\;\delimsize\|\;#2%
}
\usepackage{accents}
\newcommand{\ubar}[1]{\underaccent{\bar}{#1}}
\newcommand{\staD}{\mathbb{D}\states}
\newcommand{\chD}{\mathbb{D}\channels}

\DeclareOldFontCommand{\rm}{\normalfont\rmfamily}{\mathrm}

\DeclareMathOperator\supp{supp}

\newtheorem{theorem}{Theorem}
\newtheorem{proposition}{Proposition}%
\newtheorem{corollary}{Corollary}
\newtheorem{lemma}{Lemma}
\newtheorem*{example*}{Example}%
\newtheorem{remark}{Remark}%
\newtheorem{dfn}{Definition}%
\newtheorem{observation}{Observation}%

\begin{document}

\title{Conditional entropy and information of quantum processes}
\author{Siddhartha Das}
\email{das.seed@iiit.ac.in}
\affiliation{Centre for Quantum Science and Technology (CQST), Center for Security, Theory and Algorithmic Research (CSTAR), International Institute of Information Technology Hyderabad, Gachibowli 500032, Telangana, India}

\author{Kaumudibikash Goswami}\email{goswami.kaumudibikash@gmail.com}
\affiliation{QICI Quantum Information and Computation Initiative, Department of Computer Science, The University of Hong Kong, Pokfulam Road, Hong Kong}

\author{Vivek Pandey}
\email{vivekpandey3881@gmail.com}
\affiliation{Institute of Physics, Faculty of Physics, Astronomy and Informatics, Nicolaus Copernicus University, Grudziadzka 5/7, 87-100 Toru\'n, Poland
}
\affiliation{Harish-Chandra Research Institute, A CI of Homi Bhabha National Institute, Chhatnag Road, Jhunsi, Prayagraj 211019, India}

\begin{abstract}
What would be a reasonable definition of the conditional entropy of bipartite quantum processes, and what novel insight would it provide? We develop this notion using four information-theoretic axioms and define the corresponding quantitative formulas. Our definitions of the conditional entropies of channels are based on the generalized state and channel divergences, for instance, quantum relative entropy. We find that the conditional entropy of quantum channels has potential to reveal insights for quantum processes that aren't already captured by the existing entropic functions, entropy or conditional entropy, of the states and channels. The von Neumann conditional entropy $S[A|B]_{\mathcal{N}}$ of the channel $\mathcal{N}_{A'B'\to AB}$ is based on the quantum relative entropy, with system pairs $A',A$ and $B',B$ being nonconditioning and conditioning systems, respectively. We identify a connection between the underlying causal structure of a bipartite channel and its conditional entropy. In particular, we provide a necessary and sufficient condition for a bipartite quantum channel $\mathcal{N}_{A'B'\to AB}$ in terms of its von Neumann conditional entropy $S[A|B]_{\mathcal{N}}$, to have no causal influence from $A'$ to $B$. As a consequence, if $S[A|B]_{\mathcal{N}}< -\log|A|$ then the channel necessarily has causal influence (signaling) from $A'$ to $B$. Our definition of the conditional entropy establishes the strong subadditivity of the entropy for quantum channels. We also study the total amount of correlations possible due to quantum processes by defining the multipartite mutual information of quantum channels.    
\end{abstract}

\maketitle

\section{Introduction}
Quantum information theory may be considered a fateful marriage between quantum theory~\cite{Neu32}, one of the two most fundamental theories that describe nature, and the information theory~\cite{Sha48}. It allows us to inspect physical processes and the theories describing them from an information-theoretic perspective. Generalization of the information-theoretic tools or concepts like (conditional) entropies of joint probability distributions of random variables to quantum (conditional) entropies of the states of composite quantum systems has proven to be pivotal in both fundamental as well as applied quests~\cite{CD94,HT07,ADHW09,SW10,RAR+11,TBH14,CBTW17,DBWH21,BGG+24}. The entropy of a system quantifies the amount of uncertainty intrinsic to the system and the conditional entropy quantifies the amount of uncertainty present in the system when partial or side information is provided. Quantum conditional entropies for states are ubiquitous in studying quantum systems and processes. In a limiting case, we can reduce quantum entropic quantities and their relations to those for random variables and classical processes in (classical) information theory (cf.~\cite{FB14,TBH14}). However, quantum entropic quantities and their relations offer insights unique to quantum systems (states), i.e., not feasible for (classical) random variables. A quintessential example is the existence of quantum states with negative conditional (von Neumann) entropy~\cite{CA97} with operational meaning in the task of quantum state merging~\cite{HOW05,HOW06}. In contrast, conditional (Shannon) entropy is always nonnegative for joint probability distributions. Negative conditional entropy of bipartite quantum states is a signature of entanglement~\cite{CA97}: no separable states~\cite{Sch35} can have negative conditional entropy.

The physical transformations of quantum states are given by quantum channels. Mathematically, quantum channels are represented by completely positive, trace-preserving linear maps~\cite{Kra83,Sud85}. In practice, the notion of quantum channels naturally arises because it is impossible to have an isolated quantum system that is not interacting with its environment. Specifically, quantum channels are physical transformations on the systems arising from a joint unitary evolution, possibly governed by an interaction Hamiltonian, with an initially uncorrelated environment~\cite{Sti55}. Particularly for bipartite quantum channels, which is relevant in our case, let $U_{A'B'E'\to ABE}$ be a unitary operator (corresponding to Hamiltonian of underlying systems), $|A'||B'||E'|=|A||B||E|$. Without loss of generality, we assume the initial environment to be in the ground state $\op{0}_{E'}$. Then the bipartite quantum channel $\mathcal{N}_{A'B'\to AB}$ is given by the transformation $\tr_E\circ\mathcal{U}_{A'B'E'\to ABE}(\,\cdot\,\otimes\op{0}_{E'})$, where $\mathcal{U}(\cdot)=U(\cdot)U^{\dag}$. 

\begin{table}[t]
    \centering
    \begin{tabular}{|c|c|c|} \hline
        & entropy & conditional entropy \\ \hline
       random variables  & $ \textcolor{orange}{0\leq H(X)}\leq \log|X|$  & $\textcolor{orange}{0\leq H(X|Y)}\leq \log|X|$ \\ \hline
        quantum states & $\textcolor{orange}{0\leq S(A)_{\rho}} \leq \log|A|$ & $ \textcolor{red}{-\log|A|\leq S(A|B)_{\rho}}\leq \log|A|$\\ \hline
      quantum channels   & $\textcolor{red}{-\log|A|\leq S[A]_{\mathcal{N}}}\leq \log|A|$ ~\cite{GW21} & $\textcolor{blue}{-\log|A'||B'||B|\leq S[A|B]_{\mathcal{N}}}\leq \log|A|$ \\ \hline
    \end{tabular}
    \caption{We present a comparison between the dimensional bounds known for Shannon entropy $H(X)$ and conditional entropy $H(X|Y)$, von Neumann entropy $S(A)_{\rho}$ and conditional entropy $S(A|B)_{\rho}$ for quantum states, and von Neumann entropy $S[A]_{\mathcal{N}}$ and conditional entropy $S[A|B]_{\mathcal{N}}$ of point-to-point channels $\mathcal{N}_{A'\to A}$ and bipartite quantum channels $\mathcal{N}_{A'B'\to AB}$, respectively. Shannon entropies can be expressed as Kullback-Leibler (KL) divergence, and von Neumann entropies can be expressed as (Umegaki) relative entropy, a quantum generalization of the KL divergence. Dimensional bounds on the von Neumann conditional entropy of a bipartite channel are from this work, and the remaining are facts known in the literature. Upper bounds for all have the same form. Hence, we may focus on lower bounds for comparison. The lower bound of the von Neumann entropy of a quantum channel is reflected in the lower bound of the von Neumann conditional entropy of a quantum state. In contrast, the lower bound for the von Neumann conditional entropy of a channel stands out from the rest.}
    \label{tab:bounds}
\end{table}

Investigating quantum channels from information-theoretic perspective is crucial to understand the dynamics of quantum systems and its utility in technological applications. This motivates us to formulate conditional entropies of bipartite quantum channels. To achieve this, we take an axiomatic approach similar to the formulation of entropies of quantum channels in \cite{Gou19,GW21}. In the limiting scenarios, we can obtain conditional entropies of quantum states and quantum entropies of quantum channels~\cite{GW21} (see also~\cite{YHW19,Yua19,SPSD24}\footnote{\cite{Yua19} and \cite{GW21} independently provided definitions of the entropies of a quantum channel, see~\cite{GW21} for comparison between their definitions based on axiomatic approach of \cite{Gou19}.}). As conditional entropies require at least two random variables or systems, for one to be conditioning (role of side memory) and the other to be nonconditioning, conditional entropy for quantum channels is introduced here for bipartite quantum channels. Lingering questions while conceptualizing the notion of conditional entropies of bipartite quantum channels would naturally be (see Table~\ref{tab:bounds}):
\begin{itemize}
    \item How novel would the conditional entropy of quantum channels be when compared with the conditional entropy of states? Can the conditional entropy of bipartite quantum channels be strictly lower than the least possible values of entropy of quantum channels? If yes, what would that signature property for quantum channels be?
\end{itemize}

Axiomatic approaches have been employed to define entropies and conditional entropies of probability distributions and quantum states~\cite{Ren65,Ren05,MDSFT13,TBH14,BH14,WWY14}. Axioms form the minimal set of properties that would make entropy functions meaningful, which could be subject to context. Some of the well-known and widely used examples in classical information theory are Shannon entropy, min-entropy, max-entropy, and R\'enyi entropies, which have quantum counterparts for states, namely, von Neumann entropy, min-entropy, max-entropy, R\'enyi entropies, and sandwiched R\'enyi entropies (see Section~\ref{sec:gen-st-div}). There are different ways to generalize classical entropic functions to quantum entropic functions. Some of these families of entropy functions find operational meanings in information-processing tasks that enhance their utility and significance, e.g.,~\cite{Ren05,MO15,DBWH21,AMT24}.

\section{Main results}
Considering the motivations and ideas discussed earlier, we introduce axioms for the conditional entropies of bipartite quantum channels. Based on these axioms and using certain families of state and channel divergences, we formulate expressions for arguably a valid notion of the conditional entropies of bipartite quantum channels. To contrast with the previously developed notion of entropy of quantum channels, we show the (von Neumann) conditional entropy of a bipartite channel can be more negative than the dimensional lower bound on the entropy of a channel, see Table~\ref{tab:bounds}. In a way, what pure state is in the context of entropy and conditional entropy for states, we have unitary (or isometric) channels in the context of entropy and conditional entropy for channels, see Table~\ref{tab:conditionalentropy}. Similar to the way the negative conditional entropy of states certifies entanglement, the conditional entropy of bipartite channel, which is more negative than the negative logarithm of the nonconditioning-output-dimension certifies causal influence (signaling) from nonconditioning input to conditioning output. Further, the conditional entropy of a quantum channel directly finds operational meaning (up to a scaling factor depending on the dimension of the nonconditioning input system) in asymmetric hypothesis testing for channel discrimination depending on the generalized channel divergence it is based on. 

\begin{table}
\begin{center}\label{tab:conditionalentropy}
\begin{tabular}{ | m{8cm} | m{8cm} | } 
  \hline
 Conditional entropy of quantum states &  Conditional entropy of quantum channels \\ 
  \hline\hline
$S(A)_{\rho}$ achieves minimal value $0$ iff the state $\rho_A$ is \textcolor{red}{pure}~\cite{Neu32}. & $S[A]_{\mathcal{N}}$ achieves minimal value $-\log \min\{|A'|,|A|\}$ iff the channel $\mathcal{N}_{A'\to A}$ is \textcolor{red}{isometric}~\cite{GW21}.\\ 
  \hline
$S(A|B)_{\rho}<0$ $\implies$ state $\rho_{AB}$ is \textcolor{blue}{entangled}~\cite{CA97}. & $S[A|B]_{\mathcal{N}}<-\log|A|$ $\implies$ channel $\mathcal{N}_{A'B'\to AB}$ is \textcolor{blue}{signaling} from $A'\to B$. \\ 
  \hline
$S(A|B)<0$ for all pure entangled states $\rho_{AB}$~\cite{CA97,HOW05}. & $S[A|B]_{\mathcal{N}}<-\log|A|$ for all isometric channels $\mathcal{N}_{A'B'\to AB}$ signaling from $A'\to B$.\\
\hline
$S(A|BC)_{\rho}\leq S(A|B)_{\rho}$ for states $\rho_{AB}$~\cite{LR73}. & $S[A|BC]_{\mathcal{N}}\leq S[A|B]_{\mathcal{N}}$ for channels $\mathcal{N}_{A'B'\to AB}$.\\
  \hline 
\end{tabular}
\caption{We compare the features of the (von Neumann) conditional entropy of quantum states with that of quantum channels.}
\end{center}
\end{table}

\subsection{Axiomatic approach to conditional entropy of quantum processes}
Consider ${\rm Q}$ be the set of all quantum channels and ${\rm Q}(A',A)$ denote the set of channels $\mathcal{M}_{A'\to A}$. Given a valid entropy function ${f}[A]_\mathcal{M}$ of a quantum channel $\mathcal{M}_{A'\to A}$, i.e., the function, $f[\cdot]: \mathrm{Q}\to \mathbb{R}$, showing characteristics of entropy for a quantum channel, we demand that the conditional entropy function ${f}[A|B]_{\mathcal{N}}$ of a bipartite channel $\mathcal{N}_{A'B'\to AB}$, with $\mathbf{A}\equiv (A',A)$ and $\mathbf{B}\equiv (B',B)$, exhibits the following characteristics (P1 to P4):
\begin{itemize}
    \item[P1] \textit{Monotonicity}. It should be nondecreasing under the actions of local pre-processing and post-processing with any random unitary channel (probabilistic mixture of unitary channels) on $\mathbf{A}$ and any local pre-processing or post-processing on $\mathbf{B}$. 
     \item[P2] \textit{Conditioning should not increase entropy}\footnote{This does beg the question of how first to have a reasonable notion of a marginal of an arbitrary bipartite channel.}. $f[A|B]_{\mathcal{N}}\leq f[A]_{\mathcal{N}}$.
     \item[P3] \textit{Normalization}. For a replacer channel $\mathcal{N}$ that outputs $\rho_{AB}$ (irrespective of what the input state is), $f[A|B]_{\mathcal{N}}=f(A|B)_{\rho}$, where $f(A|B)_{\rho}$ is the conditional entropy function of $\rho_{AB}$ (entropy of $A$ conditioned on $B$).
     \item[P4] \textit{Reduction to the entropy of the channel on the nonconditioning system for a tensor-product channel.} For a tensor-product, bipartite channel $\mathcal{N}^1_{A'\to A}\otimes\mathcal{N}^2_{B'\to B}$, $f[A|B]_{\mathcal{N}^1\otimes\mathcal{N}^2}=f[A]_{\mathcal{N}^1}$.
\end{itemize}

Valid entropy functions of channels have been introduced and discussed in \cite{Gou19,GW21} (see also \cite{SPSD24}). For the above axiom (P2) to make sense, we need a well-meaning definition for the reduced channel, which is a marginal channel analogous to the marginal of states. For any bipartite channel $\mathcal{N}_{A'B'\to AB}$, a channel $\tr_B\circ\mathcal{N}_{A'B'\to AB}$ is its \textit{reduced channel} on $A$. A reduced channel defined like this is a valid, well-defined channel, as the composition of two channels is always a channel. A reduced channel plays the role of marginal fittingly to define the conditional entropy of the quantum channel and discuss the strong subadditivity of the channel's entropy.

\subsection{Formulation of conditional entropies of quantum channels}
We formulate the conditional entropies of quantum channels in terms of generalized channel divergence, which in turn are derived from the generalized state divergence. A generalized state divergence~\cite{SW12} is a function $\mathbb{D}(\cdot\Vert\cdot):{\rm St}(A')\times{\rm St}(A')\to \mathbb{R}$ between quantum states $\rho,\sigma\in\mathrm{St}(A')$ that is monotone (nonincreasing) under the action of quantum channels $\mathcal{N}_{A'\to A}$. Based on the state divergence, a divergence function between a channel $\mathcal{N}_{A'\to A}$ and a completely positive map $\mathcal{M}_{A'\to A}$ is be defined as~\cite{LKDW18,SPSD24} $\mathbb{D}[\mathcal{N}\Vert\mathcal{M}]:=\sup_{\rho_{RA'}}\mathbb{D}(\mathcal{N}(\rho_{RA'})\Vert\mathcal{M}(\rho_{RA'}))$, for states $\rho_{RA'}\in{\rm St}(RA')$. We make an observation that a divergence function between channels as defined earlier, is always monotone (nonincreasing) under the action of quantum superchannels, i.e., pre-processing and post-processing with quantum channels. In this sense, these divergence functions between channels are deemed as generalized channel divergence.

Let $\mathcal{R}_{A'\to A}$ denote the completely mixing (depolarizing) linear map that has replacer action, $\mathcal{R}(X_{A'})=\tr(X_{A'})\mathbbm{1}_A$ for all trace-class $X_{A'}$, and $\mathbbm{1}_A$ is the identity operator. The generalized entropy of a quantum channel $\mathcal{N}_{A'\to A}$ is defined as~\cite{GW21,SPSD24} $\mathbf{S}[A]_{\mathcal{N}}=-\mathbf{D}[\mathcal{N}_{A'\to A}\Vert\mathcal{R}_{A'\to A}]$, where $\mathbf{D}[\cdot\Vert\cdot]$ denotes a generalized channel divergence such that $\mathbf{S}[A]_{\mathcal{N}}$ is a valid entropy function of a channel, see for instance~\cite{GW21}.

We define the generalized conditional entropy of an arbitrary bipartite quantum channel $\mathcal{N}_{A'B'\to AB}$, finite- or infinite-dimensional, as
\begin{equation}\label{eq:con-ent-def}
    \mathbf{S}[A|B]_{\mathcal{N}}=- \inf_{\mathcal{M}\in{Q}(B',B)}\mathbf{D}[\mathcal{N}_{A'B'\to AB}\Vert\mathcal{R}_{A'\to A}\otimes\mathcal{M}_{B'\to B}],
\end{equation}
where $\mathbf{D}[\cdot\Vert\cdot]$ is a generalized channel divergence based on the generalized state divergence $\mathbf{D}(\cdot\Vert\cdot)$ that is additive and nonnegative for states. This definition satisfies our axioms (P1-P4) for conditional entropies of quantum channels.

The above definition can be extended to the generalized mutual information of an arbitrary bipartite quantum channel $\mathcal{N}_{A'B'\to AB}$, finite- or infinite-dimensional, as
\begin{equation}
      \mathbf{I}[A;B]_{\mathcal{N}}=- \inf_{\mathcal{M}^1\in{Q}(A',A),\mathcal{M}^2\in{Q}(B',B)}\mathbf{D}[\mathcal{N}_{A'B'\to AB}\Vert\mathcal{M}^1_{A'\to A}\otimes\mathcal{M}^2_{B'\to B}].
\end{equation}

Some suitable generalized channel divergence for our purpose are (Umegaki) relative entropy~\cite{Ume62}, min-relative entropy~\cite{Dat09}, max-relative entropy~\cite{Ren05,Dat09}, and subsets of families of Sandwiched R\'enyi and geometric R\'enyi relative entropies~\cite{MDSFT13,WWY14,FF21}. These generalized divergences are widely used because of their utility and operational meanings. If the associated generalized channel divergence has a semidefinite representation, then the conditional entropy and mutual information of a bipartite channel can be efficiently computed via semidefinite programming. The generalized conditional entropy and mutual information of a bipartite quantum channel naturally and directly find some operational meanings related to the operational meanings of the associated generalized channel and state divergence.

The generalized entropy exhibits conditional entropy-based strong subadditivity for an arbitrary tripartite channel $\mathcal{N}_{A'B'C'\to ABC}$,
\begin{equation}
    \mathbf{S}[A|B]_{\mathcal{N}}-\mathbf{S}[A|BC]_{\mathcal{N}}\geq 0,
\end{equation}
where $\mathcal{N}_{A'B'C'\to AB}=\tr_C\circ\mathcal{N}_{A'B'C'\to AB}$ is a reduced channel of $\mathcal{N}$. We also have mutual information-based strong subadditivity, 
\begin{equation}
    \mathbf{I}[A;BC]_{\mathcal{N}}-\mathbf{I}[A;B]_{\mathcal{N}}\geq 0.
\end{equation}
Interestingly, when the associated generalized divergence is taken to be relative entropy~\cite{Ume62}, we have ${S}[A|B]_{\mathcal{T}}-{S}[A|BC]_{\mathcal{T}}={I}[A;BC]_{\mathcal{T}}-{I}[A;B]_{\mathcal{T}}$ for tele-covariant tripartite channels~\cite[Definition 5]{DBWH21} $\mathcal{T}_{A'BC'\to ABC}$.

\subsection{von Neumann conditional entropy and signaling in bipartite quantum channel}
By considering the generalized channel (or state) divergence in Eq.~\eqref{eq:con-ent-def} to be Umegaki relative entropy, $D(\rho_A\Vert\sigma_A):=\lim_{\varepsilon\to 0^+}\tr(\rho_A(\log\rho_A-\log(\sigma_A+\varepsilon\mathbbm{1}_A)))$, we arrive at the von Neumann conditional entropy $S[A|B]_{\mathcal{N}}$ of a quantum channel $\mathcal{N}_{A'B'\to AB}$. In general,
\begin{equation}
    -\log|B|+S[AB]_{\mathcal{N}}\leq S[A|B]_{\mathcal{N}}\leq \inf_{\op{\psi}\in{\rm St}(RA'B')}S(A|RB)_{\mathcal{N}(\op{\psi})},
\end{equation}
where $S(A|B)_{\omega}:= - D(\rho_{AB}\Vert\mathbbm{1}_A\otimes\rho_B)$ is the von Neumann conditional entropy of a state $\omega_{AB}$ and $\omega_{B}:=\tr_A(\omega_{AB})$ denotes a reduced state of $\omega_{AB}$. The saturation of the r.h.s. holds, i.e.,
\begin{equation}
    S[A|B]_{\mathcal{N}}= \inf_{\op{\psi}\in{\rm St}(RA'B')}S(A|RB)_{\mathcal{N}(\op{\psi})}
\end{equation}
\textit{if and only if} the bipartite channels $\mathcal{N}_{A'B'\to AB}$ are semicausal, i.e., no-signaling from $A'\to B$~\cite{BGNP01,PHHH06}; whereas, for channels $\mathcal{N}_{A'B'\to AB}$ with causal influence, signaling from $A'\to B$, the strict inequality ``$<$" holds,
\begin{equation}
    S[A|B]_{\mathcal{N}}< \inf_{\op{\psi}\in{\rm St}(RA'B')}S(A|RB)_{\mathcal{N}(\op{\psi})}.
\end{equation}
The value $S[A|B]_{\mathcal{N}}<-\log |A|$ guarantees that the channel has causal influence from $A'\to B$. The minimum possible dimensional bound $-\log\min\{|A||B|^2,|A'||B'||B|\}$ is possible only for isometric channels that are signaling from $A'\to B$. For the two-qu$d$it unitary SWAP channel, the von Neumann conditional entropy is $-3\log d$, where $|A'|=|B'|=|A|=|B|=d$. 

For tele-covariant bipartite channels $\mathcal{T}_{A'B'\to AB}$~\cite{HHH99,CDP09,DBW17}, i.e., bipartite channels simulable via teleportation with the resource states being their Choi states $\Phi^{\mathcal{T}}_{R_AABR_B}:=\mathcal{T}(\Phi_{R_AR_B:A'B'})$, where $\Phi_{A:B}$ denotes a maximally entangled state of Schmidt rank $\min\{|A|,|B|\}$ and $R_A\simeq A', R_B\simeq B'$,
\begin{equation}
S[A|B]_{\mathcal{T}}=S(R_AA|R_BB)_{\Phi^{\mathcal{N}}}-\log|A'|.
\end{equation}
For a tele-covariant bipartite unitary channel $\mathcal{U}_{A'B'\to AB}$, we have $S[A|B]_{\mathcal{U}}+\log |A'|=-S(R_AA)_{\Phi^\mathcal{U}}=-S(R_BB)_{\Phi^{\mathcal{U}}}$.

If we consider the strong subadditivity of the von Neumann entropy for a tripartite quantum channel $\mathcal{N}_{A'B'C'\to ABC}$ in terms of the conditional entropy, then we get
\begin{equation*}
    S[A|B]_{\mathcal{N}}-S[A|BC]_{\mathcal{N}}\geq 0,
\end{equation*}
which is analogous to $0\leq S(A|B)_{\omega}-S(A|BC)_{\omega}:= I(A;C|B)_{\rho}$ for quantum states $\omega_{ABC}$, and $I(A;C|B)_{\omega}$ is the quantum conditional mutual information of a state $\omega_{ABC}$. In this sense, $I[A;C|B]_{\mathcal{N}}=S[A|B]_{\mathcal{N}}-S[A|BC]_{\mathcal{N}}$ can be deemed as the conditional mutual information of tripartite channel $\mathcal{N}_{A'B'C'\to ABC}$. For a tele-covariant tripartite channel $\mathcal{T}_{A'B'C'\to ABC}$, we get
\begin{equation*}
    I[A;C|B]_{\mathcal{T}}=I(R_AA;C|R_BBR_C)_{\Phi^{\mathcal{T}}}\leq I(R_AA;R_CC|R_BB)_{\Phi^{\mathcal{T}}},
\end{equation*}
with the Choi state $\Phi^{\mathcal{T}}_{R_AAR_BBR_CC}$, $R_X\simeq X'$ for $X\in\{A,B,C\}$.

\subsection{Outline of the paper}
In the latter sections of the paper, we discuss the methods and results in detail. The structure of this paper is as follows. In Section~\ref{prelims&notations}, we fix our notations and present the preliminaries and background necessary for the discussion of the results. In Section~\ref{sec:conditional_ent}, we provide the definition and formula for the conditional entropy of quantum channels and discuss some of its desirable properties. In Section~\ref{sec:family_of_conditional_entropies},
  we discuss properties and calculations of the conditional entropy of the channel based on relative entropy, called the von Neumann conditional entropy. We derive lower and upper bounds on the von Neumann conditional entropy for arbitrary channels and analytically and numerically calculate the conditional entropy for some relevant quantum channels. We provide expressions for the conditional entropies based on the min-, max-, and geometric R\'enyi relative entropies. We numerically evaluate the geometric R\'enyi conditional entropy and the conditional min-entropy of some channels of practical interest. In Section~\ref{sec:signaling_vs_semicausal}, we discuss the connection between the presence of causal influence in bipartite processes and their conditional entropies. We find expression for the conditional entropy for semicausal (nosignaling from nonconditioning input system to conditioning output system) quantum channels in terms of the minimal conditional entropy of the output state. We provide a lower bound on conditional entropy for $k$-extendible tele-covariant quantum channels. In Section~\ref{sec:operational_meaning}, we discuss the operational interpretation of the conditional entropy of quantum channels in terms of asymmetric quantum hypothesis testing for quantum channel discrimination in Stein's setting and for communication task called conditional quantum channel merging, a variant of quantum state merging protocol. In Section~\ref{sec:strong_subad_QCH}, we study the conditional entropy-based strong subadditivity property for quantum channels and define the conditional mutual information for quantum channels. Analogous to quantum Markov states, we define quantum Markov channels and study the properties of approximate quantum Markov channels in terms of their recoverability when nonconditioning systems become partially inaccessible or lost. In Section~\ref{sec:mutual_information_CH}, we define multipartite mutual information for multipartite quantum channels and derive the expression for tele-covariant quantum channels in terms of their Choi states. In Section~\ref{sec:bidirectional_CH}, we have extended our definition of conditional entropy for bidirectional quantum channels, where two distant parties hold each input and output pair and discuss their properties. We conclude by providing brief remarks and some of the future directions from this work in Section~\ref{sec:discussion}.

\tableofcontents
\section{Preliminaries} \label{prelims&notations}
A composite quantum system $AB$ is associated with a separable tensor-product Hilbert space $\mathcal{H}_{AB}:=\mathcal{H}_A\otimes\mathcal{H}_B$, where $\mathcal{H}_A$ denotes the Hilbert space of $A$, $|A|:=\dim{(\mathcal{H}_A)}$, and in general $|A|\leq +\infty$. Going forward we may use $A$ to denote the Hilbert space of system $A$. $\mathrm{L}(A)$, $\mathrm{Pos}(A)$, $\mathrm{St}(A)$ are the set of all linear bounded operators, positive semidefinite and bounded operators, and density operators (states), respectively, defined on the Hilbert space $A$. A linear bounded operator $\rho_A$ denotes $\rho\in\mathrm{L}(A)$. Quantum states are given by density operators that are positive semidefinite operators with unit trace. Pure quantum states are rank-one density operators. A bipartite state $\rho_{AB}$ is called separable between $A:B$ if it can be written as a convex combination of product states, $\rho_{AB}=\sum_{x}p(x)\omega^x_A\otimes\sigma^x_B$ for $\sum_xp(x)=1$, $p(x)\in[0,1]$ and $\omega^x\otimes\sigma^x\in{\rm St}(AB)$ for all $x$. If a bipartite state $\rho_{AB}$ is not separable, then it is called entangled. Let $\Gamma_{RA'}:=\op{\Gamma}_{RA'}$ denote a maximally entangled operator, i.e., $\ket{\Gamma}_{RA'}:=\sum_{i=0}^{\min\{|A'|,|R|\}-1}\ket{i,i}_{RA'}$, where $\{\ket{i}\}_i$ forms an orthonormal set of vectors. We know  $\Gamma_{R_AR_B:A'B'}=\Gamma_{R_AA'}\otimes\Gamma_{R_BB'}$.

Let ${\rm CP}(A',A)$ denote the set of all completely positive maps $\mathcal{M}_{A'\to A}$ that outputs bounded operators $\sigma\in{\rm Pos}(A)$ if input is a bounded operator $\rho\in{\rm Pos}(A)$. Let $\mathrm{Q}(A',A)$ denote the set of all quantum channels with input system $A'$ and output system $A$. A quantum channel $\mathcal{N}_{A'\to A}:{\rm St}(A')\to {\rm St}(A)$ is a completely positive, trace-preserving map and $\mathcal{N}_{A'\to A}\in{\rm Q}(A',A)$. $\id_{A}$ denotes the identity channel acting on $X\in{\rm L}(A)$. $\id_R\otimes\mathcal{N}_{A'\to A}$ may simply be denoted as $\mathcal{N}_{A'\to A}$ whenever context is clear. The Choi operator $\Gamma^\mathcal{N}$ of a completely positive map is $\mathcal{N}_{A'\to A}(\Gamma_{RA'})$, meaning $\id_R\otimes\mathcal{N}_{A'\to A}(\Gamma_{RA'})$ where $R\simeq A'$, i.e., $|R|=|A'|$. The Choi state $\Phi^{\mathcal{N}}$ of a quantum channel $\mathcal{N}_{A'\to A}$, when $|A'|,|A|<+\infty$, is its normalized Choi operator, i.e., $\Phi^{\mathcal{N}}_{RA}:=\frac{1}{|A'|}\Gamma^{\mathcal{N}}_{RA}$. In general, for a linear positive map $\mathcal{M}_{A'\to A}$, $\Phi^{\mathcal{M}}_{RA}:=\mathcal{M}(\Phi_{RA'})$. Let $\pi_A:=\frac{1}{|A|}\mathbbm{1}_A$ denote the maximally mixed state for $|A|<+\infty$. $\tr$ denotes the trace operation and $\tr_A$ denotes the partial trace of a subsystem $A$ over a composite system; for $\tr(X_A)$ to be finite, we need $X\in{\rm L}(A)$ to be a trace-class operator. $\tr(\mathbbm{1}_{A})=|A|$, which is infinite if $|A|=+\infty$. $\log$ denotes the logarithm with base $2$. $\norm{X}_{p}$ for $p\in[1,\infty]$ denotes the Schatten $p$-norm, $\norm{X}_{p}:=\tr[|X|^{p}]^{1/p}$, $|X|=\sqrt{X^\dag X}$ for $X\in\mathrm{L}(A')$. The diamond norm of a completely positive map $\mathcal{N}_{A'\to A}$ is $\norm{\mathcal{N}}_{\diamond}:=\sup_{\rho\in\mathrm{St}(RA')}\norm{\id_R\otimes\mathcal{N}(\rho_{RA'})}_1$, where it suffices to optimize over pure states with $|R|=|A'|$. The average energy of a system $A'$ in the state $\rho_{A'}$ is $\tr[\widehat{H}_{A'}\rho_{A'}]$, where $\widehat{H}_{A'}$ denotes the Hamiltonian of $A'$. An isomteric channel $\mathcal{V}_{A'\to A}$ is a quantum channel with a single Kraus operator, $\mathcal{V}_{A'\to A}(\cdot)=V(\cdot)V^{\dag}$, where $V_{A'\to A}$ is an isomtery operator, i.e., $V^\dag V=\mathbbm{1}_{A'}$ and $VV^{\dag}=\Pi_{A}$ for some projector $\Pi_{A}$ on the support of $A$. An unitary channel $\mathcal{U}_{A'\to A}$ is an isomteric channel with $|A'|=|A|$. Let $\mathrm{SC}((A',A),(C',C))$ denote the set of all superchannels that map $\mathrm{Q}(A',A)$ to $\mathrm{Q}(C',C)$. The action of a superchannel $\Theta\in\mathrm{SC}((A',A),(C',C))$ on any quantum channel $\mathcal{N}_{A \rightarrow A'}$ can be written as a composition of quantum channel as follows~\cite{CDP09}:
  \begin{equation}
      \Theta \left( \mathcal{N}_{A' \rightarrow A} \right) = \mathcal{M}^{\rm{post}}_{AR \rightarrow C} \circ \mathcal{N}_{A' \rightarrow A}  \circ \mathcal{M}^{\rm{pre}}_{C' \rightarrow A'R}, \label{equ:superchannel_action}
  \end{equation}
  where $\mathcal{M}^{\rm post}_{AR \rightarrow C}$ and $\mathcal{M}^{\rm pre}_{C' \rightarrow A'R}$ are postprocessing and preprocessing channels, respectively. We will call a superchannel $ \Theta$ as isomteric superchannel if $\mathcal{M}^{\rm{post}}$ and $ \mathcal{M}^{\rm{pre}}$ are isomteric channels.

\subsection{Generalized state divergence}
\begin{dfn}[\cite{SW12}]
A generalized state divergence $\mathbb{D}(\cdot\Vert\cdot)$ is a function between two positive semidefinite operators defined on the same Hilbert space that is monotonically nonincreasing for the states under the action of an arbitrary quantum channel. That is, generalized state divergence $\mathbb{D}(\rho\Vert\sigma)$ for any $\rho\in\mathrm{St}(A')$ and $\sigma\in \mathrm{St}(A')$ and an arbitrary quantum channel $\mathcal{N}_{A'\to A}$ satisfies
\begin{equation}
    \mathbb{D}(\rho\Vert\sigma)\geq \mathbb{D}(\mathcal{N}(\rho)\Vert\mathcal{N}(\sigma)).
\end{equation}
\end{dfn}
A generalized state divergence $\mathbb{D}(\cdot\Vert\cdot)$ between arbitrary two states $\rho,\sigma\in \mathrm{St}(A)$ always satisfies the following properties~\cite{WWY14}:
\begin{align}
    \mathbb{D}(\rho\Vert\sigma) &=\mathbb{D}(\rho\otimes\omega\Vert\sigma\otimes\omega),\\
    \mathbb{D}(\rho\Vert\sigma) & =\mathbb{D}(\mathcal{V}(\rho)\Vert\mathcal{V}(\sigma)),
\end{align}
where $\omega\in\mathrm{St}(B)$ is some arbitrary state and $\mathcal{V}_{A\to E}$ is an arbitrary isomteric channel. Some generalized states divergences may have some additional properties:
\begin{enumerate}
    \item[S1] \textit{Faithfulness}. A generalized state divergence is called faithful if $\mathbb{D}(\rho\Vert\sigma)\geq 0$ for all states $\rho,\sigma\in\mathrm{St}(A)$ and $\mathbb{D}(\rho\Vert\sigma)= 0$ if and only if $\rho=\sigma$.
    \item[S2] \textit{(Sub)additivity}. A generalized state divergence is called subadditive if for arbitrary states $\rho_A,\sigma_A,\omega_B,\tau_B$ (additive when equality always holds),
    \begin{equation}
        \mathbb{D}(\rho\otimes\omega\Vert\sigma\otimes\tau)\geq \mathbb{D}(\rho\Vert\sigma)+  \mathbb{D}(\omega\Vert\tau).
    \end{equation}
\end{enumerate}

\subsubsection{Examples of generalized state divergence}\label{sec:gen-st-div}
Some of the specific generalized divergences with wide applications and use are the families of sandwiched R\'enyi relative entropy, quantum (Umegaki) relative entropy, $\varepsilon$-hypothesis testing relative entropy, negative of (Uhlmann) fidelity, purified distance, geometric R\'enyi divergence, trace-distance, etc.

The trace-distance between $\rho,\sigma\in\mathrm{Pos}(A')$ is $ \norm{\rho-\sigma}_1$, where $\norm{X}_1:=\tr|X|$ for $X\in\mathrm{L}(A')$. The trace-distance between quantum states is monotonically nonincreasing under the action of positive trace-preserving maps. The (Uhlmann) fidelity between $\rho,\sigma\in\mathrm{Pos}(A')$ is $F(\rho,\sigma)=\norm{\sqrt{\rho}\sqrt{\sigma}}^2_1$ ($\sqrt{F(\rho,\sigma)}=\tr\{\sqrt{\rho\sigma}\}$~\cite{BJ23,SBM24}). The fidelity between quantum states is monotonically nondecreasing under the action of positive trace-preserving maps. The purified distance $P(\rho,\sigma):=\sqrt{1-F(\rho,\sigma)}$ is monotonically nonincreasing between quantum states under the action of positive trace-preserving maps. 

The quantum (Umegaki) relative entropy between $\rho\in\mathrm{St}(A')$ and $\sigma\in\mathrm{Pos}(A')$ is defined as~\cite{Ume62} 
\begin{equation}
D(\rho\Vert \sigma)\coloneqq  \left\{ 
\begin{tabular}{c c}
$\tr\{\rho[\log\rho-\log\sigma]\}$, & $\supp(\rho)\subseteq\supp(\sigma)$\\
$+\infty$, &  otherwise.
\end{tabular} 
\right.
\end{equation}
The quantum relative entropy is monotonically nonincreasing under the action of positive trace-preserving maps \cite{MR15}, that is $D(\rho\Vert\sigma)\geq D(\mathcal{M}(\rho)\Vert\mathcal{M}{(\sigma)})$ for any $\rho\in\mathrm{St}(A')$ and $\sigma\in\mathrm{St}(A')$ and a positive trace-preserving map $\mathcal{M}_{A'\to A}$.

The sandwiched R\'enyi relative entropy $\widetilde{D}_\alpha(\cdot\Vert\cdot)$ \cite{MDSFT13,WWY14} is defined between any
$\rho\in\mathrm{St}(A')$ and $\sigma\in\mathrm{Pos}(A')$, for $ \alpha\in (0,1)\cup(1,\infty)$, as
\begin{align}
\widetilde{D}_\alpha(\rho\Vert \sigma)\coloneqq  \frac{1}{\alpha-1}\log \tr\left\{\left(\sigma^{\frac{1-\alpha}{2\alpha}}\rho\sigma^{\frac{1-\alpha}{2\alpha}}\right)^\alpha \right\}.\label{eq:def_sre}
\end{align}
It is set to $+\infty$ for $\alpha\in(1,\infty)$ if $\supp(\rho)\nsubseteq \supp(\sigma)$. The sandwiched R\'enyi relative entropy is monotonically nondecreasing in $\alpha$ \cite{MDSFT13}:
\begin{equation}\label{eq:mono_sre}
\widetilde{D}_\alpha(\rho\Vert\sigma)\leq \widetilde{D}_\beta(\rho\Vert\sigma) ,
\end{equation}
 if   $\alpha\leq \beta$,  for  $\alpha,\beta\in(0,1)\cup(1,\infty)$.

For certain range of $\alpha$, the sandwiched R\'enyi relative entropy $\widetilde{D}_\alpha(\rho\Vert\sigma)$ is a generalized state divergence~\cite{MDSFT13,WWY14}: Let $\mathcal{N}\in\mathrm{Q}(A',A)$ and let $\rho,\sigma\in\mathrm{St}(A')$, then for all $ \alpha\in [1/2,1)\cup (1,\infty)$,
\begin{equation}
\widetilde{D}_\alpha(\rho\Vert\sigma)\geq \widetilde{D}_\alpha(\mathcal{N}(\rho)\Vert\mathcal{N}(\sigma)).
\end{equation} 
In the limit $\alpha\to 1$, the sandwiched R\'enyi relative entropy $\widetilde{D}_\alpha(\rho\Vert\sigma)$ converges to the quantum relative entropy $D(\rho\Vert\sigma)$ \cite{MDSFT13,WWY14}. In the limit $\alpha\to \infty$, the sandwiched R\'enyi relative entropy $\widetilde{D}_\alpha(\rho\Vert\sigma)$ converges to the max-relative entropy~\cite{MDSFT13}, which is defined as~\cite{Ren05,Dat09}
\begin{equation}\label{eq:max-rel}
D_{\max}(\rho\Vert\sigma) \coloneqq \log \inf  \{\lambda:\ \rho \leq \lambda\sigma\}=\log\norm{\sigma^{-1/2}\rho\sigma^{-1/2}}_{\infty}, 
\end{equation}
with $D_{\max}(\rho\Vert\sigma)=+\infty$ if $\supp(\rho)\nsubseteq\supp(\sigma)$.

The $\varepsilon$-hypothesis-testing relative entropy between $\rho\in\mathrm{St}(A')$ and $\sigma\in\mathrm{Pos}(A')$ is defined as~\cite{BD10,WR12}
\begin{align}
D^\varepsilon_h(\rho\Vert\sigma)\coloneqq -\log\inf_{\Lambda}\{\tr\{\Lambda\sigma\}:  0\leq\Lambda\leq \mathbbm{1} \wedge\tr\{\Lambda\rho\}\geq 1-\varepsilon\}\quad \text{for}\, \varepsilon\in[0,1].\label{eq:hypo-test-div}
\end{align}
In the limit $\varepsilon=0$, $D^{\varepsilon=0}_h(\rho\Vert\sigma)$ is equal to the min-relative entropy~\cite{Dat09}
\begin{equation}
    D_{\min}(\rho\Vert\sigma) \coloneqq \left\{ 
\begin{tabular}{c c}
$-\log \tr (\sigma \Pi_{\rho})$, & $ \tr[\rho \sigma] \neq 0$\\
$+\infty$, &  otherwise,
\end{tabular} 
\right.
\end{equation}
where $\Pi_{\rho}$ denote projection onto the support of $\sigma$. 

\subsection{Generalized channel function and divergence}
A generalized channel function $\mathbb{D}[\cdot\Vert\cdot]$ between any two completely positive maps, quantum channel $\mathcal{N}_{A'\to A}$ and a completely positive map $\mathcal{M}_{A'\to A}\in{\rm CP}(A',A)$ can be defined as~\cite{LKDW18,SPSD24}
    \begin{equation}\label{eq:gen-ch-f}
        \mathbb{D}[\mathcal{N}\Vert\mathcal{M}]\coloneqq \sup_{\rho\in\mathrm{St}(RA')}\mathbb{D}(\id_{R}\otimes\mathcal{N}(\rho_{RA'})\Vert\id_R\otimes\mathcal{M}(\rho_{RA'})),
    \end{equation}
    where $\mathbb{D}(\cdot\Vert\cdot)$ is a generalized state divergence. It suffices to restrict $\rho_{RA'}$ in \eqref{eq:gen-ch-f} to pure states and $R\simeq A'$. The above definition is generalization of \cite[Definition II.2]{LKDW18} (see \cite[Definition 3]{SPSD24}).

If the input states to the channel are accessible only from a set $\mathsf{Access}(RA')\subset \mathrm{St}(RA')$ of states because of some physical constraints, then we define the resource-constraint generalized channel function between a quantum channel $\mathcal{N}_{A'\to A}$ and a completely positive map $\mathcal{M}\in{\rm CP}(A',A)$ as (cf.~\cite[Eq.~(67)]{SPSD24})
\begin{equation}\label{eq:gen-ch-f2}
        \mathbb{D}_{\mathsf{Access}}[\mathcal{N}\Vert\mathcal{M}]\coloneqq \sup_{\rho_{RA'}\in\mathsf{Access}(RA')}\mathbb{D}(\id_{R}\otimes\mathcal{N}(\rho_{RA'})\Vert\id_R\otimes\mathcal{M}(\rho_{RA'})),
    \end{equation}
and hence, \begin{equation}
    \mathbb{D}_{\mathsf{Access}}[\mathcal{N}\Vert\mathcal{M}]\leq \mathbb{D}[\mathcal{N}\Vert\mathcal{M}].
\end{equation}
When $\mathrm{Access}(RA')=\mathrm{St}(RA')$, then $\mathbb{D}_{\mathrm{Access}}[\mathcal{N}\Vert\mathcal{M}]=\mathbb{D}[\mathcal{N}\Vert\mathcal{M}]$. 

\begin{lemma}\label{lem:gen-ch-order}
    Let $\mathbb{F}(\cdot\Vert\cdot)$ and $\mathbb{G}(\cdot\Vert\cdot)$ be two generalized state divergences such that for arbitrary pair of (valid) states $\rho,\sigma$, $\mathbb{F}(\rho\Vert\sigma)\leq \mathbb{G}(\rho\Vert\sigma)$. Then, the generalized channel functions $\mathbb{F}[\cdot\Vert\cdot]$ and $\mathbb{G}[\cdot\Vert\cdot]$ also obey the same order for arbitrary pair of (valid) channels $\mathcal{N}, \mathcal{M}$, i.e., 
    \begin{align}
        \mathbb{F}[\mathcal{N}\Vert\mathcal{M}]\leq \mathbb{G}[\mathcal{N}\Vert\mathcal{M}], \qquad
        \mathbb{F}_{\rm Access}[\mathcal{N}\Vert\mathcal{M}]\leq \mathbb{G}_{\rm Access}[\mathcal{N}\Vert\mathcal{M}].
    \end{align}
\end{lemma}
\begin{proof}
     Consider a pair of states $\rho_{RA}, \sigma_{RA}$, and a pair of quantum channels $\mathcal{N}_{A \rightarrow A},\mathcal{M}_{A \rightarrow A}$. It is given that for any pair of states $\widetilde{\rho}_{RA}$ and $\widetilde{\sigma}_{RA}$, we have 
    $\mathbb{F}(\widetilde{\rho}_{RA}\Vert\widetilde{\sigma}_{RA})\leq\mathbb{G}(\widetilde{\rho}_{RA}\Vert\widetilde{\sigma}_{RA})$. Let us define $\mathrm{id}_{R} \otimes \mathcal{N}_{A \rightarrow A}(\rho_{RA}):=\widetilde{\rho}_{RA}$ and $\mathrm{id}_{R} \otimes \mathcal{M}_{A \rightarrow A}(\rho_{RA}):=\widetilde{\sigma}_{RA}$. We then have 
    \begin{equation}
        \mathbb{F}(\mathrm{id}_{R} \otimes \mathcal{N}_{A \rightarrow A}(\rho_{RA}) \Vert \mathrm{id}_{R} \otimes \mathcal{M}_{A \rightarrow A}(\rho_{RA}))\leq\mathbb{G}(\mathrm{id}_{R} \otimes \mathcal{N}_{A \rightarrow A}(\rho_{RA})\Vert\mathrm{id}_{R} \otimes \mathcal{M}_{A \rightarrow A}(\rho_{RA})). \label{channel_inequality}
    \end{equation}
    If we take the supremum over all the states $\rho_{RA}$ in the l.h.s. of the above inequality, then the l.h.s. will become $\mathbb{F}[\mathcal{N}\Vert\mathcal{M}]$. Let $\psi_{RA}$ be the state realizing the supremum of $\mathbb{F}[\mathcal{N}\Vert\mathcal{M}]$. Now let us take the supremum over all the states $\rho_{RA}$ in the r.h.s. of the above inequality. Clearly $\psi_{RA}$ is also a valid choice of $\rho_{RA}$ to calculate the supremum of $\mathbb{G}(\cdot\Vert\cdot)$ and we have the following inequality $\mathbb{F}(\mathrm{id}_{R} \otimes \mathcal{N}_{A \rightarrow A}(\psi_{RA}) \Vert \mathrm{id}_{R} \otimes \mathcal{M}_{A \rightarrow A}(\psi_{RA})) = \mathbb{F}[\mathcal{N}\Vert\mathcal{M}] \leq\mathbb{G}(\mathrm{id}_{R} \otimes \mathcal{N}_{A \rightarrow A}(\psi_{RA})\Vert\mathrm{id}_{R} \otimes \mathcal{M}_{A \rightarrow A}(\psi_{RA})) \leq \mathbb{G}[\mathcal{N}\Vert\mathcal{M}]$, which proves our claim. The similar argument can be used to prove $ \mathbb{F}_{\rm Access}[\mathcal{N}\Vert\mathcal{M}]\leq \mathbb{G}_{\rm Access}[\mathcal{N}\Vert\mathcal{M}]$.
\end{proof}

\begin{dfn}[Bidirectional generalized channel function]\label{def:bi-gen-ch}
    A bidirectional generalized channel function $\mathbb{D}^{\leftrightarrow}[\cdot\Vert\cdot]$ between two bipartite completely positive maps $\mathcal{N},\mathcal{M}\in{\rm CP}(A'B',AB)$ is given by
    \begin{equation}
        \mathbb{D}^{\leftrightarrow}[\mathcal{N}\Vert\mathcal{M}]:=\sup_{\rho\in\mathrm{St}(R_AA'),\omega\in\mathrm{St}(R_BB')}\mathbb{D}(\mathcal{N}(\rho_{R_AA'}\otimes\omega_{R_BB'})\Vert\mathcal{M}(\rho_{R_AA'}\otimes\omega_{R_BB'})),
    \end{equation}
    where it suffices to optimize over pure product states $\rho_{R_AA'}\otimes\omega_{R_BB'}$.
\end{dfn}

Before we proceed further, we make an important observation regarding the monotonicity of the generalized channel function. The proof argument of the theorem below is directly lifted from \cite{Yua19,SPSD24}. See detailed proof in Appendix~\ref{prop:monotonicity_of_channeldivergence_proof} for the sake of completeness.
\begin{theorem}\label{chgen_div_mon}
 The generalized channel function $ \mathbb{D}\left[\mathcal{N}\Vert \mathcal{M}\right]$ between an arbitrary pair of completely positive maps, $\mathcal{N}\in{\rm Q}(A',A)$ and $\mathcal{M}\in{\rm CP}(A',A)$, is nonincreasing under the action of a quantum superchannel $\Theta \in SC ((A',A),(C',C))$, i.e.,
    \begin{equation}
      \mathbb{D}\left[\mathcal{N}\Vert \mathcal{R}\right]\geq    \mathbb{D}\left[\Theta\left(\mathcal{N}\right) \Vert \Theta \left(\mathcal{R}\right)\right],
    \end{equation}
where $ \mathbb{D}(\cdot \Vert \cdot)$ is a generalized state divergence. That is, a generalized channel function is a generalized channel divergence as it is monotone under the action of quantum superchannels.
\end{theorem}

\begin{proposition}
    The generalized channel divergence between arbitrary quantum channels $\mathcal{N}_{A'\to A},\mathcal{M}_{A'\to A}$ satisfy following properties:
\begin{enumerate}
    \item[C1] Invariance under the action of unitary superchannels, i.e., pre- and post-processing with unitary channels,
    \begin{equation}
        \mathbb{D}[\mathcal{N}\Vert\mathcal{M}]=\mathbb{D}[\mathcal{U}^{\rm post}\circ\mathcal{N}\circ\mathcal{U}^{\rm pre}\Vert\mathcal{U}^{\rm post}\circ\mathcal{M}\circ\mathcal{U}^{\rm pre}],\label{eq:C1}
    \end{equation}
    for any unitary channels $\mathcal{U}^{\rm pre}_{A"\to A'E'}, \mathcal{U}^{\rm post}_{AE'\to B}$.
    \item[C2] Invariance under appending of arbitrary quantum channels,
    \begin{equation}
          \mathbb{D}[\mathcal{N}\Vert\mathcal{M}]=  \mathbb{D}[\mathcal{N}\otimes\mathcal{T}\Vert\mathcal{M}\otimes\mathcal{T}],\label{eq:C2}
    \end{equation}
    where $\mathcal{T}_{B'\to B}$ is an arbitrary quantum channel. 
\end{enumerate}
\end{proposition}
\begin{proof}
    The first property (C1) follows by twice applying the monotonicity of the generalized channel divergence under the action of unitary channels.
    To prove the second property (C2), we first note that appending a quantum channel is an action of a superchannel, $\mathcal{N}_{A'\to A}\otimes\mathcal{T}_{B'\to B}= \mathcal{N}_{A'\to A}\circ\mathcal{T}_{B'\to B}=\mathcal{T}_{B'\to B}\circ\mathcal{N}_{A'\to A}$, and as composition with partial trace is also a superchannel, then
    \begin{align}
    \mathbb{D}[\mathcal{N}\Vert\mathcal{M}] & \geq \mathbb{D}[\mathcal{N}\otimes\mathcal{T}\Vert\mathcal{M}\otimes\mathcal{T}]\nonumber\\
    & \geq \mathbb{D}[\tr_{B}\circ\mathcal{N}\otimes\mathcal{T}\Vert\tr_{B}\circ \mathcal{M}\otimes\mathcal{T}]\nonumber\\
    & = \sup_{\rho_{RA'B'}\in\mathrm{St}(RA'B')}\mathbb{D}(\tr_{B}\circ\mathcal{N}\otimes\mathcal{T}(\rho_{RA'B'})\Vert\tr_{B}\circ \mathcal{M}\otimes\mathcal{T}(\rho_{RA'B'}))\nonumber\\
    & \geq \sup_{\rho_{RA'}\in\mathrm{St}(RA')}\mathbb{D}(\tr_B\circ\mathcal{N}\otimes\mathcal{T}(\rho_{RA'}\otimes\sigma_{B'})\Vert\tr_B\circ\mathcal{M}\otimes\mathcal{T}(\rho_{RA'}\otimes\sigma_{B'}))\nonumber \\
    & = \mathbb{D}[\mathcal{N}\Vert\mathcal{M}].
    \end{align}
\end{proof}

\section{Conditional quantum entropies: Definitions and Characteristics} \label{sec:conditional_ent}
\subsection{Recipe for deriving information measures of quantum channels}
In this section, we lay down the recipe for generalizing entropic and information measures of quantum states to quantum channels, with the primary focus being conditional entropy. We first recall the definition of generalized entropy of a quantum channel $\mathcal{N}_{A'\to A}$~\cite{Gou19,GW21}. To do so, we now define the completely mixing map, also called a completely depolarizing map.
\begin{dfn}[\cite{SPSD24}] \label{def_completely_depolarising_map}
The completely mixing map $\mathcal{R}_{A' \rightarrow A}$ is a linear completely positive map which replaces input trace-class operators with a scalar times the identity operator; in particular, its action on the trace-class operators $X\in\mathrm{L}(A')$ is
\begin{equation}\label{equ:completely_depolarizing_map}
        \mathcal{R}_{A' \rightarrow A}(X):= \mathrm{tr}(X)\mathbbm{1}_{A}.
    \end{equation}
    The Choi operator of the completely mixing map $\mathcal{R}_{A' \rightarrow A}$ is $\Gamma_{RA}^\mathcal{R}=\mathbbm{1}_{R}\otimes\mathbbm{1}_{A}$, where $R\simeq A'$.
\end{dfn}

\begin{dfn}[Generalized entropy~\cite{Gou19,GW21}]
    The generalized entropy $\mathbf{S}[\mathcal{N}]$ of a point-to-point channel $\mathcal{N}_{A'\to A}$ is defined as the generalized channel divergence between the quantum channel $\mathcal{N}\in{\rm Q}(A',A)$ and the completely mixing map $\mathcal{R}\in {\rm CP}(A',A)$~\cite{GW21,SPSD24},
    \begin{equation}
        \mathbf{S}[\mathcal{N}]:=\mathbf{S}[A]_{\mathcal{N}}:=-\mathbf{D}[\mathcal{N}\Vert\mathcal{R}],
    \end{equation}
    where $\mathbf{D}[\cdot\Vert\cdot]$ is a generalized channel divergence such that $\mathbf{S}[\mathcal{N}]$ is a satisfactory entropic function of a channel. 
\end{dfn} 

\textit{Generalized entropy function of a quantum channel~\cite{SPSD24}}--- The generalized entropy function $\mathbf{S}[\mathcal{N}]$ of a point-to-point channel $\mathcal{N}_{A'\to A}$ is defined as the generalized channel divergence between the quantum channel $\mathcal{N}$ and the completely mixing map $\mathcal{R}$,
    \begin{equation}
        \mathbb{S}[\mathcal{N}]:=\mathbb{S}[A]_{\mathcal{N}}:=-\mathbb{D}[\mathcal{N}\Vert\mathcal{R}].
    \end{equation}

To formulate entropic functions that properly capture the notion of the conditional entropies of channel, we first revisit the definitions of von Neumann entropy $S(A)_{\rho}$, conditional entropy $S(A|B)_{\rho}$, and mutual information $I(A;B)_{\rho}$ for states in terms of the quantum relative entropy, with no explicit dependence on the dimensional factors.
\begin{align}
    S(A)_{\rho} & = -\tr(\rho_A\log\rho_A) = - D(\rho_A\Vert\mathbbm{1}_A),\label{eq:vN-RE}\\
    S(A|B)_{\rho} & = S(AB)_{\rho}-S(B)_{\rho} = - D(\rho_{AB}\Vert\mathbbm{1}_{A}\otimes\rho_{B})=-\inf_{\omega\in{\rm St}(B)}D(\rho_{AB}\Vert\mathbbm{1}_A\otimes\omega_B),\label{eq:cvN-RE}\\
    I(A;B)_{\rho} & = S(A)_{\rho}-S(A|B)_{\rho} = D(\rho_{AB}\Vert\rho_A\otimes\rho_B)=\inf_{\tau\in{\rm St}(A),\omega\in{\rm St}(B)}D(\rho_{AB}\Vert\tau_A\otimes\omega_B). \label{eq:mi-RE}
\end{align}
We note that the negative conditional entropy $-S(A|B)_{\rho}$ of a bipartite state $\rho_{AB}$ is equal to the coherent information $I(A\rangle B)_{\rho}$, a information metric that is of significance in the context of quantum communication~\cite{BNS98,ADHW09}.

The von Neumann entropy, conditional entropy, and mutual information of states $\rho_{AB}$ (Eqs.~\eqref{eq:vN-RE}-\eqref{eq:mi-RE}) have been extended (generalized) in the literature using appropriate generalized state divergences $\mathbf{D}$ in the following way (e.g.,~\cite{TBH14,MDSFT13,BD10}).
\begin{align}
    \mathbf{S}(A)_{\rho} & = - \mathbf{D}(\rho_A\Vert\mathbbm{1}_A),\label{eq:ge-vN}\\
    \mathbf{S}(A|B)_{\rho} & = -\inf_{\omega\in{\rm St}(B)}\mathbf{D}(\rho_{AB}\Vert\mathbbm{1}_A\otimes\omega_B),\label{eq:ge-cvN}\\
    \mathbf{I}(A;B)_{\rho} & =\inf_{\tau\in{\rm St}(A),\omega\in{\rm St}(B)}\mathbf{D}(\rho_{AB}\Vert\tau_A\otimes\omega_B).\label{eq:ge-mi}
\end{align}

We here intentionally recalled forms of the definitions that don't have explicit dependence on the dimensional factors, as the definitions then work also for infinite-dimensional systems. It also makes the form look simpler with less number of variables or arguments when thinking of generalizing the ideas to higher-order processes. Writing entropic and information quantities, whenever possible, in terms of the single function instead of a sum or differences of functions also helps in generalizing these quantities to a wider class or family. That is, the form $f(\cdot)=g(\cdot)+h(\cdot)$ could be less tractable (or convenient) than $f(\cdot)=j(\cdot)$ as the latter provides comprehensive insight for generalization compared to the former. Also, in infinite-dimensional systems, there are situations for certain entropic functions that the absolute values of some terms in the sum could blow up while the single-term form may provide an approach to get finite value whenever possible (cf.~\cite{Shi23}). 

When defining the entropy function for higher order quantum processes~\cite{CDGP08,Gou19} based on the observations made in \cite{Gou19,GW21}, it was observed in \cite{SPSD24} that, crudely speaking, the role of the identity operator $\mathbbm{1}$ at the level of the state is played by
\begin{enumerate}
    \item the replacer map $\mathcal{R}^{(1)}$ at the level of channels that outputs the identity operator $\mathbbm{1}=:\mathcal{R}^{(0)}$.
    \item the replacer supermap $\mathcal{R}^{(2)}$ at the level of superchannels that outputs the replacer map $\mathcal{R}^{(1)}$.
    \item the replacer $n$-order map $\mathcal{R}^{(n)}$ at the level of $n$-order processes that outputs the replacer $(n-1)$-order map $\mathcal{R}^{(n-1)}$, for $n\geq 1$.
\end{enumerate}
Also, the role of a generalized state divergence for entropy of states would be played by a generalized $n$-order process divergence defined at the level $n$-order process, that are monotone under the action of $(n+1)$-order processes.

If we focus simply on generalization the formula $S(A|B)_{\rho}$ in Eq.~\eqref{eq:cvN-RE} to the conditional entropy of a bipartite channel, we focus on $S(A|B)_{\rho}=-D(\rho_{AB}\Vert\mathbbm{1}_A\otimes\rho_B)$ and $S(A|B)_{\rho}=-\inf_{\omega\in{\rm St}(B)}D(\rho_{AB}\Vert\mathbbm{1}_A\otimes\rho_B)$. We have to inspect on how to ``channelize" $\mathbbm{1}_A\otimes\omega_B$ given that $\rho_{AB}$ is substituted with channel $\mathcal{N}_{A'B'\to AB}$ and relative entropy $D(\cdot\Vert\cdot)$ between states is replaced with relative entropy $D[\cdot\Vert\cdot]$ between channels. Two of the multiple possible ways could be
\begin{itemize}
    \item \textit{Candidate A}. Taking $\mathcal{R}_{A\to A}\circ\mathcal{N}_{A'B'\to AB}$ as a possible channelization of $\mathbbm{1}_{A}\otimes\rho_B$, since $\mathbbm{1}_A\otimes\rho_B=\mathcal{R}_{A\to A}(\rho_{AB})$. This would give a candidate for the conditional entropy of $\mathcal{N}_{A'B'\to AB}$ equivalent to $-D(\rho_{AB}\Vert\mathbbm{1}_A\otimes\rho_B)$,
    \begin{equation}\label{eq:CA}
       S^{\not\to}[A|B]_{\mathcal{N}}:= -D[\mathcal{N}_{A'B'\to AB}\Vert \mathcal{R}_{A\to A}\circ\mathcal{N}_{A'B'\to AB}].
    \end{equation}
    \item \textit{Candidate B}. Taking the form $\mathcal{R}_{A'\to A}\otimes\mathcal{M}_{B'\to B}$, where $\mathcal{M}\in{\rm Q}(B',B)$, to be a possible channelization of $\mathbbm{1}_A\otimes\omega_B$ as $\omega\in{\rm St}(B)$. $\mathbbm{1}_A$ is substituted by $\mathcal{R}_{A'\to A}$ and state $\omega_B$ is substituted by channel $\mathcal{M}_{B'\to B}$. This would give the following potential candidate for the conditional entropy of $\mathcal{N}_{A"B'\to AB}$,
    \begin{equation}\label{eq:CB}
        S[A|B]_{\mathcal{N}}= -\inf_{\mathcal{M}\in{\rm Q}(B',B)}D[\mathcal{N}_{A'B'\to AB}\Vert \mathcal{R}_{A'\to A}\otimes\mathcal{M}_{B'\to B}].
    \end{equation}
\end{itemize}
\begin{remark}
    While Candidate A (Eq.~\eqref{eq:CA}) is simpler than Candidate B (Eq.~\eqref{eq:CB}) for the conditional entropy of channel, the form of Candidate B seems generalizable to the (multipartite) mutual information of multipartite channels in a straightforward manner. We can even consider some other appropriate generalized channel divergence $\mathbf{D}$ instead of $D$ in Eqs.~\eqref{eq:CA} and \eqref{eq:CB} to lift the definitions there to the potential candidates for generalized conditional entropies.

    A novel definition of the conditional entropy of a bipartite channel is desired, meaning that it can be used to identify the presence of certain behaviors in quantum processes. For example, the negative conditional entropy of quantum states necessarily reveals the presence of entanglement in those states. Also, it would be a bonus if the dimensional bounds for the conditional entropy of a bipartite channel are distinct from that of the conditional entropy of states (or entropy of channel).
\end{remark}

We first begin to discuss with Candidate B as our preferred definition for the conditional entropy of the channel. We do so because the form of Candidate B is generalizable to the definition for mutual information of channel. We will, in the latter sections, see that Candidate B seems operationally preferable to Candidate A in general.
\subsection{Definitions and properties}
\begin{dfn}[Generalized conditional entropy of a bipartite quantum channel]
    The generalized conditional entropy $\mathbf{S}[A|B]_{\mathcal{N}}$ of an arbitrary bipartite quantum channel $\mathcal{N}_{A'B'\to AB}$ is defined as
    \begin{equation}\label{eq:def-gen-con-ent}
        \mathbf{S}[A|B]_{\mathcal{N}}\coloneqq - \inf_{\mathcal{M}\in\mathrm{Q}(B',B)} \mathbf{D}[\mathcal{N}_{A'B'\to AB}\Vert \mathcal{R}_{A'\to A}\otimes\mathcal{M}_{B'\to B}],
    \end{equation}
    where the infimum is over all channels $\mathcal{M}_{B'\to B}$ and $\mathbf{D}[\cdot\Vert\cdot]$ is a generalized channel divergence such that the axioms of conditional entropies of a bipartite quantum channel are satisfied.

    The resource-constraint generalized entropy $\mathbf{S}_{\mathrm{Access}}[A|B]_{\mathcal{N}}$ of an arbitrary bipartite channel $\mathcal{N}_{A'B'\to AB}$ is defined as
    \begin{equation}\label{eq:res-constraint-con-ent}
        \mathbf{S}_{\mathrm{Access}}[A|B]_{\mathcal{N}}\coloneqq - \inf_{\mathcal{M}\in\mathrm{Q}(B',B)} \mathbf{D}_{\mathrm{Access}}[\mathcal{N}_{A'B'\to AB}\Vert \mathcal{R}_{A'\to A}\otimes\mathcal{M}_{B'\to B}],
    \end{equation}
    where the infimum is over all channels $\mathcal{M}_{B'\to B}$. In particular, if we optimize (supremum)  $\mathbf{D}_{\rm Access}$ over the set $\mathrm{E}_e=\{\rho_{RA'B'} : \tr[\widehat{H}_{RAB}\rho_{RA'B'}]\leq e\in\mathbb{R}^+\}$ of input states $\rho_{RA'B'}$ with bounded energy ($\leq e\in(0,\infty)$), then $\mathbf{S}_{\mathrm{E}_e}[A|B]_{\mathcal{N}}$, with $\mathrm{Acess}=\mathrm{E}_e$, is called the energy-constraint generalized conditional entropy.
\end{dfn}
A direct consequence of the above definitions is
\begin{align}
    \mathbf{S}[A|B]_{\mathcal{N}} \leq \mathbf{S}_{\mathrm{Access}}[A|B]_{\mathcal{N}}, \qquad \mathbb{S}[A|B]_{\mathcal{N}} \leq \mathbb{S}_{\mathrm{Access}}[A|B]_{\mathcal{N}}.
\end{align}
If $\mathbb{F}(\cdot\Vert\cdot)$ and $\mathbb{G}(\cdot\Vert\cdot)$ are two generalized state divergences such that for arbitrary pair of (valid) states $\rho,\sigma$, $\mathbb{F}(\rho\Vert\sigma)\leq \mathbb{G}(\rho\Vert\sigma)$. Now applying Lemma~\ref{lem:gen-ch-order}, when $|A|<+\infty$, we note that the resource-constraint generalized conditional entropy functions $\mathbb{F}_{\rm Access}[A|B]_{\mathcal{N}}$ and $\mathbb{G}_{\rm Access}[A|B]_{\mathcal{N}}$ of an arbitrary bipartite channel $\mathcal{N}_{A'B'\to AB}$ are ordered as
\begin{equation}\label{eq:gen-con-order}
    \mathbb{F}_{\rm Access}[A|B]_{\mathcal{N}}\geq \mathbb{G}_{\rm Access}[A|B]_{\mathcal{N}}.
\end{equation}

Using the notion of the purified distance and generalized channel divergence function, we define smoothed the generalized conditional entropy function $\mathbb{S}[A|B]_{\mathcal{N}}$ of a quantum channel $\mathcal{N}_{A'B'\to AB}$:
\begin{dfn}
    The smooth conditional generalized entropy function ${\mathbb{S}^{\epsilon}}[A|B]_{\mathcal{N}}$ of a bipartite quantum channel $\mathcal{N}_{A'B' \rightarrow AB}$ is defined as, for $\epsilon \in [0,1)$, 
    \begin{equation}
        {\mathbb{S}^{\epsilon}}[A|B]_{\mathcal{N}}  = \sup_{\mathcal{M}\in{\rm Q}(A'B',AB):\, P[\mathcal{N},\mathcal{M}] \leq \epsilon} \mathbb{S}[A|B]_{\mathcal{M}},
    \end{equation}
    where $P[\mathcal{K},\mathcal{L}]$ is purified channel distance~\cite{GLN05} between $\mathcal{K},\mathcal{L}\in{\rm Q}(A',A)$,
    \begin{equation}\label{eq:sin-ch}
    P[\mathcal{K},\mathcal{L}] := \sup_{{\psi}\in{\rm St}{(RA')}} P \left(\mathcal{K} (\psi_{RA'}), \mathcal{L}(\psi_{RA'})\right),
\end{equation}
$P(\cdot, \cdot)$ is the purified distance and it suffices to take $\psi_{RA'}$ to be pure.
\end{dfn}

\begin{dfn}[Reduced channel]
    For a bipartite channel $\mathcal{N}_{A'B'\to AB}$, its reduced channel on $A$ is a quantum channel $\mathcal{N}_{A'B'\to A}:=\tr_{B}\circ\mathcal{N}_{A'B'\to AB}$ obtained by (partially) tracing out $B$ from the output of $\mathcal{N}_{A'B'\to AB}$. The channel $\mathcal{N}_{A'B'\to A}$ is called a reduced channel of $\mathcal{N}_{A'B'\to AB}$.
\end{dfn}

\begin{theorem}\label{thm:axioms}
    The generalized conditional entropy $\mathbf{S}[A|B]_{\mathcal{N}}$ of an arbitrary bipartite channel $\mathcal{N}_{A'B'\to AB}$ satisfies the axioms (P1-P4) for the conditional entropy functions, given that the associated generalized state divergence $\mathbf{D}$ is additive and nonnegative for states.
 \begin{enumerate}
     \item[A1] It is nondecreasing under the actions of an arbitrary local $\mathcal{R}$-preserving superchannel $\widehat{\Theta}\in\mathrm{SC}((A',A),(C',C))$ and arbitrary local superchannel $\Theta\in\mathrm{SC}((B',B),(D',D))$, i.e.,
    \begin{equation}
        \mathbf{S}[A|B]_{\mathcal{N}}\leq \mathbf{S}[C|D]_{\mathcal{N}'},
    \end{equation}
    where $\mathcal{N}'_{C'D'\to D'D}= \widehat{\Theta}\otimes\Theta [\mathcal{N}_{A'B'\to AB}]$ is the resultant bipartite channel.
    \item[A2]\label{thm:con-dont-ent}[Conditioning doesn't increase entropy]
The generalized conditional entropy $\mathbf{S}[A|B]_{\mathcal{N}}$ is always less than or equal to the generalized entropy of its reduced channel $\mathcal{N}_{A'B'\to AB}$, i.e.,
\begin{equation}
    \mathbf{S}[A|B]_{\mathcal{N}}\leq \mathbf{S}[A]_{\mathcal{N}}:=\mathbf{S}[\tr_{B}\circ\mathcal{N}_{A'B'\to AB}].
\end{equation}
 \item[A3][Normalization] If $\mathcal{N}_{A'B'\to AB}$ is a replacer channel, i.e.,
$\mathcal{N}_{A'B'\to AB}(\rho_{A'B'})=\sigma_{AB}$ for all $\rho\in \mathrm{St}(A'B')$ and some state $\sigma_{AB}$,
\begin{equation}
    \mathbf{S}[A|B]_{\mathcal{N}}=\mathbf{S}(A|B)_{\sigma}.
\end{equation}
\item[A4]\label{thm:pro-ch-gen-div}
    For a tensor-product channel $\mathcal{N}_{A'B'\to AB}=\mathcal{N}^1_{A'\to A}\otimes \mathcal{N}^2_{B'\to B}$ for some $\mathcal{N}^1\in{\rm Q}(A',A)$ and $\mathcal{N}^2\in{\rm Q}(B',B)$,
    \begin{equation}
        \mathbf{S}[A|B]_{\mathcal{N}}=\mathbf{S}[A]_{\mathcal{N}^1}.
    \end{equation}
 \end{enumerate} 
\end{theorem}
    \begin{proof}
\textit{Proof of A1}.---  Consider an arbitrary channel $\mathcal{M}_{B'\to B}$, the following inequalities hold.
    \begin{align}
 & \mathbf{D}[\mathcal{N}_{A'B'\to AB}\Vert\mathcal{R}_{A'\to A}\otimes\mathcal{M}_{B'\to B}]\nonumber\\
 &   \geq  \mathbf{D}[\widehat{\Theta}\otimes\Theta[\mathcal{N}_{A'B'\to AB}]\Vert\widehat{\Theta}[\mathcal{R}_{A'\to A}]\otimes\Theta[\mathcal{M}_{B'\to B}]]\nonumber\\
      & =  \mathbf{D}[\widehat{\Theta}\otimes\Theta[\mathcal{N}_{A'B'\to AB}]\Vert\mathcal{R}_{C'\to C}\otimes\Theta[\mathcal{M}_{B'\to B}]]\nonumber\\
      &  \geq  \inf_{\mathcal{M}'\in\mathrm{Q}(D', D)}\mathbf{D}[\widehat{\Theta}\otimes\Theta[\mathcal{N}_{A'B'\to AB}]\Vert\mathcal{R}_{C'\to C}\otimes\mathcal{M}'_{D'\to D}]\nonumber\\
      & = - \mathbf{S}[C|D]_{\widehat{\Theta}\otimes\Theta[\mathcal{N}]},
    \end{align}
    where the first inequality follows from the monotonicity of the generalized channel divergence, and the final equality follows from the definition. As the above relations holds for arbitrary quantum channel $\mathcal{M}_{B'\to B}$, they will also hold true for a channel $\mathcal{M}_{B'\to B}$ that minimizes $\mathbf{D}[\mathcal{N}_{A'B'\to AB}\Vert\mathcal{R}_{A'\to A}\otimes\mathcal{M}_{B'\to B}]$. That is, $-\mathbf{S}[A|B]_{\mathcal{N}}\geq -\mathbf{S}[C|D]_{\widehat{\Theta}\otimes\Theta[\mathcal{N}]}$. 
    
\textit{Proof of A2}.---  We have
    \begin{align}
        -\mathbf{S}[A]_{\mathcal{N}}& =
        -\mathbf{S}[\tr_{B}\circ\mathcal{N}_{A'B'\to AB}]\nonumber\\
        &= \mathbf{D}[\tr_{B}\circ\mathcal{N}_{A'B'\to AB}\Vert\mathcal{R}_{A'B'\to A}]\nonumber\\
       &= \mathbf{D}[\tr_{B}\circ\mathcal{N}_{A'B'\to AB}\Vert\tr_{B}\circ\mathcal{R}_{A'\to A}\otimes\mathcal{M}_{B'\to B}]\nonumber\\
       & \leq \mathbf{D}[\mathcal{N}_{A'B'\to AB}\Vert\mathcal{R}_{A'\to A}\otimes\mathcal{M}_{B'\to B}]\nonumber\\
        & \leq \inf_{\mathcal{M}\in{\rm Q}(B',B)}\mathbf{D}[\mathcal{N}_{A'B'\to AB}\Vert\mathcal{R}_{A'\to A}\otimes\mathcal{M}_{B'\to B}]\nonumber\\
      & \leq -\mathbf{S}[A|B]_{\mathcal{N}}.
    \end{align}

 \textit{Proof of A3}.--- 
    Consider a quantum channel $\mathcal{M}_{B'\to B}$ such that the first following equality holds.
    \begin{align}\label{eq:thm-proof-def}
        - \mathbf{S}[A|B]_{\mathcal{N}} & =\mathbf{D}[\mathcal{N}_{A'B'\to AB}\Vert \mathcal{R}_{A'\to A}\otimes\mathcal{M}_{B'\to B}]\nonumber\\
        &\geq \mathbf{D}(\mathcal{N}_{A'B'\to AB}(\rho_{A'B'})\Vert \mathcal{R}_{A'\to A}\otimes\mathcal{M}_{B'\to B}(\rho_{A'B'}))\nonumber\\
        &= \mathbf{D}(\sigma_{AB}\Vert \mathbbm{1}_A\otimes\mathcal{M}_{B'\to B}(\rho_{B'}))\nonumber\\
        &\geq \inf_{\omega_{B}\in\mathrm{St}(B)}\mathbf{D}(\sigma_{AB}\Vert\mathbbm{1}_A\otimes\omega_{B})=-\mathbf{S}(A|B)_{\sigma}.
    \end{align}
    Now consider an arbitrary replacer channel $\mathcal{M}'_{B'\to B}$ whose action is $\mathcal{M}'_{B'\to B}(\rho_{B'})=\omega_B$ for all $\rho_{B'}\in\mathrm{St}(B')$ and some state $\omega_B$. Let $\rho_{RA'B'}$ be a state that gives supremum of $\mathbf{D}[\mathcal{N}_{A'B'\to AB}\Vert \mathcal{R}_{A'\to A}\otimes\mathcal{M}'_{B'\to B}]$. Then,
    \begin{align}
        - \mathbf{S}[A|B]_{\mathcal{N}} &\leq \mathbf{D}[\mathcal{N}_{A'B'\to AB}\Vert \mathcal{R}_{A'\to A}\otimes\mathcal{M}'_{B'\to B}]\nonumber\\
       & = \mathbf{D}(\rho_R\otimes\sigma_{AB}\Vert\rho_R\otimes\mathbbm{1}_A\otimes\omega_B)\nonumber\\
       & =\mathbf{D}(\sigma_{AB}\Vert\mathbbm{1}_A\otimes\omega_B).
    \end{align}
    As the above inequality holds for arbitrary replacer channel $\mathcal{M}'_{B'\to B}$, it also holds for replacer channel for which $\omega_B$ gives infimum to $\mathbf{D}(\sigma_{AB}\Vert\mathbbm{1}_A\otimes\omega_B)$. That is, $ - \mathbf{S}[A|B]_{\mathcal{N}}\leq-\mathbf{S}(A|B)_{\sigma}\leq - \mathbf{S}[A|B]_{\mathcal{N}} $, which concludes the proof. 

\textit{Proof of A4}.--- We first prove that $-\mathbf{S}[A|B]_{\mathcal{N}}\geq -\mathbf{S}[A]_{\mathcal{N}^1}$. Let $\mathcal{M}_{B'\to B}$ be a channel such that the following identity holds
\begin{align}
 -\mathbf{S}[A|B]_{\mathcal{N}}= &  \mathbf{D}[\mathcal{N}^1_{A'\to A}\otimes\mathcal{N}^2_{B'\to B}\Vert\mathcal{R}_{A'\to A}\otimes\mathcal{M}_{B'\to B}]\nonumber\\
 \geq & \sup_{\psi_{R_AA'}\otimes\phi_{R_BB'}\in\mathrm{St}(R_AA'R_BB')}\mathbf{D}(\mathcal{N}^1(\psi_{R_AA'})\otimes\mathcal{N}^2(\phi_{R_BB'})\Vert\mathcal{R}(\psi_{R_{A}A'})\otimes\mathcal{M}(\phi_{R_BB'}))\nonumber\\
 = & \sup_{\psi_{R_AA'}} \mathbf{D}(\mathcal{N}^1(\psi_{R_AA'})\Vert\mathcal{R}(\psi_{R_AA'}))+ \sup_{R_BB'}\mathbf{D}(\mathcal{N}^2(\phi_{R_BB'})\Vert\mathcal{M}(\phi_{R_BB'}))\nonumber\\
 \geq & -\mathbf{S}[A]_{\mathcal{N}^1},
\end{align}
where to get the second equality and the second inequality, we have employed additive and nonnegative properties of $\mathbf{D}(\cdot\Vert\cdot)$ for states. Now, observing that appending a channel keeps the generalized channel divergence between two quantum channels invariant, we have $\mathbf{D}[\mathcal{N}^1_{A'\to A}\otimes\mathcal{N}^2_{B'\to B}\Vert\mathcal{R}_{A'\to A}\otimes\mathcal{N}^2_{B'\to B}]=\mathbf{D}[\mathcal{N}^1_{A'\to A}\Vert\mathcal{R}_{A'\to A}]=-\mathbf{S}[\mathcal{N}^1]$. This concludes the proof.
\end{proof}

We note that the additive and nonnegative property of the generalized state divergence was used only in the proof of property A4. The following corollary is a direct consequence of the property A1 in the theorem above.
\begin{corollary}
    The generalized conditional entropy $\mathbf{S}[A|B]_{\mathcal{N}}$ of an arbitrary bipartite channel $\mathcal{N}_{A'B'\to AB}$ remains invariant when locally acted with unitary superchannels ${\Theta}^1\in\mathrm{SC}((A',A),(C',C)), {\Theta}^2\in\mathrm{SC}((B',B),(D',D))$, i.e.,
    \begin{equation}
        \mathbf{S}[A|B]_{\mathcal{N}}=\mathbf{S}[C|D]_{\mathcal{N}'},
    \end{equation}
    where $\mathcal{N}'_{C'D'\to CD}=\Theta^1\otimes\Theta^2[\mathcal{N}_{A'B'\to AB}]$.
\end{corollary}

\begin{theorem}[Strong subadditivity]
    For an arbitrary tripartite channel $\mathcal{N}_{A'B'C'\to ABC}$, the generalized (conditional) entropy exhibits strong subadditivity property in the following sense, 
\begin{equation}
    \mathbf{S}[A|BC]_{\mathcal{N}}\leq \mathbf{S}[A|B]_{\mathcal{N}},
\end{equation}
where $\mathcal{N}_{A'B'C'\to AB}:=\tr_{C}\circ\mathcal{N}_{A'B'C'\to AB}$ is a reduced channel of $\mathcal{N}_{A'B'C'\to ABC}$.
\end{theorem}
\begin{proof}
    Consider $\mathcal{M}_{B'C'\to BC}$ be a channel such that the first following identity holds,
    \begin{align}
        - \mathbf{S}[A|BC]_{\mathcal{N}}& = \mathbf{D}[\mathcal{N}_{A'B'C'\to ABC}\Vert\mathcal{R}_{A'\to A}\otimes\mathcal{M}_{B'C'\to BC}]\nonumber\\
        & \geq \mathbf{D}[\tr_C\circ\mathcal{N}_{A'B'C'\to ABC}\Vert\tr_C\circ\mathcal{R}_{A'\to A}\otimes\mathcal{M}_{B'C'\to BC}]\nonumber\\
        &\geq \inf_{\mathcal{P}\in\mathrm{Q}(B'C',B)}\mathbf{D}[\tr_C\circ\mathcal{N}_{A'B'C'\to ABC}\Vert\mathcal{R}_{A'\to A}\otimes\mathcal{P}_{B'C'\to B}]= - \mathbf{S}[A|B]_{\mathcal{N}},
    \end{align}
    where the first inequality follows from the monotonicity of the generalized channel divergence as tracing out is a quantum channel and composition with a channel is an action of a superchannel.
\end{proof}

\section{Families of conditional quantum entropies} \label{sec:family_of_conditional_entropies}
In this section, we inspect the properties of conditional entropies derived from some examples of generalized channel divergences. These channel divergences are obtained from some of the well-known and widely used examples of generalized state divergences, namely (Umegaki) relative entropy, geometric R\'enyi relative entropy, and max- and min-relative entropies.

\subsection{Conditional von Neumann entropy}
The von Neumann conditional entropy of an arbitrary bipartite channel $\mathcal{N}_{A'B'}$ is defined as
\begin{equation}
    S[A|B]_{\mathcal{N}}:= -\inf_{\mathcal{M}\in{\rm Q}(B',B)}D[\mathcal{N}_{A'B'\to AB}\Vert\mathcal{R}_{A'\to A}\otimes\mathcal{M}_{B'\to B}],
\end{equation}
where $D[\cdot\Vert\cdot]$ is the (Umegaki) relative entropy between completely positive maps, derived from the (Umegaki) relative entropy between states.
\begin{proposition}\label{prop:vN-upperbound}
    For an arbitrary bipartite channel $\mathcal{N}_{A'B'\to AB}$, its von Neumann conditional entropy $S[A|B]_{\mathcal{N}}$ is upper bounded as
    \begin{equation}
        S[A|B]_{\mathcal{N}}\leq \inf_{\op{\psi}\in{\rm St}(RA'B')}S(A|RB)_{\mathcal{N}(\op{\psi})}.
    \end{equation}
\end{proposition}
\begin{proof}
From the definition of the von Neumann conditional entropy of a bipartite channel $\mathcal{N}_{A'B'\to AB}$, we have
\begin{align}
        -S[A|B]_{\mathcal{N}}& = \inf_{\mathcal{M}\in{\rm Q}(B',B)}D[\mathcal{N}_{A'B'\to AB}\Vert\mathcal{R}_{A'\to A}\otimes\mathcal{M}_{B'\to B}]\nonumber\\
      & = \inf_{\mathcal{M}\in{\rm Q}(B',B)}\sup_{\op{\phi}\in{\rm St}(RA'B')}D(\mathcal{N}(\phi_{RA'B'})\Vert\mathcal{R}\otimes\mathcal{M}(\phi_{RA'B'}))\nonumber\\
      & = \inf_{\mathcal{M}\in{\rm Q}(B',B)}\sup_{\op{\phi}\in{\rm St}(RA'B')}D(\mathcal{N}(\phi_{RA'B'})\Vert\mathbbm{1}_A\otimes\mathcal{M}(\phi_{RB'}))\nonumber\\
      & \geq \sup_{\op{\phi}\in{\rm St}(RA'B')} D(\mathcal{N}(\phi_{RA'B'})\Vert\mathbbm{1}_A\otimes\tr_A(\mathcal{N}(\phi_{RA'B'})))\nonumber\\
      & = -\inf_{\op{\phi}\in{\rm St}(RA'B')}S(A|RB)_{\mathcal{N}(\phi)}.
\end{align}
Here the inequality holds due to \cite[Lemma 1]{DMHB13} which states that, for an arbitrary state $\rho_{AB}$ and $\sigma\in{\rm Pos}(A)$, we have $\inf_{\omega\in{\rm St}(B)}D(\rho_{AB}\Vert\sigma_{A}\otimes\omega_B)=D(\rho_{AB}\Vert\sigma_A\otimes\rho_{B})$.
\end{proof}

\begin{lemma}\label{lem:tele-cov-reduced}
Consider a reduced channel $\mathcal{T}^{\mathcal{N}}_{B'\to B}$ obtained from a bipartite tele-covariant channel $\mathcal{N}_{A'B'\to AB}$~\cite[Definition 4]{DBW17}, when $|A'|,|A|<+\infty$, and the maximally mixed state $\pi_{A'}$,
\begin{equation}\label{eq:tele-cov-reduced}
\mathcal{T}^{\mathcal{N}}_{B'\to B}(\cdot):=\tr_{A}\circ\mathcal{N}_{A'B'\to AB}\big(\pi_{A'}\otimes(\cdot)\big),
\end{equation}
then the bipartite channels $\frac{1}{|A|}\mathcal{R}_{A'\to A}\otimes\mathcal{T}^{\mathcal{N}}_{B'\to B}$ and $\mathcal{N}_{A'B'\to AB}$ are always jointly tele-covariant (see~\cite[Remark 1]{DBW17} and~\cite[Definition 5]{DW19}).
 
Consider a reduced channel $\mathcal{T}^{\mathcal{N}}_{B'C'\to B}$ obtained from a tele-covariant tripartite channel $\mathcal{N}_{A'B'C'\to ABC}$ \cite[Definition 24]{DBWH21},
\begin{equation}\label{eq:tele-cov-reduced-tri}
\mathcal{T}^{\mathcal{N}}_{B'C'\to B}(\cdot):=\tr_{C}\circ\tr_{A}\circ\mathcal{N}_{A'B'C'\to ABC}(\pi_{A'}\otimes\cdot),
\end{equation}
then the channels $\frac{1}{|A|}\mathcal{R}_{A'\to A}\otimes\mathcal{T}^{\mathcal{N}}_{B'C'\to B}$ and $\tr_C\circ\mathcal{N}_{A'B'C'\to ABC}$ are always jointly tele-covariant.
\end{lemma}

\begin{theorem}\label{thm:vN-choi-state}
    The von Neumann conditional entropy of an arbitrary tele-covariant channel $\mathcal{N}_{A'B'\to AB}$, when $|A'|<+\infty$, is related to the von Neumann conditional entropy of its Choi state in the following way,
    \begin{equation}
   S[A|B]_{\mathcal{N}}= S(R_AA|R_BB)_{\Phi^\mathcal{N}}-\log|A'|,
    \end{equation}
where $\Phi^{\mathcal{N}}_{R_AAR_BB}:=\mathcal{N}(\Phi_{R_AA'}\otimes\Phi_{R_BB})$, $R_A\simeq A', R_B\simeq B'$.
\end{theorem}
\begin{proof}
From Lemma~\ref{lem:tele-cov-reduced} and \cite[Remark 5]{DBWH21}, it directly follows that the bipartite channel $\frac{1}{|A|}\mathcal{R}_{A'\to A}\otimes\mathcal{T}^{\mathcal{N}}_{B'\to B}$ is jointly tele-covariant with the channel $\mathcal{N}_{A'B'\to AB}$, where $\mathcal{T}^{\mathcal{N}}$ is defined in Lemma~\ref{lem:tele-cov-reduced}. Then, $-S[A|B]_{\mathcal{N}}\leq {D}[\mathcal{N}\Vert\mathcal{R}\otimes\mathcal{T}^{\mathcal{N}}]={D}(\Phi^{\mathcal{N}}\Vert\frac{1}{|A'|}\mathbbm{1}_{R_AA}\otimes\Phi^{\mathcal{T}^{\mathcal{N}}})=\log|A'|-S(R_AA|R_BB)_{\Phi^\mathcal{N}}$, where we have used the result \cite[Corollary II.5]{LKDW18} for the first equality. We now prove the converse part to conclude the proof. For a quantum channel $\mathcal{M}_{B'\to B}$ that gives infimum for ${D}[\mathcal{N}_{A'B'\to AB}\vert\mathcal{R}_{A'\to A}\otimes\mathcal{M}_{B'\to B}]$, we have
\begin{align}
  -{S}[A|B]_{\mathcal{N}} &=  {D}[\mathcal{N}_{A'B'\to AB}\vert\mathcal{R}_{A'\to A}\otimes\mathcal{M}_{B'\to B}] \nonumber\\
  &\geq {D}(\Phi^{\mathcal{N}}_{R_AAR_BB}\Vert\frac{1}{|A'|}\mathbbm{1}_{R_AA}\otimes\Phi^{\mathcal{M}}_{R_BB})\nonumber\\
 &\geq \log|A'|+ \inf_{\sigma_{R_BB}\in\mathrm{St}(R_BB)}{D}(\Phi^{\mathcal{N}}_{R_AAR_BB}\Vert\mathbbm{1}_{R_A}\otimes\mathbbm{1}_A\otimes\sigma_{R_BB})\nonumber \label{eq:thm-proof-con}\\
 &=\log|A'|-S(R_AA|R_BB)_{\Phi^\mathcal{N}}. 
\end{align}
\end{proof}

\begin{corollary}\label{cor:cVN-unitaries}
    The von Neumann conditional entropy ${S}[A|B]_{\mathcal{U}}$ of a tele-covariant bipartite isometric channel $\mathcal{U}_{A'B'\to AB}$ is related to the entropy of the marginal (reduced state) $\Phi^{\mathcal{U}}_{R_AA}$ of its Choi state $\Phi^{\mathcal{U}}_{R_AAR_BB}$,
    \begin{equation}
       {S}[A|B]_{\mathcal{U}}+\log|A'|= - S(R_AA)_{\Phi^{\mathcal{U}}}=- S(R_BB)_{\Phi^{\mathcal{U}}}.
    \end{equation}
\end{corollary}

\begin{proposition} \label{prop:lowerbound_cond_entropy}
    The absolute value of the von Neumann conditional entropy of any bipartite channel $\mathcal{N}_{A'B'\to AB}$ is upper bounded by the log of the dimension of the nonconditioning output system, i.e.,
    \begin{align}
    - \log|B|+ S[AB]_{\mathcal{N}}& \leq S[A|B]_{\mathcal{N}}\leq \log|A|,
    \end{align}
where $S[AB]_{\mathcal{M}}:=S[\mathcal{M}_{A'B'\to AB}]$ is the von Neumann entropy of a bipartite channel $\mathcal{M}_{A'B'\to AB}$ and from \cite{GW21},
\begin{equation}\label{eq:lower-bound-entropy}
    -\log\min\{|A||B|,|A'||B'|\}\leq  S[AB]_{\mathcal{N}}\leq \log|A||B|.
\end{equation}
\end{proposition}
\begin{proof}
To prove the lower bound, we notice that
\begin{align}
   - S[A|B]_{\mathcal{N}} &= \inf_{\mathcal{M}\in{\rm Q}(B',B)}  D[\mathcal{N}_{A'B'\to AB}\Vert\mathcal{R}_{A'\to A}\otimes\mathcal{M}_{B'\to B}]\nonumber\\
   & \leq D[\mathcal{N}_{A'B'\to AB}\Vert\mathcal{R}_{A'\to A}\otimes\frac{1}{|B|}\mathcal{R}_{B'\to B}]\nonumber\\
   & = \log|B|-S[AB]_{\mathcal{N}}.
\end{align}
$S[A|B]_{\mathcal{N}}\leq S[A]_{\mathcal{N}}$ is a direct consequence of [A2, Theorem~\ref{thm:axioms}] and $S[A]_{\mathcal{N}}\leq \log|A|$ follows from \cite{GW21}. We can also arrive at the upper bound using Proposition~\ref{prop:vN-upperbound}.

Inequality~\eqref{eq:lower-bound-entropy} is a direct consequence of \cite[Proposition 6]{GW21}. From \cite[Corollary 8]{GW21}, we can say that $S[AB]_{\mathcal{N}}=-\log|A'||B'|$ if and only if $\mathcal{N}_{A'B'\to AB}$ is an isometric channel. 
\end{proof}

Based on [A4, Theorem~\ref{thm:axioms}] and Proposition~\ref{prop:lowerbound_cond_entropy}, we note the following.
\begin{remark}\label{obs:cVN-bounds}
For all noninteraction bipartite channel $\mathcal{V}^1_{A'\to A}\otimes\mathcal{M}^2_{B'\to B}$, where $\mathcal{V}^1\in{\rm Q}(A',A), \mathcal{M}^2\in{\rm Q}(B',B)$,
    \begin{equation}
        S[A|B]_{\mathcal{V}^1\otimes\mathcal{M}^2}=-\log|A'|
    \end{equation}
if and only if $\mathcal{V}^1_{A'\to A}$ is an isometric channel, otherwise $S[A|B]_{\mathcal{V}^1\otimes\mathcal{M}^2}> -\log |A'|$ if $\mathcal{V}^1$ is not an isometric channel. For the (unitary) SWAP channel $\mathcal{S}_{A'B'\to AB}(\cdot):=\mathrm{SWAP}(\cdot)\mathrm{SWAP}^\dag$, we have $S[A|B]_{\mathcal{S}}=-3\log|A|$, where $\mathrm{SWAP}\ket{\psi}_{A'}\ket{\phi}_{B'}=\ket{\phi}_{A}\ket{\psi}_{B}$ and $|A'|=|B'|=|A|=|B|$. The least value $ -\log\min\{|A||B|^2,|A'||B'||B|\}$ for the von Neumann conditional entropy $S[A|B]_{\mathcal{N}}$ is not possible for any noisy (nonisometric) channel $\mathcal{N}_{A'B'\to AB}$. 
\end{remark}

\begin{proposition}[Continuity of conditional entropy]
If a bipartite channel $\mathcal{N}_{A'B'\to AB}$ is $\varepsilon$-close to a tele-covariant bipartite channel $\mathcal{T}_{A'B'\to AB}$ in the diamond norm, i.e., $\frac{1}{2}\norm{\mathcal{N}-\mathcal{M}}_{\diamond}= \varepsilon\in[0,1]$, then
\begin{equation}\label{eq:tele-afw}
    \abs{S[A|B]_{\mathcal{T}}-S[A|B]_{\mathcal{N}}}\leq 2\varepsilon\log|A'||A|+g_2(\varepsilon),
\end{equation}
where $g_2(\varepsilon):=(1+\varepsilon)\log(1+\varepsilon)-\varepsilon\log\varepsilon$.
\end{proposition}
\begin{proof}
    The proposition is a direct consequence of the AFW continuity theorem~\cite{Win16} of the conditional von Neumann entropy which states that for any bipartite state $\rho,\sigma\in\mathrm{St}(AB)$, if $\frac{1}{2}\norm{\rho-\sigma}_1\leq \varepsilon\in[0,1]$, then
    \begin{equation}\label{eq:AFW}
        \abs{S(A|B)_{\rho}-S(A|B)_{\sigma}}\leq 2\varepsilon\log|A|+g_2(\varepsilon).
    \end{equation}
If $\frac{1}{2}\norm{\mathcal{N}-\mathcal{T}}_{\diamond}= \varepsilon$ for any pair of channels (or, completely positive maps), then their Choi states are $\varepsilon$-close in trace distance as $\frac{1}{2}\norm{\id_{R}\otimes\mathcal{N}(\rho_{RA'B'})-\id_{R}\otimes\mathcal{T}(\rho_{RA'B'})}_1\leq \varepsilon$ for all $\rho_{RA'B'}\in\mathrm{St}(RA'B')$ with arbitrary $|R|$. As $\frac{1}{2}\norm{\Phi^{\mathcal{N}}_{R_AAR_BB}-\Phi^{\mathcal{T}}_{R_AAR_BB}}\leq \varepsilon$, invoking the AFW continuity theorem (inequality~\eqref{eq:AFW}), we arrive at the bound~\eqref{eq:tele-afw}.
\end{proof}
\subsubsection*{Examples of tele-covariant channels and their conditional entropies}
We now analytically and numerically calculate the von Neumann conditional entropy using Theorem \ref{thm:vN-choi-state} for some tele-covariant quantum channels of practical interest in quantum computing and information processing tasks.
\begin{figure}[hbtp]
\begin{center}
\includegraphics[width=0.5\columnwidth]{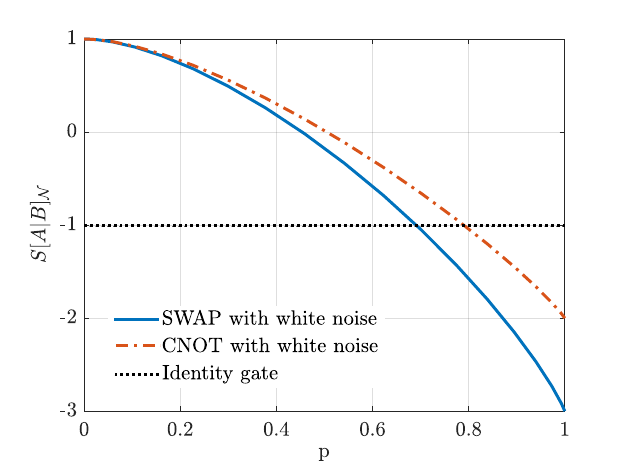}
\caption{In this figure, we plot the von Neumann conditional entropies $S[A|B]_{\mathcal{N}^{p,U}}$, where $U\in\{{\rm SWAP,CNOT}\}$, for noisy SWAP and CNOT operations with respect to the channels' parameter $p$, given by Eqs.~\eqref{equ:swap_ent} and \eqref{equ:cnot_ent}. The horizontal line at $S[A|B]_{\mathcal{N}}=-1$ is for the two-qubit noiseless identity gate. }
\label{conditional_entropy_tlecov}
\end{center}
\end{figure}

The von Neumann conditional entropy of the two-qu$d$it noisy SWAP gate $\mathcal{N}^{p,{\rm SWAP}(d)}_{A'B'\to AB}(\cdot):=(1-p)\widetilde{R} (\cdot) + p \mathrm{SWAP} (\cdot) \mathrm{SWAP}^\dag$, i.e., SWAP gate with white noise, is given by
\begin{equation}
    S[A|B]_{\mathcal{N}^{p,{\rm SWAP}(d)}}=-\left(p+\frac{1-p}{d^4}\right)\log\left(p+\frac{1-p}{d^4}\right)-(d^4-1)\left(\frac{1-p}{d^4}\log\frac{1-p}{d^4}\right)-3\log d,
\end{equation}
where $|A'|=|B'|=|A|=|B|=d$ and $\widetilde{\mathcal{R}}_{A'B'\to AB}$ is the completely mixing channel, i.e., $\widetilde{\mathcal{R}}_{A'B'\to AB}=\widetilde{\mathcal{R}}_{A'\to A}\otimes\widetilde{\mathcal{R}}_{B'\to B}$.

{\it SWAP and CNOT gates with white noise}.--- We now consider the two-qubit SWAP and CNOT gates mixed with the completely mixing channel $\widetilde{\mathcal{R}}_{A'B'\to AB}$ and $|A'|=|B'|=|A|=|B|=2$. The two-qubit noisy SWAP and CNOT gates are given by $\mathcal{N}^{p,\mathrm{SWAP}(2)}_{A'B'\to AB}(\cdot)$ and $ \mathcal{N}^{p,\mathrm{CNOT}}_{A'B'\to AB}(\cdot) = (1-p) \widetilde{R}(\cdot) + p \mathrm{CNOT} (\cdot) \mathrm{CNOT}^\dag$, respectively, for $p\in [0,1]$. The von Neumann conditional entropies of $\mathcal{N}^{p,\mathrm{SWAP}(2)}_{A'B'\to AB}$ and $\mathcal{N}^{p,\mathrm{CNOT}}_{A'B'\to AB}$ are, see Figure~\ref{conditional_entropy_tlecov},
\begin{align}
    S[A|B]_{\mathcal{N}^{p,\mathrm{SWAP}(2)}} &= -\frac{15p}{16} \log \frac{p}{16}-\left(1-\frac{15 p}{16}\right) \log (1-\frac{15 p}{16})-3,\label{equ:swap_ent}\\
      S[A|B]_{\mathcal{N}^{p,\mathrm{CNOT}}} &=  -\frac{15p}{16} \log \frac{p}{16}+\left(1-\frac{15 p}{16}\right) \log (1-\frac{15 p}{16})-\left(1-\frac{p}{2}\right) \log (\frac{1}{2}-\frac{p}{4})+\frac{p}{2} \log \frac{p}{4}-1. \label{equ:cnot_ent}
\end{align}

\subsection{Geometric R\'enyi conditional entropies}
In this section, we consider finite-dimensional systems and channels. Let $\mathcal{N}_{A' \rightarrow A}$ be a quantum channel and $\mathcal{M}_{A' \rightarrow A}$ be a completely positive map. The geometric R\'enyi relative entropy between $\mathcal{N}$ and $\mathcal{M}$ is defined as
\begin{equation}
    \widehat{D}_{\alpha} \left[\mathcal{N} \Vert \mathcal{M}\right] = \sup_{\psi_{RA'} \in \mathrm{St}(RA')} \widehat{D}_{\alpha} \left(\mathrm{id}_{R}\otimes\mathcal{N} \left(\psi_{RA'}\right) \Vert \mathrm{id}_{R}\otimes\mathcal{M} \left(\psi_{RA'}\right)\right),
\end{equation}
where $\alpha \in (0,1) \cup (1, \infty)$ and $\widehat{D}_{\alpha}\left(X \Vert Y\right):= \frac{1}{\alpha-1} \log \tr \left[ Y^{\frac{1}{2}} \left(Y^{-\frac{1}{2}} X Y^{-\frac{1}{2}}\right)^\alpha Y^{\frac{1}{2}}\right]$ is the geometric R\'enyi divergence between $X\in\mathrm{St}(A')$ and $Y\in\mathrm{Pos}(A')$.

For any finite-dimensional quantum channel $\mathcal{N}_{A' \rightarrow A}$, a finite-dimensional completely positive trace nonincreasing map $\mathcal{M}_{A' \rightarrow A}$, and $\alpha \in (1,2]$, the  geometric R\'enyi channel divergence is given by the following closed form~\cite{FF21}
\begin{equation}
    \widehat{D}_{\alpha} \left[\mathcal{N} \Vert \mathcal{M}\right] = \frac{1}{\alpha -1} \log\norm{\tr_{A} G_{1-\alpha} \left(\Gamma^{\mathcal{N}}_{RA}, \Gamma^{\mathcal{M}}_{RA}\right)}_{\infty}, \label{equ:geometric_renyi_div}
\end{equation}
where $G_{\alpha} \left(X,Y\right):= X^{\frac{1}{2}} \left(X^{-\frac{1}{2}}Y X^{-\frac{1}{2}}\right)^{\alpha}X^{\frac{1}{2}}$ is the weighted matrix geometric mean, $\Gamma^{\mathcal{N}}_{RA}$ denotes the Choi operator of $\mathcal{N}$. It is monotonically nonincreasing under the action of quantum superchannels for $\alpha\in(1,2]$.

\begin{proposition}\label{pro:geo-con-ent}
    The geometric R\'enyi conditional entropy $\widehat{S}_{\alpha}$, for $\alpha\in(1,2]$, of an arbitrary finite-dimensional bipartite channel $\mathcal{N}_{A'B'\to AB}$ is,
\begin{align}\label{eq:geo-con-ent}
     \widehat{{S}}_{\alpha}[A|B]_{\mathcal{N}} &=  -\min_{\mathcal{M}\in\mathrm{Q}(B',B)}  \frac{1}{\alpha -1}\log \norm{\tr_{AB} G_{1-\alpha} \left(\Gamma^{\mathcal{N}}_{R_AAR_BB}, \mathbbm{1}_{R_AA}\otimes\Gamma^{\mathcal{M}}_{R_BB}\right)}_{\infty},
\end{align}
where $\Gamma^{\mathcal{T}}_{R_CC}:=\mathcal{T}_{C'\to C}(\Gamma_{R_CC'})=|C'|\Phi^{\mathcal{T}}$ denotes the Choi operator of a completely positive map $\mathcal{T}_{C'\to C}$ with $|R_C|=|C'|, $|C|$<+\infty$.
\end{proposition}
\begin{proof}
The proof follows by invoking ~\cite[Lemma 5]{FF21} in Definition Eq. ~\eqref{equ:geometric_renyi_div}: 
\begin{align}
    -\widehat{{S}}_{\alpha}[A|B]_{\mathcal{N}} &= \min_{\mathcal{M}\in\mathrm{Q}(B',B)}  \frac{1}{\alpha -1}\log \norm{\tr_{AB} G_{1-\alpha} \left(\Gamma^{\mathcal{N}}_{R_AAR_BB}, \Gamma^{\mathcal{R}\otimes\mathcal{M}}_{R_AAR_BB}\right)}_{\infty}\nonumber\\
    & = \min_{\mathcal{M}\in\mathrm{Q}(B',B)}  \frac{1}{\alpha -1}\log \norm{\tr_{AB} G_{1-\alpha} \left(\Gamma^{\mathcal{N}}_{R_AAR_BB}, \mathbbm{1}_{R_A}\otimes\mathbbm{1}_A\otimes\Gamma^{\mathcal{M}}_{R_BB}\right)}_{\infty}. \label{Eq:geom_cond_entr}
\end{align}
\end{proof}

 Due to \cite[Lemma 9]{FF21}, the geometric R\'enyi conditional entropy $\widehat{S}_{\alpha}$, for $\alpha\in(1,2]$, as in Eq.~\eqref{Eq:geom_cond_entr}, of an arbitrary finite-dimensional bipartite channel $\mathcal{N}_{A'B'\to AB}$ has a semidefinite representation as shown in the following corollary.

\begin{corollary}\label{cor:sdp_cond_geom}
For $\alpha= 1 + 2^{-\ell}$ with $\ell \in \mathbb{N}$, the geometric R\'enyi conditional entropy $\widehat{{S}}_{\alpha}[A|B]_{\mathcal{N}}$ takes the form $\widehat{{S}}_{\alpha}[A|B]_{\mathcal{N}}=-2^\ell \log y^*$, where $y^*$ is obtained from the following semidefinite programme:
\begin{align}
 &\mathrm{minimize:} \ y \nonumber \\
 &\mathrm{subject \  to:} \  P \in \mathsf{Herm}(R_AAR_BB), \{Q_i \in \mathsf{Herm}(R_AAR_BB)\}_{i=0}^\ell,  \nonumber \\
 &\qquad \qquad \quad    \begin{pmatrix}
P & \Gamma^\mathcal{N} \\
\Gamma^\mathcal{N} & Q_{\ell} \\
\end{pmatrix} \ge 0, \ \bigg\{\begin{pmatrix}
\Gamma^\mathcal{N} & Q_i \\
Q_i & Q_{i-1} \\
\end{pmatrix} \ge 0\bigg\}_{i=1}^\ell, \ \Gamma^{\mathcal{M}}_{R_BB} \ge 0,\  \tr_B \Gamma^{\mathcal{M}} = \mathbbm{1}_{R_B},   \nonumber \\
 & \qquad \qquad \quad Q_0 = \mathbbm{1}_{R_AA} \otimes \Gamma^{\mathcal{M}}_{R_BB}, \ \tr_{AB} P \le y \mathbbm{1}_{R_AR_B},
\end{align}
where $\mathsf{Herm}(R_AAR_BB)$ is the set of all Hermitian operators in $\mathrm{L}(R_AAR_BB)$.
\end{corollary}

\subsection{Conditional min and max entropies}
The max-relative entropy between a quantum channel $\mathcal{N}_{A'\to A}$, when $|A'|<\infty$, and a completely positive map $\mathcal{M}\in{\rm CP}(A',A)$ is given as~\cite[Lemma 12]{WBHK02}
\begin{align}
    D_{\max}[\mathcal{N} \parallel \mathcal{M}] &:=\sup _{\rho_{RA'}\in \mathrm{St}(RA')} D_{\max}(\id_R \otimes \mathcal{N}_{A'\to A}(\rho_{RA'}) \parallel \id_R \otimes \mathcal{M}_{A'\to A}(\rho_{RA'}))\nonumber\\
    & = D_{\max}(\Gamma^{\mathcal{N}}_{RA}\Vert\Gamma^{\mathcal{M}}_{RA}),\label{eq:max_channel_div}
\end{align}
the max-relative entropy $D_{\max}(\cdot\Vert\cdot)$~\eqref{eq:max-rel} between the Choi operators $\Gamma^{\mathcal{N}}, \Gamma^{\mathcal{M}}$ of linear maps $\mathcal{N},\mathcal{M}$, respectively. We recall that, for any $\rho\in\mathrm{St}(A'), \sigma\in\mathrm{Pos}(A')$,  
\begin{equation}
   \widetilde{D}_{\infty}(\rho\Vert\sigma):= \lim_{\alpha\to\infty}\widetilde{D}_{\alpha}(\rho\Vert\sigma)=D_{\max}(\rho\Vert\sigma),
\end{equation}
where $\widetilde{D}_{\alpha}(\rho\Vert\sigma)$ is the sandwiched R\'enyi relative entropy~\eqref{eq:def_sre}. 

\begin{dfn}[Conditional min-entropy]\label{def:max-con}
    The conditional min-entropy ${S}_{\min}[A|B]_{\mathcal{N}}$ of an arbitrary bipartite quantum channel $\mathcal{N}_{A'B'\to AB}$ is defined as
    \begin{equation}
       {S}_{\min}[A|B]_{\mathcal{N}}\coloneqq \widetilde{S}_{\infty}[A|B]_{\mathcal{N}}= - \inf_{\mathcal{M}\in\mathrm{Q}(B',B)} {D}_{\max}[\mathcal{N}_{A'B'\to AB}\Vert \mathcal{R}_{A'\to A}\otimes\mathcal{M}_{B'\to B}], \label{eq:cond_min_entr}
    \end{equation}
    where $D_{\max}[\cdot\Vert\cdot]$ is the max-relative entropy between completely positive maps as in Eq.~\eqref{eq:max_channel_div}.
\end{dfn}

\begin{proposition}
    The conditional min-entropy of an arbitrary finite-dimensional quantum channel $\mathcal{N}_{A'B'\to AB}$ is
    \begin{equation}
        {S}_{\min}[A|B]_{\mathcal{N}}=-\log\min_{\mathcal{M}\in\mathrm{Q}(B', B)}\norm{\mathbbm{1}_{R_AA}\otimes(\Gamma^{\mathcal{M}}_{R_BB})^{-\frac{1}{2}}\Gamma^{\mathcal{N}}_{R_AAR_BB}\mathbbm{1}_{R_AA}\otimes(\Gamma^{\mathcal{M}}_{R_BB})^{-\frac{1}{2}}}_{\infty},
    \end{equation}
    which has a semidefinite representation and can be efficiently computed using a semidefinite programme.
\end{proposition}
\begin{proof}
    Following Eq.~\eqref{eq:max_channel_div} and Definition~\ref{def:max-con}, we have
    \begin{align}
         {S}_{\min}[A|B]_{\mathcal{N}} = & -\log\min_{\mathcal{M}\in\mathrm{Q}(B', B)}\norm{(\Gamma^{\mathcal{R}\otimes\mathcal{M}}_{R_AAR_BB})^{-\frac{1}{2}}\Gamma^{\mathcal{N}}_{R_AAR_BB}(\Gamma^{\mathcal{R}\otimes\mathcal{M}}_{R_AAR_BB})^{-\frac{1}{2}}}_{\infty}\nonumber\\
         =         & -\log\min_{\mathcal{M}\in\mathrm{Q}(B', B)}\norm{\mathbbm{1}_{R_AA}\otimes(\Gamma^{\mathcal{M}}_{R_BB})^{-\frac{1}{2}}\Gamma^{\mathcal{N}}_{R_AAR_BB}\mathbbm{1}_{R_AA}\otimes(\Gamma^{\mathcal{M}}_{R_BB})^{-\frac{1}{2}}}_{\infty}. 
    \end{align}
\end{proof}
The conditional min-entropy of a bipartite channel has a semidefinite representation, as shown in the following corollary.
\begin{corollary} \label{cor:sdp_cond_min}
The conditional min-entropy can be estimated via the following semidefinite program:
\begin{align}
  S_{\min}[A|B]_{\mathcal{N}} =  -\log \min \bigg\{ \frac{\tr (\widetilde{M})}{ \abs{B'}} \ | \  \Gamma ^{\mathcal{N}} \le \mathbbm{1}_{R_AA} \otimes \widetilde{M}_{R_BB}, \ \widetilde{M} \ge 0, \ \tr_B\widetilde{M}=\tr(\widetilde{M}) \frac{\mathbbm{1}_{R_B}} { \abs{B'}} \bigg\}. 
\end{align}    
\end{corollary}
\begin{proof}
    As a consequence of Eq.~\eqref{eq:max_channel_div} and Eq.~\eqref{eq:cond_min_entr}, we have 

        \begin{align}
     S_{\min}[A|B]_{\mathcal{N}}&=-\inf_{\mathcal{M}\in Q(B',B)} D_{\max}(\Gamma^{\mathcal{N}}_{R_AAR_BB}\parallel \mathbbm{1}_{R_AA}\otimes \Gamma^{\mathcal{M}}_{R_BB}). \end{align}
    
From the semidefinite representation of the max-relative entropy as shown in Eq.~\eqref{eq:max-rel}, we have 
\begin{align}
     &\inf_{\mathcal{M}\in Q(B',B)}D_{\max}(\Gamma^{\mathcal{N}}_{R_AAR_BB}\parallel \mathbbm{1}_{R_AA}\otimes \Gamma^{\mathcal{M}}_{R_BB}) \nonumber \\
     &= \log \min \{y \  | \ \Gamma^{\mathcal{N}} \le y (\mathbbm{1}_{R_AA} \otimes \Gamma^{\mathcal{M}}_{R_BB}),\  \Gamma^{\mathcal{M}} \ge 0, \ \tr_B \Gamma^{\mathcal{M}} = \mathbbm{1}_{R_B}\}.
\end{align}
With this we prove the corollary by defining $\widetilde{M}:=y \Gamma^{\mathcal{M}}$ and noting $y= \tr(\widetilde{M})/ \abs{B'}$.
\end{proof}

\begin{figure}
    \centering
    \includegraphics[width=\columnwidth]{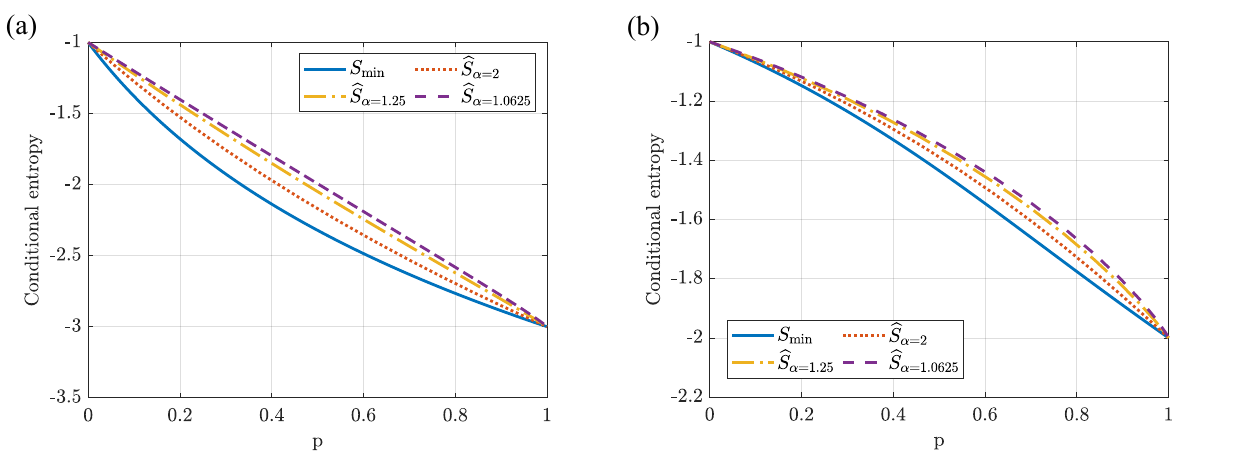}
    \caption{We plot the conditional entropies $\widehat{S}_{\alpha}[A|B]_{\mathcal{S}^{p,U}}$ and ${S}_{\min}[A|B]_{\mathcal{S}^{p,U}}$ of bipartite channels $\mathcal{S}^{p,U}$ for $U\in\{{\rm SWAP,CNOT}\}$ vs channels' mixture parameter $p\in [0,1]$. The dashed and dashed-dotted and dotted lines represent the geometric R\'enyi conditional entropies with $\alpha=1.0625$, $\alpha=1.25$, and $\alpha =2$, respectively, all of which are estimated using Corollary~\ref{cor:sdp_cond_geom}. The solid lines represent the conditional min-entropy $S_{\min}[A|B]_{\mathcal{S}^{p,U}}$ estimated using Corollary~\ref{cor:sdp_cond_min}.  Fig.(a) represents the mixture of SWAP and the identity gates (Eq.~\eqref{eq:sdp_SWAP}) and Fig.(b) represents the mixture of CNOT and the identity gates (Eq.~\eqref{eq:sdp_CNOT}).}
    \label{fig:SDP-conditional_entr}
\end{figure}

Using the semidefinite representations in Corollaries~\ref{cor:sdp_cond_geom} and ~\ref{cor:sdp_cond_min}, we estimate the geometric R\'enyi conditional entropy and the conditional min-entropy of the probabilistic mixtures of the two-qubit SWAP and identity gates and the two-qubit CNOT and identity gates:
\begin{align}
&\mathcal{S}^{p,\mathrm{SWAP}}_{A'B'\to AB}(\cdot) = p \mathrm{SWAP} (\cdot) \mathrm{SWAP}^\dag + (1-p)\id (\cdot) , \quad p\in[0,1], \label{eq:sdp_SWAP}  \\
&\mathcal{S}^{p,\mathrm{CNOT}}_{A'B'\to AB}(\cdot) = p \mathrm{CNOT} (\cdot) \mathrm{CNOT}^\dag + (1-p)\id(\cdot), \quad p\in[0,1]. \label{eq:sdp_CNOT}
\end{align}
In Fig.~\ref{fig:SDP-conditional_entr}, we plot the conditional entropies $\widehat{S}_{\alpha}[A|B]_{\mathcal{S}^{p,U}}$ and ${S}_{\min}[A|B]_{\mathcal{S}^{p,U}}$ of the channels in Eqs.~\eqref{eq:sdp_SWAP} and \eqref{eq:sdp_CNOT} with respect to the channels' mixing parameter $p$. For the case of the geometric R\'enyi conditional entropies $\widehat{S}_{\alpha}[A|B]_{\mathcal{S}^{p,U}}$ with $U\in\{{\rm SWAP,CNOT}\}$, we choose $\alpha=2$, $\alpha = 1.25$, and $\alpha=1.0625$, which correspond to $\ell=0$, $\ell=2$ and $\ell=4$, respectively, as in Corollary~\ref{cor:sdp_cond_geom}. Fig.~\ref{fig:SDP-conditional_entr}(a) plots for the mixture of SWAP and the identity (Eq.~\eqref{eq:sdp_SWAP}) and Fig.~\ref{fig:SDP-conditional_entr}(b) plots for the mixture of CNOT and the identity (Eq.~\eqref{eq:sdp_CNOT}). The plots show $S_{\min}[A|B]_{\mathcal{S}^{p,U}} \le \widehat{S}_{\alpha_1}[A|B]_{\mathcal{S}^{p,U}} \le \widehat{S}_{\alpha_2}[A|B]_{\mathcal{S}^{p,U}}$ for $\alpha_1 \ge \alpha_2$ in accordance with Eq.~\eqref{eq:gen-con-order}.

\begin{dfn}[Conditional max-entropy]\label{def:max-con_ent}
    The conditional max-entropy ${S}_{\max}[A|B]_{\mathcal{N}}$ of an arbitrary bipartite quantum channel $\mathcal{N}_{A'B'\to AB}$ is defined as
    \begin{equation}
       {S}_{\max}[A|B]_{\mathcal{N}} = - \inf_{\mathcal{M}\in\mathrm{Q}(B',B)} D_{\min}[\mathcal{N}_{A'B'\to AB}\Vert \mathcal{R}_{A'\to A}\otimes\mathcal{M}_{B'\to B}],
    \end{equation}
    where $D_{\min}[\mathcal{N} \parallel \mathcal{M}] :=\sup _{\psi_{RA'}\in \mathrm{St}(RA')} D_{\min}(\id_R \otimes \mathcal{N}_{A'\to A}(\psi_{RA'}) \parallel \id_R \otimes \mathcal{M}_{A'\to A}(\psi_{RA'}))$ is the min-relative entropy between $\mathcal{N},\mathcal{M}\in{\rm CP}(A',A)$.
\end{dfn}

\section{Bipartite processes: Signaling vs Semicausal} \label{sec:signaling_vs_semicausal}
A quantum channel $\mathcal{N}_{A'B' \to AB}$ is semicausal from $A'$ to $B$, if any local intervention on the input system $A'$ does not affect the output quantum state at $B$: signaling from $A'$ to $B$ is impossible. From now on, we refer to the set of such semicausal channels as $\rm{NS}_{A' \not \to B}$. A quantum channel $\mathcal{N}_{A'B' \to AB} \in \rm{NS}_{A' \not \to B}$ satisfies the following~\cite{BGNP01,PHHH06}:

\begin{align}
    \tr_{A} \circ \mathcal{N}_{A'B'\to AB}  = \tr_{A'} \otimes \widetilde{\mathcal{N}}_{B'\to B}. \label{Eq:semicausal}
\end{align}
See Fig.~\ref{fig:semicausal_channel}(a) for a pictorial depiction. Here $\widetilde{\mathcal{N}} \in Q(B',B)$ is the reduced quantum channel satisfying $\widetilde{\mathcal{N}}_{B'\to B}(\cdot)=\tr_{A}\circ \mathcal{N}_{A'B'\to AB}\left(\rho_{A'} \otimes (\cdot)_{B'}\right)$ for all $\rho \in \mathrm{St}(A')$. All semicausal channels in $\rm{NS}_{A' \not \to B}$ has the semilocalizable structure~\cite{Eggeling_2002}: $\mathcal{N}_{A'B'\to AB} = (\mathcal{N^A} \otimes \id_B) \circ (\id_{A'} \otimes \mathcal{N^B})$, with $\mathcal{N^A}\in Q(A'M, A)$ and $\mathcal{N^B}\in Q(B', MB)$. See Fig.~\ref{fig:semicausal_channel}(b). The reduced channel $\widetilde{N} \in Q(B',B)$ is then given as $\widetilde{\mathcal{N}}_{B'\to B}= \tr_M \circ \mathcal{N^B}_{MB'\to B}$. Note Eq.~\eqref{Eq:semicausal} can be rewritten in terms of $\mathcal{R}$ map:
\begin{align}
    \mathcal{R}_{A\to A} \circ \mathcal{N}_{A'B'\to AB}  = \mathcal{R}_{A'\to A} \otimes \widetilde{\mathcal{N}}_{B'\to B}, \qquad \ \forall \mathcal{N}_{A'B'\to AB} \in \rm{NS}_{A' \not \to B}. \label{Eq:semicausal2}
\end{align}
Clearly, for semicausal unitary channels in $\mathrm{NS}_{A'\not \to B}$, $\mathcal{N^A}$ and $\mathcal{N^B}$ have to be unitary channels~\cite{SW05,LB21}. In this scenario, we have the dimensional constraints $\abs{A'}\abs{M}= \abs{A}$, and $\abs{B'}=\abs{M}\abs{B}$. Particularly, when $A' \simeq A$ and $B' \simeq B$, we have $\abs{M}=1$, which indicates the unitary channels must be in the product form \cite{BGNP01, PHHH06}.

Let us now explore the properties of the semicausal quantum channel's conditional von Neumann entropy. We only present the case for $\mathcal{N}_{A'B'\to AB} \in \rm{NS}_{A' \not \to B}$ as the results for $\rm{NS}_{B' \not \to A}$ are similar. 

\begin{figure}
    \centering
    \includegraphics[width=\columnwidth]{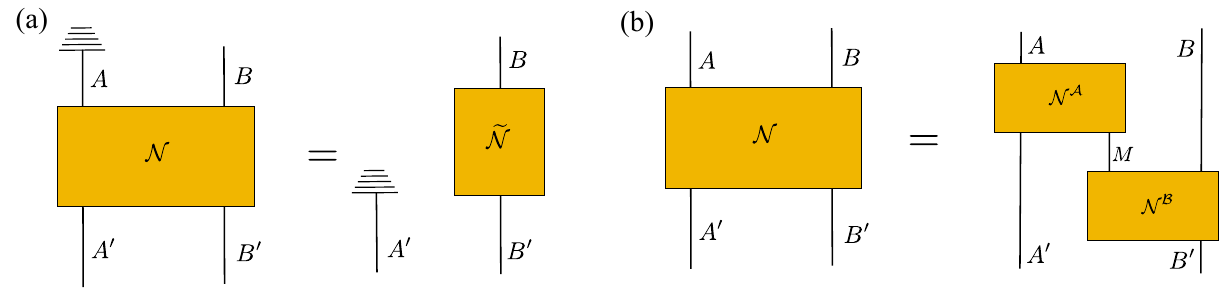}
    \caption{A semicausal channel $\mathcal{N}_{A'B'\to AB}$ with no-signaling constraint from $A'$ to $B$ ($\rm{NS}_{A' \not \to B}$). Time is flowing from down to up. In Fig.(a), the inverted earth symbol represents tracing out the relevant system. Fig.(a) represents Eq.~\eqref{Eq:semicausal}. Fig.(b) shows the semilocalizable structure of any channel in $\rm{NS}_{A' \not \to B}$.    } 
    \label{fig:semicausal_channel}
\end{figure}

\begin{theorem}~\label{Thm:semicausal}
    For a quantum channel $\mathcal{N}_{A'B'\to AB}$, its von Neumann conditional entropy $S[A|B]_{\mathcal{N}}$ satisfies the following identities,
    \begin{align}
        {S}[A|B]_{\mathcal{N}}  &  = - {D}[\mathcal{N}_{A'B' \to AB}\Vert\mathcal{R}_{A \to A}\circ \mathcal{N}_{A'B' \to AB}]\label{eq:ns-rel-ent}\\
       & = \inf_{\op{\psi} \in \mathrm{St}(RA'B')} {S}(A|RB)_{\mathcal{N}(\op{\psi})}, \label{eq:ns-con-ent}
    \end{align}
    if and only if $\mathcal{N}$ is no-signaling from $A'$ to $B$, i.e., $\mathcal{N}_{A'B'\to AB} \in \rm{NS}_{A' \not \to B}$.

    A bipartite quantum channel $\mathcal{N}_{A'B'\to AB}$ is signaling from $A'$ to $B$, i.e., there is causal influence from $A'$ to $B$ ($\mathcal{N}\notin{\rm NS}_{A'\not\to B})$, if and only if
    \begin{equation}\label{eq:ns_upperbound}
        S[A|B]_{\mathcal{N}}< \inf_{\op{\psi}\in{\rm St}(RA'B')}S(A|RB)_{\mathcal{N}(\op{\psi})}.
    \end{equation}
\end{theorem}
\begin{proof}
    Let us first go through the following chain of equations.
    \begin{align}
        &-S[A|B]_{\mathcal{N}} \numeq{1} \inf_{\mathcal{M}\in Q(B',B)} D[\mathcal{N}_{A'B'\to AB}\parallel \mathcal{R}_{A' \to A} \otimes \mathcal{M}_{B'\to B}] \nonumber \\
        &\numeq{2}\inf_{\mathcal{M}\in Q(B',B)} \sup _{\psi_{RA'B'} \in \mathrm{St}(RA'B')} D\left(\rho_{RAB} \parallel \mathbbm{1}_{A} \otimes (\id_R \otimes \mathcal{M})(\psi_{RB'})\right)\nonumber \\
        &\numeq{3}\inf_{\mathcal{M}\in Q(B',B)} \sup _{\psi_{RA'B'} \in \mathrm{St}(RA'B')} \bigg(D(\rho_{RAB} \parallel \mathbbm{1}_{A} \otimes \rho_{RB}) +  D\left((\id_R \otimes \widetilde{\mathcal{N}})(\psi_{RB'}) \parallel (\id_R \otimes \mathcal{M})(\psi_{RB'})\right) \bigg ) \nonumber \\
        &\numeq{4}\sup _{\psi_{RA'B'} \in \mathrm{St}(RA'B')} D\left(\rho_{RAB} \parallel \mathbbm{1}_{A} \otimes \rho_{RB}\right)  \nonumber \\ 
        &\numeq{5}D\left[\mathcal{N}_{A'B'\to AB} \parallel \mathcal{R}_{A\to A} \circ \mathcal{N}_{A'B'\to AB}\right]. \label{Eq:thm:semicausal}
    \end{align}
    Here, $\rho_{RAB}:=(\id_R \otimes \mathcal{N}_{A'B'\to AB} ) (\psi_{RA'B'})$. In the third equality, we have used $D(\alpha_{XY}\parallel \mathbbm{1}_X \otimes \beta_Y) = D(\alpha_{XY}\parallel \mathbbm{1}_X \otimes \alpha_Y) + D(\alpha_Y \parallel \beta_Y)$ for all $\alpha \in \mathrm{St}(XY)$, $\beta \in \mathrm{St}(Y)$. Further we have invoked Eq.~\eqref{Eq:semicausal} to write $\rho_{RB}=(\id_R \otimes \widetilde{\mathcal{N}})(\psi_{RB'})$. To obtain the fourth equality we note:

    \begin{align}
       \inf_{\mathcal{M}\in Q(B',B)}  D\left((\id_R \otimes \widetilde{\mathcal{N}})(\psi_{RB'}) \parallel (\id_R \otimes \mathcal{M})(\psi_{RB'})\right) = 0, \qquad \forall \psi_{RA'B'} \in \mathrm{St}(RA'B'). \label{Eq:semicausal_infimum_channel}
    \end{align}
    The infimum is achieved when $\widetilde{\mathcal{N}} = \mathcal{M}$. Now consider the following inequalities:

    \begin{align}
       &\inf_{\mathcal{M}\in Q(B',B)} \sup _{\psi_{RA'B'} \in \mathrm{St}(RA'B')} \bigg(D(\rho_{RAB} \parallel \mathbbm{1}_{A} \otimes \rho_{RB}) +  D\left((\id_R \otimes \widetilde{\mathcal{N}})(\psi_{RB'}) \parallel (\id_R \otimes \mathcal{M})(\psi_{RB'})\right) \bigg ) \nonumber \\
       &\le \sup _{\psi_{RA'B'} \in \mathrm{St}(RA'B')} D\left(\rho_{RAB} \parallel \mathbbm{1}_{A} \otimes \rho_{RB}\right) + \inf_{\mathcal{M}\in Q(B',B)} \sup _{\psi_{RA'B'} \in \mathrm{St}(RA'B')} D\left((\id_R \otimes \widetilde{\mathcal{N}})(\psi_{RB'}) \parallel (\id_R \otimes \mathcal{M})(\psi_{RB'})\right) \nonumber \\
       &=\sup _{\psi_{RA'B'} \in \mathrm{St}(RA'B')} D\left(\rho_{RAB} \parallel \mathbbm{1}_{A} \otimes \rho_{RB}\right). \label{Eq:semicausal_supremum}
    \end{align}
    The first inequality is because we take the supremum over the individual summands. As the first summand is independent of the channel $\mathcal{M}$, there is no infimum over it. In the second equation, we use Eq.~\eqref{Eq:semicausal_infimum_channel}. To show the inequality in the other direction, we first use the max–min inequality and then follow similar steps:

    \begin{align}
     &\inf_{\mathcal{M}\in Q(B',B)} \sup _{\psi_{RA'B'} \in \mathrm{St}(RA'B')} \bigg(D(\rho_{RAB} \parallel \mathbbm{1}_{A} \otimes \rho_{RB}) +  D\left((\id_R \otimes \widetilde{\mathcal{N}})(\psi_{RB'}) \parallel (\id_R \otimes \mathcal{M})(\psi_{RB'})\right) \bigg ) \nonumber \\
     &\ge \sup _{\psi_{RA'B'} \in \mathrm{St}(RA'B')} \inf_{\mathcal{M}\in Q(B',B)}  \bigg(D(\rho_{RAB} \parallel \mathbbm{1}_{A} \otimes \rho_{RB}) +  D\left((\id_R \otimes \widetilde{\mathcal{N}})(\psi_{RB'}) \parallel (\id_R \otimes \mathcal{M})(\psi_{RB'})\right) \bigg ) \nonumber \\
     &=\sup _{\psi_{RA'B'} \in \mathrm{St}(RA'B')} D\left(\rho_{RAB} \parallel \mathbbm{1}_{A} \otimes \rho_{RB}\right). \label{Eq:semicausal_supremum2}
    \end{align}
    Thus Eqs.~\eqref{Eq:semicausal_supremum},~\eqref{Eq:semicausal_supremum2} ensure the equality sign. Finally, we obtain Eq.~\eqref{Eq:thm:semicausal} from the definition of the channel divergence. Rearranging the sign, we prove Eq.~\eqref{eq:ns-rel-ent}. 
    
    For Eq.~\eqref{eq:ns-con-ent}, we start from the fourth equality in the above chain of equations and note that $D(\rho_{RAB}\parallel \mathbbm{1}_{A} \otimes \rho_{RB})=-S(A|RB)_{\rho}$. Since the chain of equations is guaranteed to be saturated, i.e., equalities hold, we have an ``if and only if" statement. We now apply Proposition~\ref{prop:vN-upperbound} to arrive at inequality~\eqref{eq:ns_upperbound}.
\end{proof}

The following corollary is a direct consequence of the above theorem.
\begin{corollary} \label{Cor:semicausal}
          All semicausal channels $\mathcal{N}_{A'B'\to AB} \in{\rm{NS}_{A'\not \to B}}$ satisfy  ${S}[A|B]_{\mathcal{N}} \geq -\log\abs{A}$. The inequality is saturated for the semicausal unitary channels. The presence of causal influence from $A'\to B$ for a bipartite channel $\mathcal{N}_{A'B'\to AB}$ is guaranteed if $S[A|B]_\mathcal{N} < -\log {\abs{A}}$.
\end{corollary}

\begin{proof}
    This is due to Eq.~\eqref{eq:ns-con-ent} in Theorem~\ref{Thm:semicausal} and the lower bound $S(A|BR)_\rho \ge -\log \abs{A}$. The inequality is saturated for semicausal unitary channels, $\mathcal{U}_{A'B'\to AB} \in \rm{NS}_{A'\not \to B}$, with the initial state 
    \begin{align}
        \psi_{RA'B'}= (\id_R \otimes \mathcal{U^\dagger}_{AB\to A'B'})(\Phi_{RAB}),
    \end{align}
    where $\Phi_{RAB}$ is a maximally entangled state in the bipartition $A:BR$.
\end{proof}

A generalization of  Corollary~\ref{Cor:semicausal} is possible using a tele-covariant channel $\mathcal{N}_{A'B'\to AB}$, which is $k$-extendible with respect to the systems $A'$ and $A$ \cite{KDWW21}. Such a channel, in the limiting case of $k\to \infty$, becomes a one-way LOCC channel with classical communication from $B'\to A$, i.e., the channel $\mathcal{N}_{A'B'\to AB}$ is semicausal: $\mathcal{N}_{A'B'\to AB} \in \mathrm{NS}_{A'\to B}$. We first define a $k$-extendible channel below and then propose our generalization in Lemma~\ref{Lem:k_extendible}.

\begin{dfn}\label{Defn:k_extendibility}
    \cite{KDWW21} A bipartite channel $\mathcal{N}_{A'B'\to AB}$ is $k$-extendible for a natural number $k\geq 2$ with respect to the system $A', A$, if there exists a quantum channel $\mathcal{M} \in {\rm Q}(A'_1A'_2\cdots A'_kB', A_1A_2\cdots A_kB)$, where $A' {\simeq} A'_1{\simeq} A'_2 {\simeq \cdots \simeq} A'_k$ and $A {\simeq} A_1{\simeq} A_2 {\simeq \cdots \simeq} A_k$, having the following properties:
    \begin{enumerate}
        \item The channel $\mathcal{M}$ is permutation covariant with respect to the $A$ system, i.e., for all $g\in S_k$, where $S_k$ is the permutation group, we have
        \begin{align}
            \mathcal{M}_{A'_1A'_2\cdots A'_kB' \to A_1A_2\cdots A_kB}\circ \mathcal{U}^g_{A'_1A'_2\cdots A'_k}  = \mathcal{U}^g_{A_1A_2\cdots A_k} \circ \mathcal{M}_{A'_1A'_2\cdots A'_kB' \to A_1A_2\cdots A_kB}.
        \end{align}
        Here $\mathcal{U}^g$ is the unitary representation of the element $g\in S_k$.

        \item The channel $\mathcal{N}_{A'B'\to AB}$ is the marginal of $\mathcal{M}$ in the sense

        \begin{align}
           \mathcal{N}_{A'B'\to AB} \otimes  \Tr_{A'_2\cdots A'_k}= \Tr_{A_2\cdots A_k} \circ \mathcal{M}_{A'_1A'_2\cdots A'_kB' \to A_1A_2\cdots A_kB}. 
        \end{align}
    \end{enumerate}
             We call $\mathcal{M}$ a $k$-extension of $\mathcal{N}$. Note that comparing with Eq.~\eqref{Eq:semicausal}, we observe the following no-signaling constraints hold: input systems $A'_i$ cannot signal to $A_j$ for $i\neq j$. 
\end{dfn}

\begin{lemma}\label{Lem:k_extendible}
    [Conditional entropy for $k$-extendible tele-covariant channel] If a tele-covariant bipartite channel $\mathcal{N}_{A'B'\to AB}$ is $k$-extendible with respect to systems $A',A$, as in Definition~\ref{Defn:k_extendibility}, then its von Neumann conditional entropy is lower bounded as
    \begin{align}
        S[A|B]_{\mathcal{N}} \ge - \min\bigg\{\frac{2}{k}\log \abs{B} - \log \abs{A}, \frac{1}{k}\log \abs{B'} \abs{B}  - \log \abs{A'}\bigg\}.
    \end{align}
\end{lemma}

\begin{proof}
   We use the following chain of equations, where we use $\Pi_{i=1}^k \abs{A_i}=\abs{A}^k$ and $\Pi_{i=1}^k \abs{A'_i}=\abs{A'}^k$. Also, note $\mathcal{M}$ is a $k$-extension of $\mathcal{N}$.
    \begin{align}
        -\log \min\{\abs{A}^k\abs{B}^2, \abs{A'}^k\abs{B'}\abs{B}\} &\numleq{1} S[A_1A_2\cdots A_k|B]_{\mathcal{M}} \nonumber\\&\numleq{2} S(R_{A_1}A_1 R_{A_2}A_2 \cdots R_{A_k}A_k| R_B B)_{\Phi^{\mathcal{M}}} - k \log \abs{A'} \nonumber \\
        &\numleq{3}  \sum_{i=1}^k S(R_{A_i}A_i|R_BB)_{\Phi^{\mathcal{M}}} - k \log\abs{A'} \nonumber \\
        &\numeq{4} k S(R_AA|R_BB)_{\Phi^\mathcal{N}} - k \log\abs{A'} \nonumber 
        \\
        &\numeq{5}kS[A|B]_{\mathcal{N}}.
    \end{align}
     The first inequality is due to the general lower bound as in Proposition~\ref{prop:lowerbound_cond_entropy}. The second inequality is due to the fact that, in general, the initial maximally entangled state is suboptimal to evaluate the conditional entropy of a channel, see Eq.~\eqref{eq:thm-proof-con}. The third inequality is due to the subadditivity of the von Neumann entropy. In the fourth and fifth equality, we use the fact that $\mathcal{N}$ is a tele-covariant channel. Finally, the statement in the lemma follows by rearranging the l.h.s. and r.h.s. of the above inequality, 
\end{proof}

\subsection{Generalized conditional NS-entropy}
\textit{Generalized conditional NS-entropy}.--- Let $S^{\not\to}[A|B]_{\mathcal{N}}$ denote the function in the r.h.s. of Eq.~\eqref{eq:ns-rel-ent} for an arbitrary bipartite quantum channel $\mathcal{N}_{A'B'\to AB}$. This is the candidate A in Eq.~\eqref{eq:CA}. Its generalized form can be denoted as $\mathbf{S}^{\not\to}[A|B]_{\mathcal{N}}$,
\begin{equation}
    \mathbf{S}^{\not\to}[A|B]_{\mathcal{N}}:=\mathbf{D}[\mathcal{N}_{A'B'\to AB}\Vert\mathcal{R}_{A\to A}\circ\mathcal{N}_{A'B'\to AB}].
\end{equation}
$\mathbf{S}^{\not\to}[A|B]_{\mathcal{N}}$ satisfies all the properties in Theorem~\ref{thm:axioms}. Therefore, it can be deemed a valid conditional entropy of a bipartite channel. We refer to $\mathbf{S}^{\not\to}[A|B]_{\mathcal{N}}$ as the generalized conditional NS-entropy, where NS denotes connection to no-signaling. However, it is not our first preference for the definition of the conditional entropy when compared to $\mathbf{S}[A|B]_{\mathcal{N}}$.

For the von Neumann conditional NS-entropy $S^{\not\to}[A|B]_{\mathcal{N}}$ of a bipartite channel $\mathcal{N}_{A'B'\to AB}$, we have $S^{\not\to}[A|B]_{\mathcal{N}}=\inf_{\psi\in{\rm St}(RA'B')}S(A|RB)_{\mathcal{N}(\psi)}$, where it suffices to optimize over pure states $\psi_{RA'B'}$; the dimensional bounds
\begin{equation}
    -\log|A|\leq S^{\not\to}[A|B]_{\mathcal{N}}\leq \log|A|
\end{equation}
resemble that of the von Neumann conditional entropy of quantum states and the von Neumann entropy of a channel $\mathcal{N}_{A'B'\to A}=\tr_B\circ\mathcal{N}_{A'B'\to AB}$, see Table~\ref{tab:bounds}. 

Let $\mathcal{V}^{\mathcal{N}}_{A'B'\to ABE}$ denote an isometric extension channel of $\mathcal{N}_{A'B'\to AB}$, i.e.,
\begin{equation}
    \tr_E\circ\mathcal{V}^\mathcal{N}_{A'B'\to ABE}=\mathcal{N}_{A'B'\to AB}.
\end{equation}
A complementary channel $\mathcal{N}^c_{A'B'\to E}$ of a bipartite channel $\mathcal{N}_{A'B'\to AB}$ is given as
\begin{equation}
    \mathcal{N}^c_{A'B'\to E}=\tr_{AB}\circ\mathcal{V}^{\mathcal{N}}_{A'B'\to ABE}.
\end{equation}

 The von Neumann conditional NS-entropy obeys the following relation,
\begin{equation}
   {S}^{\not\to}[A|B]_{\mathcal{N}}+\sup_{\rho\in{\rm St}(A'B')}S(A|E)_{\mathcal{V}^{\mathcal{N}}(\rho)}=0.
\end{equation}
It follows because of the duality relation of the von Neumann conditional entropy, $S(A|B)_{\phi}+S(A|E)_{\phi}=0$ for any pure state $\phi_{ABE}$. Then,
\begin{equation}
  \inf_{\op{\psi}\in{\rm St}(RA'B')}S(A|RB)_{\mathcal{N}(\op{\psi})} =\inf_{\rho\in{\rm St}(A'B')}-S(A|E)_{\mathcal{V}^{\mathcal{N}}(\rho)}=-\sup_{\rho\in{\rm St}(A'B')}S(A|E)_{\mathcal{V}^{\mathcal{N}}(\rho)}.
\end{equation}

\section{Operational meaning of conditional entropies} \label{sec:operational_meaning}
\subsection{Quantum hypothesis testing for channel discrimination}
The generalized conditional entropy $\mathbf{S}[A|B]_{\mathcal{N}}$ of a bipartite channel $\mathcal{N}_{A'B'\to AB}$ has an operational meaning in terms of the channel discrimination task~\cite{CMW16,FFRS20,FF21,BKSD23} based on the associated generalized divergence $\mathbf{D}[\cdot\Vert\cdot]$ between the given channel $\mathcal{N}_{A'B'\to AB}$ (null hypothesis) and some channel $\widetilde{\mathcal{R}}_{A'\to A}\otimes\mathcal{M}_{B'\to B}$ (alternate hypothesis) (cf.~\cite[Appendix D]{FF21}), where $\widetilde{\mathcal{R}}_{A'\to A}:=\frac{1}{|A|}\mathcal{R}_{A'\to A}$ is the completely mixing or depolarizing channel,
\begin{align}
    -\mathbf{S}[A|B]_{\mathcal{N}}&=\inf_{\mathcal{M}\in{\rm Q}(B',B)}\mathbf{D}[\mathcal{N}_{A'B'\to AB}\Vert\mathcal{R}_{A'\to A}\otimes\mathcal{M}_{B'\to B}] \nonumber\\
    &=\inf_{\mathcal{M}\in{\rm Q}(B',B)}\mathbf{D}[\mathcal{N}_{A'B'\to AB}\Vert|A|\widetilde{\mathcal{R}}_{A'\to A}\otimes\mathcal{M}_{B'\to B}].
\end{align}
Hence, we have $ \mathbf{S}[A|B]_{\mathcal{N}} =\log|A|- \inf_{\mathcal{M}\in{\rm Q}(B',B)}\mathbf{D}[\mathcal{N}_{A'B'\to AB}\Vert\widetilde{\mathcal{R}}_{A'\to A}\otimes\mathcal{M}_{B'\to B}]$ if $\mathbf{D}(\rho\Vert c\sigma )=-\log c+\mathbf{D}(\rho\Vert\sigma)$ for all scalars $c\in\mathbb{C}$.

\subsubsection{Parallel uses of channels for the channel discrimination task}
To elucidate, we consider a game with two different kind of adversarial challenges $\texttt{Quest}_x[\mathcal{N}]$, where $x\in\{a,b\}$. The game consists of a challenger, Charlie, and an adversary, Ash. Ash provides a black box $\mathcal{B}$, which is a bipartite quantum channel $\mathcal{B}^{y}$ ($y\in\{0,x\}$), to Charlie and the goal of Charlie is to discriminate if the given box indeed is the desired bipartite channel. In each challenge $x$, black box has only two possibilities $\{H_0,H_1^x\}$.
\begin{enumerate}
    \item[$H_0$] Null hypothesis: The desired bipartite channel is $\mathcal{B}_{A'B'\to AB}^0=\mathcal{N}_{A'B'\to AB}$.
    \item[$H_1^x$] Alternate hypothesis: The bipartite channel is, 
    \begin{itemize}
        \item[(a)] if $x=a$, $\mathcal{B}_{A'B'\to AB}^a=\widetilde{\mathcal{R}}_{A'\to A}\otimes\mathcal{M}_{B'\to B}$, where a channel $\mathcal{M}_{B'\to B}$ is chosen in the best  possible way to help adversary mimic $\mathcal{N}_{A'B'\to AB}$ given that the completely mixing channel $\widetilde{\mathcal{R}}_{A'\to A}$ is fixed.
        \item[(b)] if $x=b$, $\mathcal{B}_{A'B'\to AB}^b=\widetilde{\mathcal{R}}_{A\to A}\circ\mathcal{N}_{A'B'\to AB}$, that is, the output $A$ of the channel $\mathcal{N}_{A'B'\to AB}$ is postprocessed with the completely mixing channel.
    \end{itemize}
\end{enumerate}

In asymmetric quantum hypothesis testing, the goal is to minimize type-II error probability $\beta_{\varepsilon}$ given an upper bound on type-I error probability is $\varepsilon$~\cite{WR12}. The optimal rate for the single use of black box under $\texttt{Quest}_x[\mathcal{N}]$ is given as, for $\varepsilon\in(0,1)$,
\begin{align}
    D^{\varepsilon}_h[\mathcal{N}_{A'B'\to AB}\Vert\mathcal{B}^x_{A'B'\to AB}],\quad 
\end{align}
where $D^{\varepsilon}_h[\mathcal{N}\Vert\mathcal{B}]$ is the $\varepsilon$-hypothesis testing channel divergence:
\begin{equation}
D^{\varepsilon}_h[\mathcal{N}\Vert\mathcal{B}]:=\sup_{\psi\in{\rm St}(R'A'B')}D^{\varepsilon}_h(\mathcal{N}(\psi_{R'A'B'})\Vert\mathcal{B}(\psi_{R'A'B'})),
\end{equation}
obtained from the $\varepsilon$-hypothesis testing relative entropy $D^\varepsilon_h(\rho\Vert\sigma):=-\log\beta_{\varepsilon}(\rho\Vert\sigma)$ and $\beta_{\varepsilon}(\rho\Vert\sigma):=\inf_{0\leq \lambda\leq \mathbbm{1}}\{\tr[\lambda \sigma]: \tr[\lambda\rho]\geq 1-\varepsilon\}$. 

Let us assume that in these challenges $\texttt{Quest}^{\parallel}_x[\mathcal{N}]$, only parallel uses of the channels and product input states between different uses of channels are allowed. That is, if $n$ copies of the channel $\mathcal{B}_{A'B'\to AB}$ is provided in $\texttt{Quest}^{\parallel}_x[\mathcal{N}]$, then the challenger is allowed to input states of the form $\psi_{RA'B'}^{\otimes n}$. In the $i.i.d.$ asymmetric hypothesis testing setting, applying quantum Stein's lemma~\cite{HP91,ON00}, for all $\varepsilon\in(0,1)$,
\begin{equation}
  \lim_{n\to\infty}\frac{1}{n} D^{\varepsilon}_h(\mathcal{N}^{\otimes n}(\psi^{\otimes n}_{RA'B'})\Vert{\mathcal{B}_x}^{\otimes n}(\psi^{\otimes n}_{R'A'B'}))= D(\mathcal{N}(\psi_{R'A'B'})\Vert\mathcal{B}^x(\psi_{R'A'B'})),
\end{equation}
which implies,
\begin{equation}
     \lim_{n\to\infty}\sup_{\psi\in \mathrm{St}(R'A'B')}\frac{1}{n}D^{\varepsilon}_h(\mathcal{N}^{\otimes n}(\psi^{\otimes n}_{R'A'B'})\Vert{\mathcal{B}_x}^{\otimes n}(\psi^{\otimes n}_{R'A'B'}))= D[\mathcal{N}_{A'B'\to AB}\Vert\mathcal{B}^x_{A'B'\to AB}].
\end{equation}

\begin{observation}
    In the asymmetric hypothesis testing protocol for the challenge $\texttt{Quest}^{\parallel}_x[\mathcal{N}]$, $x\in\{a,b\}$, the optimal achievable and strong converse rate $R^{\parallel}_{x}[\mathcal{N}]$ is, 
\begin{enumerate}
    \item[(a)] when $x=a$,
    \begin{equation}
       R^{\parallel}_{a}[\mathcal{N}]= \inf_{\mathcal{M}\in{\rm Q}(B',B)}D[\mathcal{N}_{A'B'\to AB}\Vert\widetilde{R}_{A'\to A}\otimes\mathcal{M}_{B'\to B}]= - S[A|B]_{\mathcal{N}}+\log|A|.
    \end{equation}
     \item[(b)] when $x=b$,
    \begin{equation}
       R^{\parallel}_{b}[\mathcal{N}]= D[\mathcal{N}_{A'B'\to AB}\Vert\widetilde{R}_{A\to A}\circ\mathcal{N}_{A'B'\to AB}] = - S^{\not\to}[A|B]_{\mathcal{N}}+\log|A|.
            \end{equation}
        \end{enumerate}     
  
 If a quantum channel $\mathcal{N}_{A'B'\to AB}\in \rm{NS}_{A' \not \to B}$, then $ R^{\parallel}_{a}[\mathcal{N}]= R^{\parallel}_{b}[\mathcal{N}]$ always holds. In contrast, for the bipartite channels $\mathcal{N}_{A'B'\to AB}\notin \rm{NS}_{A' \not \to B}$ we have $R^{\parallel}_{a}[\mathcal{N}]> R^{\parallel}_{b}[\mathcal{N}]$.  
\end{observation}
\begin{remark}
 If the challenger is allowed to select the channel $\mathcal{B}^0_{A'B'\to AB}$ of their choice with a limitation that $|A'|=|A|=|B'|=|B|=d$ is fixed by the adversary in the challenges $\texttt{Quest}^{\parallel}_x[\mathcal{N}]$, $x\in\{a,b\}$. Then, there exists an optimal strategy, considering the null hypothesis to be the two-qu$d$it SWAP channel $\mathcal{S}_{A'B'\to AB}$, such that
    \begin{equation}
        R^{\parallel}_{a}[\mathcal{S}]- R^{\parallel}_{b}[\mathcal{S}]=2\log d,
    \end{equation}
    which can be made arbitrarily large by choosing $d$ accordingly.
\end{remark}

\subsubsection{Adaptive strategy for channel discrimination} \label{subsec:adaptive_stratgegy}
\textit{Asymmetric hypothesis testing for channel discrimination using adaptive strategy}.--- Let $\{\Lambda,\mathcal{A}\}$ denote an arbitrary adaptive protocol for discriminating quantum channels $\mathcal{N}^1,\mathcal{N}^2\in {\rm Q}(A',A)$, where $\Lambda$ denotes the measurement operator corresponding to the null hypothesis at the end of the protocol and $\mathcal{A}$ denotes the set of (adaptive) channels applied between the uses of $\mathcal{N}^1,\mathcal{N}^2$; see \cite[Section 3]{WBHK02} for the definitions of the protocol and also for the (slightly modified) notations used here. Let $\alpha_{n}\left(\{\Lambda,\mathcal{A}\}\right)$ and $\beta_{n}\left(\{\Lambda,\mathcal{A}\}\right)$ be the corresponding type-I and type-II error probabilities, respectively, where $n\in\mathbb{N}$ denotes the number of channel uses for $\mathcal{N}^1$ or $\mathcal{N}^2$ in the protocol. The asymmetric distiguishability, in Stein's setting, is defined as~\cite{WBHK02} 
 \begin{equation}
     \zeta_{n} \left(\varepsilon, \mathcal{N}^1,\mathcal{N}^2\right) := \sup_{\Lambda,\mathcal{A}} \bigg \{-\frac{1}{n}\log\beta_{n}\left(\{\Lambda,\mathcal{A}\}\right) \big | \alpha_{n}\left(\{\Lambda,\mathcal{A}\}\right) \leq \varepsilon\bigg\},
 \end{equation}
 and its asymptotic quantities are defined as
 \begin{equation}
     \overbar{\zeta} \left(\varepsilon, \mathcal{N}^1,\mathcal{N}^2\right) = \underset{n \rightarrow \infty}{\lim \sup} ~\zeta_{n} \left(\varepsilon, \mathcal{N}^1,\mathcal{N}^2\right),\  \text{and} \ \ubar{\zeta} \left(\varepsilon, \mathcal{N}^1,\mathcal{N}^2\right) = \underset{n \rightarrow \infty}{ \lim \inf}~ \zeta_{n} \left(\varepsilon, \mathcal{N}^1,\mathcal{N}^2\right).
 \end{equation}
 Let $\varepsilon \in (0,1)$ and $\alpha \in(1,2]$, then we have the following strong converse bound on $\overbar{\zeta} \left(\varepsilon, \mathcal{N}^1,\mathcal{N}^2\right)$ ~\cite[Theorem 49, Appendix D]{FF21} (also see ~\cite[Corollary 18]{WBHK02}),
 \begin{equation}\label{eq:channel_discrimination_bounds}
    D\left(\mathcal{N}^1\Vert \mathcal{N}^2\right) \leq \ubar{\zeta} \left(\varepsilon, \mathcal{N}^1,\mathcal{N}^2\right) \leq \overbar{\zeta} \left(\varepsilon, \mathcal{N}^1,\mathcal{N}^2\right) \leq\widehat{D}_{\alpha}\left(\mathcal{N}^1\Vert \mathcal{N}^2\right) \leq D_{\rm max} \left(\mathcal{N}^1\Vert \mathcal{N}^2\right).
 \end{equation}
The input state for an adaptive discrimination protocol can be arbitrary. The above bound works for any pair of quantum channels $\mathcal{N}^1, \mathcal{N}^2\in{\rm Q}(A',A)$.

 Applying inequalities~\eqref{eq:channel_discrimination_bounds}~and~\eqref{eq:gen-con-order}, for any bipartite quantum channel $\mathcal{N}_{A'B'\to AB}$,
 \begin{align}
     -S[A|B]_{\mathcal{N}}+\log |A| & \leq \inf_{\mathcal{M}\in{\rm Q}(B',B)}\ubar{\zeta} \left(\varepsilon, \mathcal{N}_{A'B'\to AB},\widetilde{\mathcal{R}}_{A'\to A}\otimes\mathcal{M}_{B'\to B}\right)\\
     &\leq \inf_{\mathcal{M}\in{\rm Q}(B',B)} \overbar{\zeta} \left(\varepsilon, \mathcal{N}_{A'B'\to AB},\widetilde{\mathcal{R}}_{A'\to A}\otimes\mathcal{M}_{B'\to B}\right)\\
     &\leq  -\widehat{S}_{\alpha}[A|B]_{\mathcal{N}}+\log|A| \quad\forall \alpha\in (1,2] \\
     &\leq  -S_{\rm min}[A|B]_{\mathcal{N}} +\log |A|.
 \end{align}

\subsection{Conditional quantum channel merging}
The conditional quantum entropy finds an information-theoretic meaning in quantum state merging~\cite{HOW05,Ber09}. Consider a tripartite  pure state $\psi_{RAB}$ where $A$ and $B$ are accessible to Ash and Bob, respectively, and neither one has access to the reference system $R$; Ash and Bob are correlated. The task of state merging is for Ash to transfer her share of state to Bob provided that the classical communication is free. The optimal entanglement cost in asymptotic limit for this task is the von Neumann conditional entropy $S(A|B)_{\psi}$~\cite{HOW06}. If $ S(A|B)_{\psi}$ is negative, then the state merging is possible and $-S(A|B)_{\psi}$ number of maximally entangled states can be left afterwards for future use~\cite{HOW06}. 

The von Neumman entropy $S[\mathcal{M}]$ of a point-to-point quantum channel $\mathcal{M}_{A'\to A}$ is equal to the quantum channel merging capacity~\cite{GW21}. The goal of quantum channel merging is to transfer the state of output $A$ to Eve possessing isomteric extension $E$ of the channel $\mathcal{M}$~\cite{GW21}, where $\mathcal{V}^{\mathcal{M}}_{A'\to AE}$ is an isometric channel extension of $\mathcal{M}$. It is a variant of a quantum state merging protocol. The entropy of a channel is equal to its completely bounded entropy $S_{\rm CB, min}$~\cite[Eq.~(1.1)]{DJKR06}. We note that the entropy of a channel is related to its private randomness capacity $P_{\rm random}(\mathcal{M})$ (up to a constant factor of log of the output dimension)~\cite[Theorem 6]{YHW19}. The private randomness capacity is the optimal rate, in an usual asymptotic $i.i.d.$ uses of channels, at which the receiver (Bob) accessing output $A$ of the channel $\mathcal{M}_{A'\to A}$ can extract \textit{private randomness}, against eavesdropper accessing isomteric extension $E$ of the channel ${\mathcal{N}}_{A'\to A}$, when the sender (Ash) sends states through the channel $\mathcal{N}$ (see~\cite{YHW19} for proper definition). In particular, we have the following complimentary relation between $P_{\rm random}(\mathcal{M})$ and $S[\mathcal{M}]$,
\begin{equation}\label{eq:trade-off-p-e}
    P_{\rm random}(\mathcal{M})+S[\mathcal{M}]=\log|A|.
\end{equation}
 The sum of the private randomness capacity $P_{\rm random}(\mathcal{M})$ and the entropy of a channel $S[A]_{\mathcal{M}}$ always equals to the logarithm of the output dimension of the channel. The private randomness capacity is maximum, $2\log|A|$, if and only if the channel is unitary and minimum, $0$, if and only if the channel is completely mixing (depolarizing).

We consider a variant of the quantum channel merging protocol, which we call conditional quantum channel merging protocol.

\textit{Conditional quantum channel merging.}--- Consider $\mathcal{V}^{\mathcal{N}}_{A'B'\rightarrow ABE}$ is an isometric channel extension of a quantum bipartite channel $\mathcal{N}_{A'B' \rightarrow AB}$. The isometric channel is a broadcasting channel with two sources (inputs) $A'$, $B'$ and three receivers (outputs) $A$ (Ash), $B$ (Bob), and $E$ (Eve). The task of Ash is to merge her share $A$ of the channel with Eve's share $E$ using forward local operation and classical communication (LOCC), also called one-way LOCC, and entanglement. Let $\psi_{RA'B'}$ be the state generated by the source, which is distributed to Ash, Bob, and Eve via isometric channel extension $\mathcal{V}^{\mathcal{N}}_{A'B'\rightarrow ABE}$. At any stage we can replace the systems under consideration with $n$  copies of itself and this will be equivalent to starting with $\psi_{RA'^{n}B'^{n}}$ which is transmitted via isometric channel extension $\left(\mathcal{V}^{\mathcal{N}}_{A'B'\rightarrow ABE} \right)^{\otimes n}$. 

Let $n,l,k \in \mathbb{N}$ and $\epsilon \in [0,1]$. Consider an one-way LOCC operation $\mathcal{L}_{A_{n}E_{n}A''_{0}E''_{0} \rightarrow \overbar{A}_{E}\overbar{E}^{n} A''_{1}E''_{1}}$ from Ash to Eve. In an $(n, M, \varepsilon)$-conditional channel merging protocol, where $M=l/k$, we demand
 \begin{align}
     \sup_{\psi_{R A'^{n} B'^{n}}} \bigg |\bigg| \mathrm{id}^{\otimes n}_{AE \rightarrow \overbar{A}_{E}\overbar{E}} & \circ \left(\mathcal{V}^{\mathcal{N}}_{A'B'\rightarrow ABE} \right)^{\otimes n} \left(\psi_{RA'^{n}B'^{n}} \otimes  \Phi^{l}_{A''_{1}E''_{1}}\right) \nonumber\\
     &-\mathcal{L}_{A_{n}E_{n}A''_{0}E''_{0} \rightarrow \overbar{A}_{E}\overbar{E}^{n} A''_{1}E''_{1} } \left(\left(\left(\mathcal{V}^{\mathcal{N}}_{A'B'\rightarrow ABE} \right)^{\otimes n} \psi_{RA'^{n}B'^{n}}\right) \otimes \Phi^{k}_{A''_{0}E''_{0}} \right)  \bigg|\bigg|_{1} \leq \varepsilon,
 \end{align}
 where $ \Phi^{l}_{A''_{1}E''_{1}}$ and  $\Phi^{k}_{A''_{0}E''_{0}}$ are maximally entangled states of Schmidt rank $l$ and $k$, respectively, and system pairs $A''_{0}, A'_{1}$ and $E''_{0}, E''_{1}$ are auxiliary systems of Ash and Eve, respectively. The entanglement gain of the protocol is $\log M := \log l -\log k $ and $R := \frac{1}{n}{\log M}$ is called the entanglement (gain) rate of the protocol. $R$ is called an achievable rate if there exists a protocol $ (n, 2^{n(R-\delta)}, \varepsilon) ~~\forall~~\epsilon\in (0,1],\delta > 0$ for sufficiently large $n$. The conditional quantum channel merging capacity $C_{\rm gain}(\mathcal{N}_{A'B' \rightarrow AB})$ is defined as 
 \begin{equation}
     C_{\rm gain}[A|B]_{\mathcal{N}} \equiv \sup \{R \vert ~R\textrm{ is achievable for conditional channel merging on $\mathcal{N}$}\}.
 \end{equation}

\begin{proposition}
    The conditional quantum channel merging capacity of an arbitrary bipartite channel $\mathcal{N}_{A'B'\to AB}$ is upper bounded (weak converse) by the 
    \begin{equation}\label{eq:converse-merging}
         C_{\rm gain}[A|B]_{\mathcal{N}} \leq \lim_{n\to\infty}\frac{1}{n}S^{\not\to}[A^n|B^n]_{\mathcal{N}^{\otimes n}},
    \end{equation}
    where $\frac{1}{n}S^{\not\to}[A^n|B^n]_{\mathcal{N}^{\otimes n}}$ is the regularized von Neumann conditional NS-entropy, i.e., $S^{\not\to}[A^n|B^n]_{\mathcal{N}^{\otimes n}}$ is the von Neumann conditional NS-entropy of the $n$ parallel uses of a bipartite channel $\mathcal{N}$. 
\end{proposition}

   We omit the detailed proof here as the proof argument for the above theorem is similar (almost same arguments) to the (converse) proof of \cite[Proposition 23]{GW21}, which is closely based on \cite[Section IV.B.]{HOW06}. To arrive at the inequality~\eqref{eq:converse-merging}, we made the observation that
    \begin{equation}
        \inf_{\op{\psi}\in{\rm St}(RA'^nB'^n)}S(A^n|RB^n)_{\mathcal{N}^{\otimes n}(\op{\psi})}=S^{\not\to}[A^n|B^n]_{\mathcal{N}^{\otimes n}}.
    \end{equation}

For any tensor-product channel $\mathcal{N}^1_{A'\to A}\otimes\mathcal{N}^2_{B'\to B}$, $C_{\rm gain}[A|B]_{\mathcal{N}^1\otimes\mathcal{N}^2}=S^{\not\to}[A|B]_{\mathcal{N}^1\otimes\mathcal{N}^2} = S[A|B]_{\mathcal{N}^1\otimes\mathcal{N}^2} =S[A]_{\mathcal{N}^1}$, due to Theorem~\ref{thm:axioms} and \cite[Theorem 10]{GW21}.

\section{Strong subadditivity of the entropy of channels} \label{sec:strong_subad_QCH}
The strong subaddivity of quantum von Neumann entropy for an arbitrary tripartite state $\rho_{ABC}$ is~\cite{LR73}
\begin{equation}
    S(A|B)_{\rho}-S(A|BC)_{\rho}\geq 0,
\end{equation}
which is equivalent to the quantum conditional mutual information $I(A;B|C)_\rho$ for quantum states being nonnegative,
\begin{equation}
    I(A;C|B)_{\rho}:=S(A|B)_{\rho}-S(A|BC)_{\rho}\geq 0.
\end{equation}
The structure of states that satisfy the strong subadditivity condition with equality were studied in \cite{HJPW04}. Tripartite quantum states $\rho_{ABC}$ for which $I(A;C|B)_{\rho}=0$ are called quantum Markov states~\cite{HJPW04,JRS+18}. In this case there exist an operation from subsystem $B$ to $BC$ which can reconstruct $\rho_{ABC}$ from $\rho_{AB}$.   

The strong subadditivity of quantum entropy for an arbitrary tripartite quantum channel $\mathcal{N}_{A'B'C'\to ABC}$ implies that
\begin{equation}
\mathbf{S}[A|B]_{\mathcal{N}}-\mathbf{S}[A|BC]_{\mathcal{N}}\geq 0. 
\end{equation}
\begin{dfn}[Generalized conditional mutual information]
    The generalized conditional mutual information $\mathbf{I}[A;C|B]_{\mathcal{N}}$ of an arbitrary tripartite channel $\mathcal{N}_{A'B'C'\to ABC}$ is defined as 
    \begin{equation}~\label{eq:qcmi-ch}
       \mathbf{I}[A;C|B]_{\mathcal{N}}:= \mathbf{S}[A|B]_{\mathcal{N}}-\mathbf{S}[A|BC]_{\mathcal{N}}.
    \end{equation}
We obtain the (quantum) conditional mutual information $I[A;C|B]_{\mathcal{N}}$ of a tripartite channel $\mathcal{N}_{A'B'C'\to ABC}$ by considering the conditional von Neumann entropy in Eq.~\eqref{eq:qcmi-ch}, i.e., the associated generalized divergence is taken to be the quantum relative entropy.
\end{dfn}
It follows from the definition that $\mathbf{I}[A;C|B]_{\mathcal{N}}\geq 0$ and so $I[A;C|B]_{\mathcal{N}}\geq 0$.

\begin{dfn}[Quantum Markov channel]
    A tripartite channel $\mathcal{N}_{A'B'C'\to ABC}$ is called a quantum Markov chain or channel, in order $A-B-C$, if its conditional mutual information vanishes, i.e.,
    \begin{equation}
        I[A;C|B]_{\mathcal{N}}=0.
    \end{equation}
\end{dfn}

\begin{theorem}\label{thm:cmi-tele-ch}
The conditional mutual information ${I}[A;C|B]_{\mathcal{N}}$ of a tele-covariant tripartite channel $\mathcal{N}_{A'B'C'\to ABC}$ is equal to the conditional mutual information of its Choi state in the following way
    \begin{equation}~\label{eq:qcmi-tele}
       {I}[A;C|B]_{\mathcal{N}}= I(R_AA;C|R_BBR_C)_{\Phi^{\mathcal{N}}},
       \end{equation}
       where $\Phi^{\mathcal{N}}_{R_AAR_BBR_CC}=\mathcal{N}(\Phi_{R_AA'}\otimes\Phi_{R_BB'}\otimes\Phi_{R_CC'})$.
\end{theorem}
\begin{proof}
For a tele-covariant channel $\mathcal{N}_{A'B'C'\to ABC}$, using proof arguments similar to the proof of Theorem~\ref{thm:vN-choi-state}, we can show that 
\begin{align}
    S[A|B]_{\mathcal{N}}=S(R_AA|R_BBR_C)_{\Phi^{\mathcal{N}}}-\log|A'|.
\end{align}
We then have
    \begin{align}
         I[A;C|B]_{\mathcal{N}}&=S[A|B]_{\mathcal{N}}-S[A|BC]_{\mathcal{N}}\nonumber\\
         &= S(R_AA|R_BBR_C)_{\Phi^{\mathcal{N}}}-S(R_AA|R_BBR_CC)_{\Phi^{\mathcal{N}}}\nonumber\\
         &=I(R_AA;C|R_BBR_C)_{\Phi^{\mathcal{N}}}.
    \end{align}
\end{proof}
\begin{corollary}
    The conditional mutual information ${I}[A;C|B]_{\mathcal{U}}$ of a tele-covariant tripartite unitary channel $\mathcal{U}_{A'B'C'\to ABC}$ is equal to the mutual information of the marginal (reduced state) $\Phi^{\mathcal{U}}_{R_AAC}$ of its Choi state $\Phi^{\mathcal{U}}_{R_AAR_BBR_CC}$,
    \begin{equation}
       {I}[A;C|B]_{\mathcal{U}}= I(R_AA;C)_{\Phi^{\mathcal{U}}}.
    \end{equation}
\end{corollary}
\begin{remark}
    For a tele-covariant tripartite channel $\mathcal{N}_{A'B'C'\to ABC}$, its conditional mutual information $I[A;C|B]_{\mathcal{N}}$ is upper bounded by the conditional mutual information $I(\overbar{A};\overbar{C}|\overbar{B})_{\Phi^{\mathcal{N}}}$ of its Choi state $\Phi^{\mathcal{N}}_{\overbar{A}\overbar{B}\overbar{C}}$, where $\overbar{A}:=R_AA, \overbar{B}:=R_BB, \overbar{C}:=R_CC$, i.e.,
    \begin{equation}
        I[A;C|B]_{\mathcal{N}}\leq I(\overbar{A};\overbar{C}|\overbar{B})_{\Phi}.
    \end{equation}
\end{remark}
\begin{proposition}[Continuity of conditional mutual information]
    If a tripartite channel $\mathcal{N}_{A'B'C'\to ABC}$ is $\varepsilon$-close to a tele-covariant tripartite channel $\mathcal{T}_{A'B'C'\to ABC}$ in the diamond norm, i.e., $\frac{1}{2}\norm{\mathcal{N}-\mathcal{T}}_{\diamond}= \varepsilon\in[0,1]$, then
    \begin{equation}\label{eq:con-qcmi-ch}
        \abs{I[A;C|B]_{\mathcal{N}}-I[A;C|B]_{\mathcal{T}}}\leq 2\varepsilon\log\min\{|A'||A|,|C|\}+2g_2(\varepsilon).
    \end{equation}
\end{proposition}
\begin{proof}
   Given that $\frac{1}{2}\norm{\mathcal{N}-\mathcal{T}}_{\diamond}= \varepsilon$, we have $\frac{1}{2}\norm{\Phi^{\mathcal{N}}-\Phi^{\mathcal{T}}}_1\leq \varepsilon$. We then arrive at the bound~\eqref{eq:con-qcmi-ch} by directly invoking Theorem~\ref{thm:cmi-tele-ch} and the continuity of quantum conditional mutual information (see \cite[Corollary 2]{Shi17}): For two states $\rho_{ABC},\sigma_{ABC}$ such that $\frac{1}{2}\norm{\rho-\sigma}_1=\varepsilon\in[0,1]$ then
    \begin{equation}
         \abs{I[A;C|B]_{\rho}-I[A;C|B]_{\rho}}\leq 2\varepsilon\log\min\{|A|,|C|\}+2g_2(\varepsilon).
    \end{equation}
\end{proof}

\subsection{Approximate Quantum Markov channel}
For an arbitrary quantum state $\rho_{ABC}$, there exists a universal recovery map $\mathcal{R}_{B\to BC}$ such that~\cite{SFR16,JRS+18} (also see~\cite{BHOS15,LDW18} for relevant discussions) 
\begin{equation}\label{eq:qcmi-fid}
    I(A;C|B)_{\rho}\geq -\log F(\rho_{ABC},\mathcal{R}_{B\to BC}(\rho_{AB})).
\end{equation}
A universal recovery map $\mathcal{R}_{B\to BC}$ in inequality~\eqref{eq:qcmi-fid} is a quantum channel depending only on $\rho_{BC}$ (not on $\rho_{ABC}$) and $\mathcal{R}_{B\to BC}(\rho_{B})=\rho_{BC}$. From inequality~\eqref{eq:qcmi-fid}, we see that for any state $\rho_{ABC}$, there exists a universal recovery map $\mathcal{R}_{B\to BC}$ such that
\begin{equation}\label{eq:fid-cmi}
    F(\rho_{ABC},\mathcal{R}_{B\to BC}(\rho_{AB}))\geq 2^{-I(A;C|B)_{\rho}}.
\end{equation}
A state $\rho_{ABC}$ is a quantum Markov state ($A-B-C$), i.e., $I(A;C|B)_{\rho}=0$, if and only if there exists a universal map $\mathcal{R}_{B\to BC}$ such that $\rho_{ABC}=\mathcal{R}_{B\to BC}(\rho_{AB})$. An approximate Markov state $\rho_{ABC}$ is approximately recoverable from $\rho_{AB}$ (or $\rho_{BC}$) via an universal recovery map $\mathcal{R}_{B\to BC}$ (or $\mathcal{R}_{B\to AB}$) as $I(A;C|B)_{\rho}\leq \varepsilon$ for some appropriately small $\varepsilon$.

\begin{dfn}
    A tripartite quantum channel $\mathcal{N}_{A'B'C'\to ABC}$ is called an $\varepsilon$-approximate quantum Markov chain (channel), in order $A-B-C$, if $I[A;C|B]_{\mathcal{N}}\leq \varepsilon$ for $0\leq \varepsilon\ll 1$.
\end{dfn}

\begin{proposition}\label{pro:tele-cov-cmi}
     If a tele-covariant tripartite channel $\mathcal{N}_{A'B'C'\to ABC}$ is an $\varepsilon$-approximate quantum Markov chain, i.e., $I[A;C|B]_{\mathcal{N}}\leq \varepsilon$ for a sufficiently small $\varepsilon$, then there exists
    \begin{enumerate}
        \item  a universal recovery map $\mathcal{R}_{R_BBR_C\to R_AAR_BBR_C}$ such that
    \begin{equation}
        F(\Phi^{\mathcal{N}}_{R_AAR_BBR_CC},\mathcal{R}_{R_BBR_C\to R_AAR_BBR_C}(\Phi^{\mathcal{N}}_{R_BBR_CC}))\geq 2^{-\varepsilon}\approx 1.
    \end{equation}
    \item  a universal recovery map $\mathcal{R}_{R_BBR_C\to R_BBR_C}$ such that
    \begin{equation}
        F(\Phi^{\mathcal{N}}_{R_AAR_BBR_CC},\mathcal{R}_{R_BBR_C\to R_BBR_CC}(\Phi^{\mathcal{N}}_{R_AAR_BBR_C}))\geq 2^{-\varepsilon}\approx 1.
           \end{equation}
    \end{enumerate}
\end{proposition}
\begin{proof}
    The proof follows by invoking Theorem~\ref{thm:cmi-tele-ch} and Eq.~\eqref{eq:fid-cmi}.
\end{proof}

    A tele-covariant channel $\mathcal{T}_{A'B'C'\to ABC}$ is simulable via teleportation with the resource state as its Choi state $\Phi^{\mathcal{T}}_{R_AAR_BBR_CC}$~\cite[Theorem 5]{DBWH21}. That is, the action of the tele-covariant channel $\mathcal{T}$ can be obtained by local (quantum) operations and classical communication (LOCC) among Ash (holding $R_A,A',A$), Bob (holding $R_B,B',B$), and Charlie (holding $R_C,C',C$), if Ash, Bob, and Charlie share the Choi state $\Phi^{\mathcal{N}}$ among themselves,
    \begin{equation}
        \mathcal{T}_{A'B'C'\to ABC} (\cdot)= \mathcal{L}_{R_AA'AR_BB'BR_CC'C\to ABC}(\cdot\otimes \Phi^\mathcal{T}_{R_AAR_BBR_CC}),
    \end{equation}
    where $\mathcal{L}$ is an LOCC channel among Ash, Bob, and Charlie (one-way LOCC channel from Ash to Bob suffices). This fact along with the Proposition~\ref{pro:tele-cov-cmi} lead to the following observation.
    \begin{observation}\label{obs:tele-sim}
        For a tele-covariant channel $\mathcal{T}_{A'B'C'\to ABC}$ that is also a quantum Markov chain $(I[A;C|B]_{\mathcal{T}}=0)$, its Choi state $\Phi^{\mathcal{T}}_{R_AAR_BBR_CC}$ can be perfectly recovered by applying a universal recovery map $\mathcal{R}^P_{R_BBR_C\to R_AAR_BBR_C}$ on the marginal state $\Phi^{\mathcal{T}}_{R_BBR_CC}$, where due to~\cite{Petz:1986},
        \begin{equation}
            \mathcal{R}^P_{R_BBR_C\to R_AAR_BBR_C}(\cdot)=(\Phi^{\mathcal{T}}_{R_AAR_BBR_C})^{\frac{1}{2}}[\mathbbm{1}_{R_A}\otimes\mathbbm{1}_A\otimes(\Phi^{\mathcal{T}}_{R_BBR_C})^{-\frac{1}{2}}(\cdot)(\Phi^{\mathcal{T}}_{R_BBR_C})^{-\frac{1}{2}}](\Phi^{\mathcal{T}}_{R_AAR_BBR_C})^{\frac{1}{2}}.
        \end{equation}
   Even if an initial reference system $R_A$ and an output $A$ of the channel $\mathcal{T}_{A'B'C'\to ABC}$ are made inaccessible (for some reason), the tripartite channel can be perfectly simulated via teleportation by first recovering the resource state (the Choi state) from the marginal state $\Phi^{\mathcal{T}}_{R_BBR_CC}$.
    \end{observation}
    We succinctly present the findings from the above observation as a lemma below (see also \cite[Theorem 2 and Proposition 9]{SPSD24} for a broader statement).
       \begin{lemma}
         A quantum Markov channel $\mathcal{T}_{A'B'C'\to ABC}$ that is also tele-covariant, is perfectly recoverable (or simulable in the sense of Observation~\ref{obs:tele-sim} and Proposition~\ref{pro:tele-cov-cmi}) from the reduced channel $\tr_A\circ\mathcal{T}_{A'B'C'\to ABC}$ or $\tr_C\circ\mathcal{T}_{A'B'C'\to ABC}$.
    \end{lemma}

\section{Generalized mutual information of quantum channels} \label{sec:mutual_information_CH}
Informational measures of quantum channels provide useful notions to quantify information processing capabilities of quantum channels, for examples see~\cite{SS09,CGLM14,HQRY16,Das19,DBWH21,HWD22,GPG+24,ERAG24}. (Multipartite) mutual information of quantum states is a useful quantity that captures the total amount of correlations, classical or quantum, in a multipartite state. In this work, we provide analogue of the multipartite quantum information $I(A;B;\ldots, N)_{\rho}$ of a $N$-partite state $\rho_{AB\ldots N}$, due to \cite{Wan60} (see \cite{CS02,Pat09,MPW+10}),
\begin{align}
    I(A;B;\ldots;N)_{\rho}& :=D(\rho_{AB\ldots N}\Vert\rho_A\otimes\rho_B\otimes\cdots\otimes\rho_N)\nonumber\\
    &=\inf_{\omega_{i}\in{\rm St}(i), \forall i\in\{A,B,\ldots, N\}} D(\rho_{AB\ldots N}\Vert\otimes_{i=A}^N\omega_i).
\end{align}
The mutual information $I(A;B)_{\rho}$ of a state $\rho_{AB}$ finds operational meaning in the task of erasing the total correlation present in the state $\rho_{AB}$~\cite{GPW05}. It also has a direct operational meaning as the optimal error exponent in the asymmetric hypothesis testing, in Stein's setting, for the state discrimination task between $\rho_{AB}$~(null hypothesis) and $\rho_{A}\otimes\rho_{B}$~(alternate hypothesis). Analogous to the mutual information of multipartite states, the multipartite generalized mutual information of a multipartite state $\rho_{AB\ldots N}$ is defined as
\begin{equation}
     \mathbb{I}(A;B;\ldots;N)_{\rho}:=\inf_{\omega_{i}\in{\rm St}(i), \forall i\in\{A,B,\ldots, N\}} \mathbb{D}(\rho_{AB\ldots N}\Vert\otimes_{i=A}^N\omega_i).
\end{equation}

We now discuss ways to quantify the total amount of correlations intrinsic to multipartite quantum channels and bipartite channels as a particular case.

\begin{dfn}[Generalized mutual information]
    The generalized mutual information of a multipartite quantum channel $\mathcal{N}_{A'B'\ldots N'\to AB\ldots N}$ is defined as
    \begin{equation}
        \mathbb{I}[A;B;\ldots;N]_{\mathcal{N}}:=\inf_{\mathcal{M}^1\otimes\mathcal{M}^2\otimes\cdots\otimes\mathcal{M}^n}\mathbb{D}[\mathcal{N}_{A'B'\ldots N'\to AB\ldots N}\Vert\mathcal{M}^1_{A'\to A}\otimes \mathcal{M}^2_{B'\to B}\otimes\cdots\otimes\mathcal{M}^n_{N'\to N}],
    \end{equation}
    where $\mathcal{M}^1\in\mathrm{Q}(A',A),\mathcal{M}^2\in\mathrm{Q}(B',B),\ldots,\mathcal{M}^n\in\mathrm{Q}(N',N)$ and $\mathbb{D}[\cdot\Vert\cdot]$ is the generalized channel divergence.
\end{dfn}

\begin{lemma}
For a multipartite channel $\mathcal{N}_{A'B'\ldots N'\to AB\ldots N}$, $\mathbb{I}[A;B;\ldots;N]_{\mathcal{N}}\geq 0$ if the associated generalized divergence is always nonnegative for the states. When the associated generalized divergence is faithful for the states then $\mathbb{I}[A;B;\ldots;N]_{\mathcal{N}}= 0$ if and only if $\mathcal{N}_{A'B'\ldots N'\to AB\ldots N}$ is product of local channels, i.e., of the form $\mathcal{N}_{A'B'\ldots N'\to AB\ldots N}=\mathcal{N}^1_{A'\to A}\otimes \mathcal{N}^2_{B'\to B}\otimes\cdots\otimes\mathcal{N}^n_{N'\to N}$.
\end{lemma}

\begin{proposition}
      The multipartite generalized mutual information of a (multipartite) quantum channel $\mathcal{N}_{A'B'\ldots N'\to AB\ldots B}$ is lower bounded by the multipartite generalized mutual information of its Choi state $\Phi^{\mathcal{N}}_{R_AAR_BB\ldots R_NN}$,
    \begin{equation}
     \mathbb{I}[A;B;\ldots;N]_{\mathcal{N}}\geq    \mathbb{I}(R_AA;R_BB;\ldots;R_NN)_{\Phi^\mathcal{N}}.
    \end{equation}
\end{proposition}
\begin{proof}
     We prove the theorem for $N=3$ and the proof argument is generalizable for any positive integer $N\geq 2$.

    Let us assume $\mathcal{M}^1_{A'\to A}, \mathcal{M}^2_{B'\to B}, \mathcal{M}^3_{C'\to C}$ be quantum channels such that the following identity holds,
    \begin{align}
        \mathbb{I}[A;B;C]_{\mathcal{N}} &= \mathbb{D}[\mathcal{N}_{A'B'C'\to ABC}\Vert \mathcal{M}^1_{A'\to A}\otimes\mathcal{M}^2_{B'\to B}\otimes\mathcal{M}^3_{C'\to C}]\nonumber\\
        &\geq \mathbb{D}(\mathcal{N}(\Phi_{R_AA'}\otimes\Phi_{R_BB'}\otimes\Phi_{R_CC'})\Vert \mathcal{M}^1(\Phi_{R_AA'})\otimes\mathcal{M}^2(\Phi_{R_BB'})\otimes\mathcal{M}^3(\Phi_{R_CC'}))\nonumber\\
        &= \mathbb{D}(\Phi^{\mathcal{N}}_{R_AAR_BBR_CC}\Vert\Phi^{\mathcal{M}^1}_{R_AA}\otimes\Phi^{\mathcal{M}^2}_{R_BB}\otimes\Phi^{\mathcal{M}^3}_{R_CC})\nonumber\\
        & \geq \mathbb{I}(R_AA;R_BB;R_CC)_{\Phi^{\mathcal{N}}}.
    \end{align}
\end{proof}

\begin{dfn}
    The multipartite mutual information of a multipartite channel $\mathcal{N}_{A'B'\ldots N'\to AB\ldots N}$ is defined as
    \begin{equation}
    I[A;B;\ldots;N]_{\mathcal{N}}=\inf_{\mathcal{M}^1\otimes\mathcal{M}^2\otimes\cdots\otimes\mathcal{M}^n} D[\mathcal{N}_{A'B'\ldots N'\to AB\ldots N}\Vert\mathcal{M}^1_{A'\to A}\otimes \mathcal{M}^2_{B'\to B}\otimes\cdots\otimes\mathcal{M}^n_{N'\to N}],
    \end{equation}
    where  $\mathcal{M}^1\in\mathrm{Q}(A',A),\mathcal{M}^2\in\mathrm{Q}(B',B),\ldots,\mathcal{M}^n\in\mathrm{Q}(N',N)$.
\end{dfn}

For a replacer channel $\mathcal{N}$ that outputs only a fixed state $\rho\in{\rm St}(AB\ldots N)$, its multipartite mutual information is equal to the multipartite mutual information of the state $\rho_{AB\ldots N}$,
\begin{equation}
    I[A;B;\ldots;N]_{\mathcal{N}}=I(A;B;\ldots;N)_{\rho}.
\end{equation}

\begin{theorem}\label{thm:mi-cov-ch}
    For a tele-covariant multipartite channel $\mathcal{T}_{A'B'\ldots N'\to AB\ldots B}$~\cite[Theorem 5]{DBWH21}, the multipartite mutual information of the channel is equal to the multipartite mutual information of its Choi state $\Phi^{\mathcal{T}}_{R_AAR_BB\ldots R_NN}$,
    \begin{equation}
     I[A;B;\ldots;N]_{\mathcal{T}}=   I(R_AA;R_BB;\ldots;R_NN)_{\Phi^\mathcal{T}}.
    \end{equation}
\end{theorem}
\begin{proof}
    We prove the theorem for $N=3$ and the proof argument is generalizable for any positive integer $N\geq 2$.

    Let us assume $\mathcal{M}^1_{A'\to A}, \mathcal{M}^2_{B'\to B}, \mathcal{M}^3_{C'\to C}$ be quantum channels such that the following identity holds,
    \begin{align}
        I[A;B;C]_{\mathcal{T}} &= D[\mathcal{T}_{A'B'C'\to ABC}\Vert \mathcal{M}^1_{A'\to A}\otimes\mathcal{M}^2_{B'\to B}\otimes\mathcal{M}^3_{C'\to C}]\nonumber\\
        &\geq D(\Phi^{\mathcal{T}}_{R_AAR_BBR_CC}\Vert\Phi^{\mathcal{M}^1}_{R_AA}\otimes\Phi^{\mathcal{M}^2}_{R_BB}\otimes\Phi^{\mathcal{M}^3}_{R_CC})\nonumber\\
        & \geq I(R_AA;R_BB;R_CC)_{\Phi^{\mathcal{T}}}.
    \end{align}
    Now, we notice that there exists tele-covariant channels $\mathcal{M}^{1,\mathcal{T}}_{A'\to A}, \mathcal{M}^{2,\mathcal{T}}_{B'\to B}, \mathcal{M}^{3,\mathcal{T}}_{C'\to C}$ that are jointly tele-covariant with $\mathcal{T}_{A'B'C'\to ABC}$, see Lemma~\ref{lem:tele-cov-reduced}, and are of the form,
    \begin{equation}
        \mathcal{M}^{1,\mathcal{T}}_{A'\to A}=\tr_{BC}\circ \mathcal{T}_{A'B'C'\to ABC}(\cdot\otimes\pi_{B}\otimes\pi_C).
    \end{equation}
    Then, noticing that $\mathcal{M}^{1,\mathcal{T}}_{A'\to A}(\Phi_{R_AA'})=\tr_{BC}\circ \mathcal{T}_{A'B'C'\to ABC}(\Phi_{R_AA}\otimes\pi_{B}\otimes\pi_C)$, we have
     \begin{align}
        I[A;B;C]_{\mathcal{T}}&\leq D[\mathcal{T}_{A'B'C'\to ABC}\Vert \mathcal{M}^{1,\mathcal{T}}_{A'\to A}\otimes\mathcal{M}^{2,\mathcal{T}}_{B'\to B}\otimes\mathcal{M}^{3,\mathcal{T}}_{C'\to C}]\nonumber\\
        &= I(R_AA;R_BB;R_CC)_{\Phi^{\mathcal{T}}}.
        \end{align}
That concludes the proof.
\end{proof}

A direct consequence of the above theorem is the following corollary.
\begin{corollary}
    The half of the mutual information of any tele-covariant bipartite unitary channel $\mathcal{U}_{A'B'\to AB}$ is equal to the von Neumann entropy of its reduced state,
    \begin{equation}
       \frac{1}{2} I[A;B]_{\mathcal{U}}=S(R_AA)_{\Phi^\mathcal{U}}=S(R_BB)_{\Phi^\mathcal{U}}.
    \end{equation}
\end{corollary}

\begin{lemma}
    The multipartite max-mutual information and the multipartite geometric R\'enyi mutual information of a (multipartite) quantum channel $\mathcal{N}_{A'B'\ldots N'\to AB\ldots N}$ are respectively given as
    \begin{align}
       I_{\max}[A;B;\ldots;N]_{\mathcal{N}}&=   \inf_{\bigotimes_{i=A}^N\mathcal{M}^i}D_{\max}(\Phi^{\mathcal{N}}_{R_AAR_BB\ldots R_NN}\Vert \otimes_{i=A}^{N}\Phi^{\mathcal{M}^i}_{R_ii}),\\
        \widehat{I}_{\alpha}[A;B;\ldots;N]_{\mathcal{N}}&=   \inf_{\bigotimes_{i=A}^N\mathcal{M}^i}\widehat{D}_{\alpha}(\Phi^{\mathcal{N}}_{R_AAR_BB\ldots R_NN}\Vert \otimes_{i=A}^{N}\Phi^{\mathcal{M}^i}_{R_ii}) \quad \forall \alpha\in(1,2],
    \end{align}
    where $\mathcal{M}^i\in{\rm Q}(i',i)$ for all $i\in\{A,B,\ldots,N\}$.
\end{lemma}

\subsection{Conditional mutual information for channels: An alternate definition}
We recall that the conditional quantum mutual information of a tripartite state $\rho_{ABC}$ is
\begin{align}
    I(A;C|B)_{\rho}&= S(A|B)_{\rho} -S(A|BC)_{\rho} =I(A;BC)_{\rho}-I(A;B)_{\rho}.\label{eq:qcmi-states-mi}
\end{align}
We can generalize the definition of the conditional mutual information of states from the first identity in Eq.~\eqref{eq:qcmi-states-mi} to tripartite channels in an alternate way as follows.
\begin{dfn}[MI-based generalized conditional mutual information]
    The mutual information (MI)-based generalized conditional mutual information of an arbitrary tripartite channel $\mathcal{N}_{A'B'C'\to ABC}$ is defined as
    \begin{equation}
       \mathbf{\Delta}[A;C|B]_{\mathcal{N}}:= \mathbf{I}[A;BC]_{\mathcal{N}}-\mathbf{I}[A;B]_{\mathcal{N}},
    \end{equation}
    where $\mathbf{I}[A;BC]_{\mathcal{M}}$ denotes the (bipartite) mutual information of a bipartite channel $\mathcal{M}_{A':B'C\to A:BC}$ and $\mathbf{I}[A;B]_{\mathcal{M}}$ denotes the mutual information of its reduced bipartite channel $\mathcal{M}_{A':B'C'\to A:B}:=\tr_C\circ\mathcal{M}_{A':B'C\to A:BC}$, i.e.,
    \begin{align}
        \mathbf{I}[A;BC]_{\mathcal{N}} &=\inf_{\mathcal{M}^1\in{\rm Q}(A',A),\mathcal{M}^2\in{\rm Q}(B'C',BC)}\mathbf{D}[\mathcal{N}_{A'B'C'\to ABC}\Vert\mathcal{M}^1_{A'\to A}\otimes\mathcal{M}^2_{B'C'\to BC}]\\
      \mathbf{I}[A;B]_{\mathcal{N}}  & = \inf_{\mathcal{M}^1\in{\rm Q}(A',A),\mathcal{M}^2\in{\rm Q}(B'C',B)}\mathbf{D}[\mathcal{N}_{A'B'C'\to AB}\Vert\mathcal{M}^1_{A'\to A}\otimes\mathcal{M}^2_{B'C'\to B}].
    \end{align}
\end{dfn}

\begin{remark}
    The MI-based generalized conditional mutual information $ \mathbf{\Delta}[A;C|B]_{\mathcal{N}}$ for any tripartite channel $\mathcal{N}_{A'B'C'\to ABC}$ is always nonnegative, $\mathbf{\Delta}[A;C|B]_{\mathcal{N}}\geq 0$. This can be referred to as MI-based strong subadditivity. It directly follows from the monotonicity of the generalized channel divergence.
\end{remark}

We can then study the MI-based quantum Markov channel in a similar spirit as in Section~\ref{sec:strong_subad_QCH}. We call a tripartite quantum channel $\mathcal{N}_{A'B'C'\to ABC}$ a MI-based quantum Markov chain or channel, in order $A-B-C$, if
\begin{equation}
    \Delta[A;C|B]_{\mathcal{N}}:= I[A;BC]_{\mathcal{N}}-I[A;B]_{\mathcal{N}}=0.
\end{equation}
It is apparently not clear if $I[A;C|B]_{\mathcal{N}}=\Delta[A;C|B]_{\mathcal{N}}$ for a tripartite channel $\mathcal{N}_{A'B'C'\to ABC}$ in general. However, for tele-covariant tripartite channels $\mathcal{T}_{A'B'C'\to ABC}$, both definitions of conditional mutual information are equal.
\begin{lemma}
    If $\mathcal{T}_{A'B'C'\to ABC}$ is a tele-covariant tripartite channel, then
\begin{equation}
    \Delta[A;C|B]_{\mathcal{T}}=I[A;C|B]_{\mathcal{T}}=I(R_AA;C|R_BBR_C)_{\Phi^{\mathcal{T}}},
\end{equation}
where $\Phi^{\mathcal{T}}_{R_AAR_BBR_CC}$ is the Choi state of $\mathcal{T}$.
\end{lemma}
\begin{proof}
Consider a tele-covariant tripartite channel $\mathcal{T}_{A'B'C'\to ABC}$. From Theorem~\ref{thm:cmi-tele-ch}, we have $I[A;C|B]_{\mathcal{T}}=I(R_AA;C|R_BBR_C)_{\Phi^{\mathcal{T}}}$. We now notice that the reduced channel of a tele-covariant multipartite channel is also a tele-covariant channel and jointly tele-covariant with the full (original) channel, see Lemma~\ref{lem:tele-cov-reduced} and \cite[Remark 5]{DBWH21}. Using this fact and Theorem~\ref{thm:mi-cov-ch},
\begin{align}
    \Delta[A;C|B]_{\mathcal{T}}& = I[A;C|B]_{\mathcal{T}}-I[A;B]_{\mathcal{T}}\nonumber\\
    & = I(R_AA;R_BBR_CC)_{\Phi^{\mathcal{T}}}-I(R_AA;R_BBR_C)_{\Phi^{\mathcal{T}}}\nonumber\\
    &= I(R_AA;C|R_BBR_C)_{\Phi^{\mathcal{T}}}=I[A;C|B]_{\mathcal{T}}.
\end{align}
\end{proof}

\section{Bidirectional conditional entropy and information} \label{sec:bidirectional_CH}
A bipartite channel $\mathcal{N}_{A'B'\to AB}$ is called a bidirectional channel if pairs $(A',A)$ and $(B',B)$ are in separate labs, say with Ash and Bob, respectively, and they have access only to the separable states for input to the channel (cf.~\cite{BHLS03,Das19,DBW17}). However, it suffices to optimize $\mathbf{D}^{\leftrightarrow}[\mathcal{N}^1\Vert\mathcal{N}^2]$ between two bidirectional quantum channels $\mathcal{N}^1,\mathcal{N}^2\in{\rm Q}(A'B',AB)$ over (pure) product input states $\psi_{R_AA'}\otimes\phi_{R_BB'}$. 

\begin{dfn}[Bidirectional entropy]
     The generalized entropy of a bidirectional channel $\mathcal{N}_{A'B'\to AB}$ is called bidirectional generalized entropy $\mathbf{S}^{\leftrightarrow}[AB]_{\mathcal{N}}$,
     \begin{equation}
         \mathbf{S}^{\leftrightarrow}[AB]_{\mathcal{N}}:= \mathbf{S}^{\leftrightarrow}[\mathcal{N}]:= - \mathbf{D}^{\leftrightarrow}[\mathcal{N}_{A'B'\to AB}\Vert\mathcal{R}_{A'\to A}\otimes\mathcal{R}_{B'\to B}],
     \end{equation}
     where $\mathbf{D}^{\leftrightarrow}$ is an appropriate generalized channel divergence.
\end{dfn}
\begin{dfn}[Bidirectional conditional entropy]
    The generalized conditional entropy of a bidirectional channel $\mathcal{N}_{A'B'\to AB}$ is called bidirectional generalized conditional entropy $\mathbf{S}^{\leftrightarrow}[A|B]_{\mathcal{N}}$,
    \begin{equation}
        \mathbf{S}^{\leftrightarrow}[A|B]_{\mathcal{N}}:=-\inf_{\mathcal{M}\in\mathrm{Q}(B',B)}\mathbf{D}^{\leftrightarrow}[\mathcal{N}_{A'B'\to AB}\Vert\mathcal{R}_{A'\to A}\otimes\mathcal{M}_{B'\to B}].
    \end{equation}
\end{dfn}

\begin{remark}
    For a bipartite channel $\mathcal{N}_{A'B'\to AB}$, the bidirectional generalized conditional entropy $\mathbf{S}^{\leftrightarrow}[A|B]_{\mathcal{N}}$ is lower bounded by the generalized conditional entropy $\mathbf{S}[A|B]_{\mathcal{N}}$, 
    \begin{equation}\label{eq:rem-bi-gen}
        \mathbf{S}[A|B]_{\mathcal{N}}\leq \mathbf{S}^{\leftrightarrow}[A|B]_{\mathcal{N}}. 
    \end{equation}
\end{remark}
\begin{proposition}
        The bidirectional conditional entropy $\mathbf{S}^{\leftrightarrow}[A|B]_{\mathcal{N}}$ of a bipartite channel $\mathcal{N}_{A'B'\to AB}$ is invariant under postprocessing with a local unitary channel $\mathcal{U}_{A\to C}$ on Ash and nondecreasing under postprocessing with a local channel on Bob $\mathcal{M}_{B\to D}$,
        \begin{equation}
            \mathbf{S}^{\leftrightarrow}[A|B]_{\mathcal{N}}= \mathbf{S}^{\leftrightarrow}[C|B]_{\mathcal{U}\circ\mathcal{N}}\quad \text{and}\quad  \mathbf{S}^{\leftrightarrow}[A|B]_{\mathcal{N}}\leq  \mathbf{S}^{\leftrightarrow}[A|D]_{\mathcal{M}\circ\mathcal{N}}.
        \end{equation}
\end{proposition}

\begin{proposition}
         The bidirectional conditional entropy $\mathbf{S}^{\leftrightarrow}[A|B]_{\mathcal{N}}$ of a bipartite channel $\mathcal{N}_{A'B'\to AB}$ that is product of local channels $\mathcal{N}^1_{A'\to A}$ and $\mathcal{N}_{A'B'\to AB}$ is equal to the entropy of the channel $\mathcal{N}^1_{A'\to A}$,
     \begin{equation}
         \mathbf{S}^{\leftrightarrow}[A|B]_{\mathcal{N}}=\mathbf{S}[A]_{\mathcal{N}^1}=:\mathbf{S}[\mathcal{N}^1].
     \end{equation}
\end{proposition}

\begin{lemma}
    The bidirectional generalized entropy $\mathbf{S}^{\leftrightarrow}[AB]_{\mathcal{N}}$ and the generalized entropy $\mathbf{S}[AB]_{\mathcal{N}}$ of a tele-covariant channel $\mathcal{N}_{A'B'\to AB}$ are equal.
\end{lemma}

\begin{proposition}
    For a bipartite channel $\mathcal{N}_{A'B'\to AB}$, the von Neumann conditional entropy of its Choi state $\Phi^{\mathcal{N}}_{R_AAR_BB}$ and the bidirectional von Neumann conditional entropy $S^{\leftrightarrow}[A|B]_{\mathcal{N}}$ are related as
    \begin{equation}
        \log|A'|\leq S(R_AA|R_BB)_{\Phi^{\mathcal{N}}}-S^{\leftrightarrow}[A|B]_{\mathcal{N}}. 
    \end{equation}
\end{proposition}
\begin{proof}
    Let $\mathcal{M}_{B'\to B}$ be a quantum channel that yields infimum for $D^{\leftrightarrow}[\mathcal{N}_{A'B'\to AB}\vert \mathcal{R}_{A'\to A}\otimes\mathcal{M}_{B'\to B}]=-S^{\leftrightarrow}[A|B]_{\mathcal{N}}$.
    \begin{align}
        -S^{\leftrightarrow}[A|B]_{\mathcal{N}}&\geq D(\Phi^{\mathcal{N}}_{R_AAR_BB}\Vert \pi_{R_A}\otimes\mathbbm{1}_A\otimes\Phi^{\mathcal{M}}_{R_BB})\nonumber\\
        &=\log|A'|+D(\Phi^{\mathcal{N}}_{R_AAR_BB}\Vert \mathbbm{1}_{R_A}\otimes\mathbbm{1}_A\otimes\Phi^{\mathcal{M}}_{R_BB})\nonumber\\
        &\geq \log|A'|+ \inf_{\sigma\in\mathrm{St}(R_BB)}D(\Phi^{\mathcal{N}}_{R_AAR_BB}\Vert \mathbbm{1}_{R_A}\otimes\mathbbm{1}_A\otimes\sigma_{R_BB})\nonumber\\
        &= \log|A'|-S(R_AA|R_BB)_{\Phi^{\mathcal{N}}}.
    \end{align}
\end{proof}

\begin{dfn}
      The generalized mutual information of a bidirectional channel $\mathcal{N}_{A'B'\ldots N'\to AB\ldots N}$ is called bidirectional generalized mutual information $\mathbb{I}^{\leftrightarrow}[A;B;\ldots;N]_{\mathcal{N}}$,
    \begin{equation}
        \mathbb{I}^{\leftrightarrow}[A;B;\ldots;N]_{\mathcal{N}}:=\inf_{\mathcal{M}^1\otimes\mathcal{M}^2\otimes\cdots\otimes\mathcal{M}^n}\mathbb{D}^{\leftrightarrow}[\mathcal{N}_{A'B'\ldots N'\to AB\ldots N}\Vert\mathcal{M}^1_{A'\to A}\otimes \mathcal{M}^2_{B'\to B}\otimes\cdots\otimes\mathcal{M}^n_{N'\to N}],
    \end{equation}
    where $\mathcal{M}^1\in\mathrm{Q}(A',A),\mathcal{M}^2\in\mathrm{Q}(B',B),\ldots,\mathcal{M}^n\in\mathrm{Q}(N',N)$.
\end{dfn}
\begin{lemma}
For a multipartite channel $\mathcal{N}_{A'B'\ldots N'\to AB\ldots N}$, $ \mathbb{I}^{\leftrightarrow}[A;B;\ldots;N]_{\mathcal{N}}\geq 0$ if the associated generalized divergence is always nonnegative for the states. When the associated generalized divergence is faithful for the states then $ \mathbb{I}^{\leftrightarrow}[A;B;\ldots;N]_{\mathcal{N}}= 0$ if and only if $\mathcal{N}_{A'B'\ldots N'\to AB\ldots N}$ is product of local channels, i.e., of the form $\mathcal{N}_{A'B'\ldots N'\to AB\ldots N}=\mathcal{N}^1_{A'\to A}\otimes \mathcal{N}^2_{B'\to B}\otimes\cdots\otimes\mathcal{N}^n_{N'\to N}$.
\end{lemma}

\section{Discussion}\label{sec:discussion}
The notions of entropy, conditional entropy, mutual information, and conditional mutual information form primitives of the information theory for both classical and quantum systems. The applications of the conditional entropy, the strong subadditivity of the entropy, and mutual information measures of the quantum states have been immense in the study of quantum systems and processes from fundamental as well as technological aspects, see, e.g., quantum error correction, quantum communication, quantum field theory, condensed matter physics, memory-effects in open quantum dynamics, quantum thermodynamics, etc.~\cite{CC97,BM01,WVHC08,ADHW09,BDW16,CBTW17,BRLW17,PP17,PW21,AP22,AMT24,ANJ+24,GPG+24,BGG+24}.

In this work, we have introduced the definitions of the conditional entropy and strong subadditivity of the entropy of quantum channels grounded in well-motivated axioms. We also provide definition of the conditional mutual information of tripartite quantum channels. Analogous to the conditional entropy of states, we have shown that the conditional entropy of bipartite quantum channels has the potential to offer valuable insights into their correlation-generating features, which cannot be captured by the entropy. In particular, it can demarcate semicausal channels with no causal influence from the input non-conditioning system to the output conditioning system with the channels that allow for signaling between the aforementioned systems. The strong subadditivity property of the entropy allows us to define the quantum Markov chain for multipartite channels, analogous to the Markov chain for random variables and quantum states.

The fundamental role of conditional entropy and conditional mutual information in information theory motivates us to claim that these notions for quantum channels will have wide applications and implications in the study of quantum processes, including higher-order processes. We have discussed the recipe for the educated guesses of the definitions of the conditional entropy, strong subadditivity of the entropy, and the multipartite mutual information of quantum channels. This allows to extend the entropic and information measures of the channels to higher order processes in a way similar to \cite[Section V and VI]{SPSD24}. It would be useful to study the properties of these definitions and their operational meanings for transition maps (classical channels)\footnote{To the best of our knowledge, we are not aware of definitions of the conditional entropies and (conditional) mutual information for multipartite channels in classical information theory. We could perhaps be mistaken.}. One may also utilize the conditional entropy of channels or their variants to algorithmically characterize their underlying causal structure~\cite{BWZ+22,GC24} and characterize the memory effects in open quantum dynamics~\cite{DMM+24,BGG+24}. The conditional entropy of channel may also be used to derive the entropic uncertainty relation for the measurements of quantum channels~\cite{SRCZ16}. Our approach could be useful in the context of error correction in quantum gates (error correction at the dynamical level) and characterizing the memory effects in open quantum dynamics. We leave open the question of finding a novel information or computation task, other than the channel discrimination, for which the conditional entropy of an \textit{arbitrary} quantum channel has operational meaning as the optimal rate.

\begin{acknowledgments}
The authors thank Sohail, Himanshu Badhani, Manabendra Nath Bera, Karol Horodecki, and Uttam Singh for the discussions. S.D. acknowledges support from the Science and Engineering Research Board, Department of Science and Technology (SERB-DST), Government of India, under Grant No. SRG/2023/000217 and the Ministry of Electronics and Information Technology (MeitY), Government of India, under Grant No. 4(3)/2024-ITEA. S.D. also thanks IIIT Hyderabad for the Faculty Seed Grant. K.G. is supported by the Senior Research Fellowship Scheme SRFS2021-7S02 and the John Templeton Foundation through grant 62312, “The Quantum Information Structure of
Spacetime” (qiss.fr).
\end{acknowledgments}
\appendix

\section{Proof of Theorem \ref{chgen_div_mon}}\label{prop:monotonicity_of_channeldivergence_proof}

\begin{proof}
     To prove that the proposed inequality holds, we start with the l.h.s., i.e., $ \mathbb{D}\left[\Theta\left(\mathcal{N}\right)\Vert \Theta \left(\mathcal{M}\right)\right]$. Now, considering a general superchannel $\Theta$, whose action on any CP map is defined in Eq.~\eqref{equ:superchannel_action}, we can write
  \begin{align}
     & \chD*{\Theta\left(\mathcal{N}\right)}{\Theta \left(\mathcal{M}\right)} \nonumber\\
     & = \chD*{\mathcal{P}^{\rm{post}}_{AE \rightarrow C}\hspace{0.05cm}\circ\left(\mathcal{N}_{A'\rightarrow A} \otimes \mathrm{id}_{E}\right) \hspace{0.05cm} \circ \mathcal{P}^{\rm{pre}}_{C' \rightarrow A'E}}{\mathcal{P}^{\rm{post}}_{AE \rightarrow C}\hspace{0.05cm}\circ\left(\mathcal{M}_{A'\rightarrow A} \otimes \mathrm{id}_{E}\right) \hspace{0.05cm} \circ \mathcal{P}^{\rm{pre}}_{C' \rightarrow A'E}} \nonumber\\
     & =  \underset{\psi_{C'B}}{\mathrm{sup}} \staD*{\left(\mathcal{P}^{\rm{post}}_{AE \rightarrow C}\hspace{0.05cm}\circ\left(\mathcal{N}_{A'\rightarrow A} \otimes \mathrm{id}_{E}\right) \hspace{0.05cm} \circ \mathcal{P}^{\rm{pre}}_{C' \rightarrow A'E} \otimes \mathrm{id}_{B}\right) \psi_{C'B} }{\left(\mathcal{P}^{\rm{post}}_{AE \rightarrow C}\hspace{0.05cm}\circ\left(\mathcal{M}_{A'\rightarrow A} \otimes \mathrm{id}_{E}\right) \hspace{0.05cm} \circ \mathcal{P}^{\rm{pre}}_{C' \rightarrow A'E} \otimes \mathrm{id}_{B}\right) \psi_{C'B}} \nonumber \\
     & \leq \underset{\psi_{A'E B}}{\mathrm{sup}} \staD*{\left(\mathcal{P}^{\rm{post}}_{AE \rightarrow C}\hspace{0.05cm}\circ\left(\mathcal{N}_{A'\rightarrow A} \otimes \mathrm{id}_{E}\right)  \otimes \mathrm{id}_{B}\right) \psi_{A'EB} }{\left(\mathcal{P}^{\rm{post}}_{AE \rightarrow C}\hspace{0.05cm}\circ\left(\mathcal{M}_{A'\rightarrow A} \otimes \mathrm{id}_{E}\right) \otimes \mathrm{id}_{B}\right) \psi_{A'EB}},\nonumber
  \end{align}
  where the inequality in the third line follows from the fact that the set of pure states obtained by varying $\psi_{C'B}$ and applying $\mathcal{P}^{\rm{post}}_{C' \rightarrow A'E}$ is smaller in comparison to the set of all pure states in the system $A'EB$, i.e, $\{\psi_{A'E B}\}$. In other words, the supremum of a set is always greater than or equal to the supremum over its subsets. Let us define $\left(\left(\mathcal{N}_{A'\rightarrow A}\otimes \mathrm{id}_{E}\right)\otimes \mathrm{id}_{B}\right)\psi_{A'EB} := \rho^{\psi}_{ABE}$, which is also a density operator since $\mathcal{N}_{A' \rightarrow A}$ is a quantum channel, and $\left(\left(\mathcal{M}_{A'\rightarrow A}\otimes \mathrm{id}_{E}\right)\otimes \mathrm{id}_{B}\right)\psi_{A'EB} := \sigma^{\psi}_{ABE}$, which is a positive operator, because map $\mathcal{M}$ is CP but not trace preserving. Now, we have 
   \begin{align}
       &  \underset{\psi_{A'E B}}{\mathrm{sup}} \staD*{\left(\mathcal{P}^{\rm{post}}_{AE \rightarrow C}\hspace{0.05cm}\circ\left(\mathcal{N}_{A'\rightarrow A} \otimes \mathrm{id}_{E}\right)  \otimes \mathrm{id}_{B}\right) \psi_{A'EB} }{\left(\mathcal{P}^{\rm{post}}_{AE \rightarrow C}\hspace{0.05cm}\circ\left(\mathcal{M}_{A'\rightarrow A} \otimes \mathrm{id}_{E}\right) \otimes \mathrm{id}_{B}\right) \psi_{A'EB}} \\
       & =\underset{\psi_{A'E B}}{\mathrm{sup}} \staD*{\left( \mathcal{P}^{\rm{post}}_{AE \rightarrow C} \otimes \mathrm{id}_{B}\right) \rho^{\psi}_{ABE}}{\left( \mathcal{P}^{\rm{post}}_{AE \rightarrow C}\otimes \mathrm{id}_{B}\right) \sigma^{\psi}_{AEB}}\nonumber \\
       & \leq \underset{\psi_{A'E B}}{\mathrm{sup}} \staD*{\rho^{\psi}_{ABE}}{\sigma^{\psi}_{AEB}}\nonumber \\
       & = \underset{\psi_{A'E B}}{\mathrm{sup}} \staD*{\left(\left(\mathcal{N}_{A'\rightarrow A}\otimes \mathrm{id}_{E}\right)\otimes \mathrm{id}_{B}\right)\psi_{A'EB}}{\left(\left(\mathcal{M}_{A'\rightarrow A}\otimes \mathrm{id}_{E}\right)\otimes \mathrm{id}_{B}\right)\psi_{A'EB}} = \chD*{\mathcal{N}}{\mathcal{M}},
   \end{align}
where the third line follows from the monotonicity of relative entropy under a quantum channel, and the last line simply follows from the definition of relative entropy between two completely positive maps.  Thus, we have the inequality
\begin{equation}
   \mathbb{D}\left[\Theta\left(\mathcal{N}\right) \Vert \Theta \left(\mathcal{M}\right)\right] \leq \mathbb{D}\left[\mathcal{N}\Vert \mathcal{M}\right].
\end{equation}
\end{proof}
\bibliography{output}{}
\end{document}